\pdfoutput=1
\documentclass[
  12pt,
  secnumarabic,
  notitlepage,
  tightenlines,
  nofootinbib,
  superscriptaddress,
  longbibliography,
  aps,
  pra
]{revtex4-2}

\usepackage{newpxtext,newpxmath}

\usepackage[utf8]{inputenc}
\usepackage{newunicodechar}
\newunicodechar{ρ}{\rho}
\newunicodechar{β}{\beta}
\newunicodechar{σ}{\sigma}
\newunicodechar{λ}{\lambda}
\newunicodechar{Λ}{\Lambda}
\newunicodechar{μ}{\mu}
\newunicodechar{ν}{\nu}
\newunicodechar{ψ}{\psi}
\newunicodechar{ϕ}{\varphi}
\newunicodechar{Φ}{\Phi}
\newunicodechar{φ}{\phi}
\newunicodechar{π}{\pi}
\newunicodechar{α}{\alpha}
\newunicodechar{ϵ}{\epsilon}
\newunicodechar{ε}{\varepsilon}
\newunicodechar{δ}{\delta}
\newunicodechar{ω}{\omega}
\newunicodechar{Δ}{\Delta}
\newunicodechar{Σ}{\Sigma}

\usepackage[thmtools-compat]{keytheorems}
\usepackage{amsthm}

\usepackage{amssymb}
\usepackage{amsmath}
\usepackage{bm,bbm}
\usepackage{braket}
\usepackage{dsfont}
\usepackage{mathdots}
\usepackage{mathtools}
\usepackage{enumerate}
\usepackage[shortlabels]{enumitem}
\usepackage{csquotes}
\usepackage{stmaryrd}
\usepackage{graphicx}
\usepackage{stackengine}
\usepackage{scalerel}
\usepackage{tensor}       
\usepackage[dvipsnames,svgnames]{xcolor}
\usepackage{xfrac}
\usepackage{centernot}
\usepackage{comment}
\usepackage{chngcntr}
\usepackage{booktabs}
\usepackage{tabularray}
\usepackage[normalem]{ulem}
\UseTblrLibrary{booktabs}

\usepackage[pdftex]{hyperref}
\hypersetup{
    bookmarksnumbered=true, 
    breaklinks=true,
    unicode=false, 
    pdfstartview={FitH}, 
    pdfnewwindow=true, 
    colorlinks=true, 
    allcolors=MidnightBlue!70!black!70!TealBlue
}
\usepackage{zref-clever}
\zcsetup{cap}

\newcommand{\widetitle}[1]{%
  \texorpdfstring{%
    \makebox[\linewidth][c]{\makebox[0pt][c]{#1}}%
  }{#1}%
}

\makeatletter
\def\thmhead@plain#1#2#3{%
  \thmname{#1}\thmnumber{\@ifnotempty{#1}{ }\@upn{#2}}%
  \thmnote{ {\the\thm@notefont#3}}}
\let\thmhead\thmhead@plain
\makeatother

\theoremstyle{plain}
\newtheorem{theorem}{Theorem}

\newtheorem*{theorem*}{Theorem}
\newtheorem*{thm*}{Theorem}

\newtheorem{lemma}{Lemma}
\newtheorem*{lemma*}{Lemma}

\newtheorem{proposition}{Proposition}

\newtheorem*{proposition*}{Proposition}
\newtheorem*{prop*}{Proposition}

\newtheorem{corollary}{Corollary}

\newtheorem*{corollary*}{Corollary}
\newtheorem*{cor*}{Corollary}

\newtheorem*{conjecture*}{Conjecture}
\newtheorem*{cj*}{Conjecture}

\theoremstyle{definition}
\newtheorem{definition}{Definition}

\newtheorem*{definition*}{Definition}
\newtheorem*{Def*}{Definition}

\newtheorem*{example*}{Example}
\newtheorem*{ex*}{Example}

\newtheorem*{question*}{Question}

\newtheorem*{problem*}{Problem}

\newtheorem*{axiom*}{Axiom}

\newtheorem*{fact*}{Fact}

\newtheorem*{claim*}{Claim}

\newtheorem*{assumption*}{Assumption}

\newtheorem*{observation*}{Observation}

\theoremstyle{remark}
\newtheorem{remark}{Remark}

\newtheorem*{remark*}{Remark}
\newtheorem*{rem*}{Remark}

\newtheorem*{note*}{Note}

\renewcommand{\thesection}{\arabic{section}}
\renewcommand{\thesubsection}{\thesection.\arabic{subsection}}

\makeatletter

\renewcommand{\section}{%
  \@startsection
    {section}{1}{\z@}%
    {2.5ex \@plus 1ex \@minus .2ex}%
    {1.2ex \@plus .2ex}%
    {\normalfont\large\bfseries\raggedright}%
}

\renewcommand{\subsection}{%
  \@startsection
    {subsection}{2}{\z@}%
    {2ex \@plus 1ex \@minus .2ex}%
    {.8ex \@plus .2ex}%
    {\normalfont\normalsize\bfseries\raggedright}%
}

\renewcommand{\subsubsection}{%
  \@startsection
    {subsubsection}{3}{\z@}%
    {1.5ex \@plus 1ex \@minus .2ex}%
    {.6ex \@plus .2ex}%
    {\normalfont\normalsize\itshape\raggedright}%
}

\def\@hangfrom@section#1#2#3{\@hangfrom{#1#2}#3}
\def\@hangfroms@section#1#2{#1#2}
\def\@hangfrom@subsection#1#2#3{\@hangfrom{#1#2}#3}
\def\@hangfroms@subsection#1#2{#1#2}
\def\@hangfrom@subsubsection#1#2#3{\@hangfrom{#1#2}#3}
\def\@hangfroms@subsubsection#1#2{#1#2}

\makeatother

\theoremstyle{plain}

\newcommand{\bb}{\begin{equation}\begin{aligned}\hspace{0pt}}
\newcommand{\bbb}{\begin{equation*}\begin{aligned}}
\newcommand{\ee}{\end{aligned}\end{equation}}
\newcommand{\eee}{\end{aligned}\end{equation*}}

\renewcommand{\epsilon}{\varepsilon}

\newcommand{\ve}{\varepsilon}

\DeclareMathOperator{\Tr}{Tr}

\DeclareMathAlphabet{\pazocal}{OMS}{zplm}{m}{n}

\DeclareMathOperator{\supp}{supp}

\newcommand{\lsmatrix}{\left(\begin{smallmatrix}}
\newcommand{\rsmatrix}{\end{smallmatrix}\right)}

\stackMath

\stackMath

\makeatletter
\newcommand*\rel@kern[1]{\kern#1\dimexpr\macc@kerna}
\newcommand*\widebar[1]{%
  \begingroup
  \def\mathaccent##1##2{%
    \rel@kern{0.8}%
    \overline{\rel@kern{-0.8}\macc@nucleus\rel@kern{0.2}}%
    \rel@kern{-0.2}%
  }%
  \macc@depth\@ne
  \let\math@bgroup\@empty \let\math@egroup\macc@set@skewchar
  \mathsurround\z@ \frozen@everymath{\mathgroup\macc@group\relax}%
  \macc@set@skewchar\relax
  \let\mathaccentV\macc@nested@a
  \macc@nested@a\relax111{#1}%
  \endgroup
}

\counterwithin*{equation}{part}
\counterwithin*{thm}{part}
\counterwithin*{figure}{part}

\definecolor{Blues5seq1}{RGB}{239,243,255}
\definecolor{Blues5seq2}{RGB}{189,215,231}
\definecolor{Blues5seq3}{RGB}{107,174,214}
\definecolor{Blues5seq4}{RGB}{49,130,189}
\definecolor{Blues5seq5}{RGB}{8,81,156}

\definecolor{Greens5seq1}{RGB}{237,248,233}
\definecolor{Greens5seq2}{RGB}{186,228,179}
\definecolor{Greens5seq3}{RGB}{116,196,118}
\definecolor{Greens5seq4}{RGB}{49,163,84}
\definecolor{Greens5seq5}{RGB}{0,109,44}

\definecolor{Reds5seq1}{RGB}{254,229,217}
\definecolor{Reds5seq2}{RGB}{252,174,145}
\definecolor{Reds5seq3}{RGB}{251,106,74}
\definecolor{Reds5seq4}{RGB}{222,45,38}
\definecolor{Reds5seq5}{RGB}{165,15,21}

\allowdisplaybreaks

\let\nc\newcommand
\let\oldbar\bar
\renewcommand{\bar}{\;\rule{0pt}{9.5pt}\right|\;}
\nc{\lset}{\left\{\left.}
\nc{\rset}{\right\}}
\nc{\lsetr}{\left\{\,}
\nc{\rsetr}{\right.\right\}}
\nc{\barr}{\;\rule{0pt}{9.5pt}\left|\;}

\usepackage[most,breakable]{tcolorbox}
  {\expandafter\ifstrequal\expandafter{#1}{filled}{\begin{tcolorbox}[colback=MidnightBlue!70!black!70!TealBlue!2!white,colframe=MidnightBlue!70!black!70!TealBlue!30!white,breakable,enhanced,left=5.75pt,right=5.75pt,grow sidewards by=10pt]}{\begin{tcolorbox}[colback=white,colframe=gray!15,breakable,enhanced,left=5.75pt,right=5.75pt,grow sidewards by=10pt]}}%
  {\end{tcolorbox}}

\nc{\R}{\mathcal{R}}
\nc{\W}{\mathcal{W}}
\nc{\T}{\mathcal{T}}
\nc{\B}{\mathcal{B}}
\nc{\C}{\mathcal{C}}
\nc{\U}{\mathcal{U}}
\nc{\E}{\mathcal{E}}
\nc{\EE}{\mathscr{E}}
\nc{\K}{\mathcal{K}}
\nc{\V}{\mathcal{V}}
\nc{\X}{\mathcal{X}}
\nc{\F}{\mathcal{F}}
\nc{\G}{\mathcal{G}}
\nc{\D}{\mathcal{D}}
\nc{\Y}{\mathcal{Y}}

\nc{\M}{\mathcal{M}}
\nc{\N}{\mathcal{N}}
\nc{\I}{\mathcal{I}}
\nc{\Q}{\mathbb{Q}}
\nc{\RR}{\mathbb{R}}
\nc{\CC}{\mathbb{C}}

\nc{\HH}{\mathbb{H}}
\nc{\MM}{\mathbb{M}}
\nc{\NN}{\mathbb{N}}
\nc{\DD}{\mathbb{D}}

\nc{\id}{\mathbbm{1}}
\nc{\idc}{\mathrm{id}}

\nc{\norm}[2]{\left\lVert#1\right\rVert_{\,#2}}
\nc{\proj}[1]{\ket{#1}\!\bra{#1}}
\nc{\lnorm}[2]{\left\lVert#1\right\rVert_{\ell_{#2}}}

\let\oldproofname\proofname
\renewcommand{\proofname}{\rm\bf{\oldproofname}}

\makeatletter
\renewenvironment{proof}[1][\proofname]{\par
\pushQED{\qed}%
\normalfont \topsep6\p@\@plus6\p@\relax
\trivlist
\item\relax
{\bfseries  
#1\@addpunct{.}}\hspace\labelsep\ignorespaces 
}{%
\popQED\endtrivlist\@endpefalse
}
\makeatother

\nc{\rhos}{\rho'}

\nc{\gm}{\mathbin\#}

\newcommand{\sPnorm}{P,=}
\newcommand{\sPsub}{P,\leq}
\newcommand{\sTnorm}{T,=}
\newcommand{\sTsub}{T,\leq}
\newcommand{\smoothing}[2]{#2,\,#1}
\NewDocumentCommand{\s}{m m O{\ve}}{%
 \lowercase{\def\sdist{#1}}
 \lowercase{\def\snorm{#2}}
 \IfEqCase{\sdist}{%
  {t}{%
   \IfEqCase{\snorm}{%
    {norm}{ \smoothing{\sTnorm}{#3} }%
    {n}{ \smoothing{\sTnorm}{#3} }%
    {sub}{ \smoothing{\sTsub}{#3}}%
    {s}{ \smoothing{\sTsub}{#3}}%
   }%
  }%
  {p}{%
   \IfEqCase{\snorm}{%
    {norm}{ \smoothing{\sPnorm}{#3} }%
    {n}{ \smoothing{\sPnorm}{#3} }%
    {sub}{ \smoothing{\sPsub}{#3}}%
    {s}{ \smoothing{\sPsub}{#3}}%
   }%
  }%
 }%
}%
\NewDocumentCommand{\Dmax}{m m O{\ve}}{%
 D_{\max}^{\s{#1}{#2}[#3]}
}%

\makeatletter

\def\@sect@ltx#1#2#3#4#5#6[#7]#8{%
    \@ifnum{#2>\c@secnumdepth}{%
        \def\H@svsec{\phantomsection}%
        \let\@svsec\@empty
    }{%
        \H@refstepcounter{#1}%
        \def\H@svsec{%
            \phantomsection
        }%
        \protected@edef\@svsec{{#1}}%
        \@ifundefined{@#1cntformat}{%
            \prepdef\@svsec\@seccntformat
        }{%
            \expandafter\prepdef
            \expandafter\@svsec
            \csname @#1cntformat\endcsname
        }%
    }%
    \@tempskipa #5\relax
    \@ifdim{\@tempskipa>\z@}{%
        \begingroup
        \interlinepenalty \@M
        #6{%
            \@ifundefined{@hangfrom@#1}{\@hang@from}{\csname @hangfrom@#1\endcsname}%
            {\hskip#3\relax\H@svsec}{\@svsec}{#8}%
        }%
        \@@par
        \endgroup
        \@ifundefined{#1mark}{\@gobble}{\csname #1mark\endcsname}{#7}%
        \addcontentsline{toc}{#1}{%
            \@ifnum{#2>\c@secnumdepth}{%
                \protect\numberline{}%
            }{%
                \protect\numberline{\csname the#1\endcsname}%
            }%
            #7}
    }{%
        \def\@svsechd{%
            #6{%
                \@ifundefined{@runin@to@#1}{\@runin@to}{\csname @runin@to@#1\endcsname}%
                {\hskip#3\relax\H@svsec}{\@svsec}{#8}%
            }%
            \@ifundefined{#1mark}{\@gobble}{\csname #1mark\endcsname}{#7}%
            \addcontentsline{toc}{#1}{%
                \@ifnum{#2>\c@secnumdepth}{%
                    \protect\numberline{}%
                }{%
                    \protect\numberline{\csname the#1\endcsname}%
                }%
                #8}%
        }%
    }%
    \@xsect{#5}}%

\makeatother

\newcommand{\suppresscontentsline}[3]{}

\providecommand{\Tr}{\Tr}
\providecommand{\supp}{\operatorname{supp}}

\providecommand{\id}{\mathrm{id}}

\newcommand{\vertiii}[1]{{\left\vert\kern-0.25ex\left\vert\kern-0.25ex\left\vert #1 
    \right\vert\kern-0.25ex\right\vert\kern-0.25ex\right\vert}}

\makeatletter
\newcommand{\hideappendixsubsections}{%
    \let\l@subsection\@gobbletwo
}
\makeatother

\begin{document}
\count\footins=1000

\title{\widetitle{Quantum de Finetti theorems for states and channels in any distance measure}}
\author{Liuhang Ye}
\email{ly400@cam.ac.uk}
\affiliation{Department of Applied Mathematics and Theoretical Physics, Centre for Mathematical Sciences, University of Cambridge, Cambridge CB3 0WA, United Kingdom}

\author{Bjarne Bergh}
\email{bb536@cam.ac.uk}
\affiliation{Department of Applied Mathematics and Theoretical Physics, Centre for Mathematical Sciences, University of Cambridge, Cambridge CB3 0WA, United Kingdom}

\author{Nilanjana Datta}
\email{n.datta@damtp.cam.ac.uk}
\affiliation{Department of Applied Mathematics and Theoretical Physics, Centre for Mathematical Sciences, University of Cambridge, Cambridge CB3 0WA, United Kingdom}
\begin{abstract}
\vspace{0.8em}
{\centering\bfseries Abstract\par}
\vspace{0.4em}
\noindent

Standard finite quantum de Finetti theorems approximate the $k$-system marginals of permutation-invariant states of $n$-systems by mixtures of independent and identically distributed (iid) states, usually in trace distance. We prove both standard and
Renner's exponential de Finetti theorems in the stronger form of operator inequalities, implying bounds in every Schatten norm and for every quantum R\'enyi divergence satisfying data processing. In the standard case, for fixed local dimension, our $k/n$ error bound in max-relative entropy improves on the previously best known $k^2/n$ scaling, even in the classical setting. The operator-inequality approach is particularly suited to study channel de Finetti representations because operator order between Choi states is equivalent to completely positive
(CP) order between the underlying channels. For permutation-covariant channels $\mathcal N^{(n)}:A^{\otimes n}\to B^{\otimes n}$, where
$d_A=\dim A$, we prove that the $k$-system reduced channel is CP-dominated by a mixture of tensor-power channels with error $O(k/\sqrt{n})$ and polynomial dependence on the local dimensions, addressing a question raised by Berta et al.~\cite{berta2022semidefinite}. Under the no-signalling condition, we also prove an exponential channel de Finetti theorem where the approximating mixture consists of Choi-almost-iid channels, whose normalized Choi states are almost-iid in the sense of Mazzola--Sutter--Renner. In the case of $r$ defects, the representation error is at most $\mathrm{poly}(n)\bigl(2d_A^4k^3/(nr^2)\bigr)^{(r+1)/2}$ and decays exponentially in $n$ for a suitable choice of parameters $r$ and $k$.

\end{abstract}
\makeatletter
\@booleanfalse\titlepage@sw
\makeatother
\maketitle

\makeatletter
\begingroup
    \setcounter{tocdepth}{2}
    \let\l@subsubsection\@gobbletwo
    \tableofcontents
\endgroup
\makeatother
\section{Introduction}
\label{sec:introduction}

De Finetti theorems are fundamental in probability~\cite{definetti1931funzione,definetti1937prevision, hewitt1955symmetric, diaconis1977finite} and quantum information theory~\cite{stormer1969symmetric,hudson1976locally,caves2002unknown,konig2005finetti,christandl2007one}: they connect permutation symmetry to mixtures of independent and identically distributed (iid) systems. A standard quantum de Finetti theorem\footnote{We use \emph{standard} to distinguish it from the \emph{exponential} de Finetti theorem discussed below.} relates the marginal of a permutation-invariant state $\rho^{(n)} \in \mathcal{D}(\mathcal{H}^{\otimes n})$ to a mixture of iid states. The most commonly used quantum de Finetti theorem is from \cite{christandl2007one} and states that for all $0\leq k \leq n$ and $\dim \mathcal{H}=d < \infty$, there exists a probability measure $\mu$ on density matrices $\mathcal{D}(\mathcal{H})$ such that:
\begin{equation}
    \norm{\Tr_{n - k}\rho^{(n)} - \int \sigma^{\otimes k} \,\mathrm d \mu(\sigma)}{1} \leq \frac{2 d^2 k}{n}\,.
\end{equation}

The representation error can be exponentially improved by relaxing the mixture of iid states to a mixture of almost-iid states, allowing a controlled number of defective positions that may occur in coherent superposition (this is Renner's exponential de Finetti theorem \cite{renner2005security, renner2007symmetry}): for all $0\leq r\leq k< n$, there exists a measure $\mu$ on \(\mathcal D(\mathcal H)\), and states \(\omega_{\sigma}^{(k,r)}\in\mathcal D(\mathcal H^{\otimes k})\) that are \(\binom{k}{r}\)-almost-iid along \(\sigma\in\mathcal D(\mathcal H)\) (see \zcref{def:almost-iid-state}), such that:
\begin{equation}
\label{eq:Renner-exp}
    \norm{\Tr_{n - k}\rho^{(n)} - \int \omega_{\sigma}^{(k,r)}\,\mathrm d \mu(\sigma)}{1} \leq 3 (n-k)^{d^2} \exp\left(-{(n-k) (r + 1) \over n }\right).
\end{equation}

With a suitable trade-off between the number \(r\) of defects and the number \(n-k\) of discarded systems, one can keep nearly all systems \(k\sim n-o(n)\), while achieving an exponentially decaying error. Such a representation of reduced states on nearly the whole system has proved particularly useful in finite-key quantum cryptography~\cite{renner2005security}.

{Classical finite de Finetti theorems have been studied in distance measures beyond total variation distance~\cite{johnson2025relative,gavalakis2024finite}. Hence, it is natural to ask whether quantum de Finetti representations can also be formulated in distance measures other than the trace distance, including the max-relative entropy.} A more comprehensive overview of previous results on classical and quantum de Finetti theorems is provided in \zcref{sec:background}.

Another natural question is whether analogous de Finetti theorems hold for quantum channels. For channels, the relevant symmetry is permutation covariance, namely invariance under simultaneously permuting the corresponding input and output systems. Unlike for states, however, defining a reduced channel requires some care: tracing out a subset of output systems does not in general define a channel acting only on the retained inputs, since the remaining outputs may still depend on inputs at the discarded positions. One must therefore either specify those inputs or impose an appropriate no-signalling condition ensuring that the reduced output is independent of them (see \zcref{sec:channel-df-setup}). Moreover, a channel de Finetti representation does not follow simply from its state counterpart. Indeed, although one may apply a state de Finetti theorem to the Choi state of a permutation-covariant channel, the resulting states in the convex mixture need not individually satisfy the Choi marginal constraint and therefore need not correspond to quantum channels being both completely positive (CP) and trace preserving (TP).

The only known channel de Finetti representation for permutation-covariant and no-signalling channels on finitely many systems, due to Berta et al.\ \cite{berta2022semidefinite}, was obtained by adding linear Choi marginal constraints to state de Finetti representations and then converting trace distance to diamond-norm distance via Choi–Jamiołkowski isomorphism~\cite{choi1975completely, jamiolkowski1972linear}. They proved that the reduced channel is close to a convex mixture of tensor-power (iid) channels. However, the trace distance to diamond-norm distance conversion inevitably introduces a factor \(d^k\) exponential in the number \(k\) of retained systems. Consequently, their representation does not yield a vanishing error when \(k\) grows polynomially with the total number \(n\) of systems. This is in sharp contrast with standard de Finetti theorems for states, which remain useful for all \(k\) sublinear in \(n\). More ambitiously, a channel de Finetti theorem with exponentially vanishing error that retains nearly all subsystems lies well beyond the existing literature.

\subsection{Main results}
Here we summarize the main results of our paper. The notations and definitions used are given in~\zcref{sec:math-prelim}. For detailed formulations and proofs of these results, see~\zcref{sec:exp-df-states} and~\zcref{sec:channel-df}.
\smallskip

\noindent
\paragraph*{\normalfont\bfseries
1.\ De Finetti theorems as operator inequalities
\normalfont\mdseries
(see \zcref{sec:std-exp-state-df})}

Our first main result strengthens state de Finetti representations to the level of operator order. Standard and exponential de Finetti theorems are conventionally formulated as statements in trace distance. We instead establish de Finetti representations as operator inequalities.

\begin{theorem}[(Quantum de Finetti representation for states)]
\label{thm:de Finetti mixed states}
Fix \(0\leq r\leq k\leq n\), and let \(\rho^{(n)}\in\mathcal D(\mathcal H^{\otimes n})\) be permutation invariant with \(\dim\mathcal H=d\). Then there exist a probability measure \(\mu_{nkr}\) on \(\mathcal D(\mathcal H)\), and states \(\omega_{\sigma}^{(k,r)}\in\mathcal D(\mathcal H^{\otimes k})\) that are \(\binom{k}{r}\)-almost-iid along \(\sigma\in\mathcal D(\mathcal H)\) (see \zcref{def:almost-iid-state}), such that
\begin{equation}
    \bigl(1-\varepsilon_{nkr}\bigr)
    \Tr_{n-k}\rho^{(n)}
    \leq
    \int
        \omega_{\sigma}^{(k,r)}
        \,\mathrm d\mu_{nkr}(\sigma).
    \label{eq:schematic-state-df}
\end{equation}
Moreover, \(\varepsilon_{nkr}\) can be upper bounded by
\begin{equation}
    \varepsilon_{nkr}
    \leq
    (n-k+1)^{d^2-1}
    \exp\!\left(
        -\frac{(n-k)(r+1)}{n}
    \right).
\end{equation}
In the case of \(r=0\), 
\begin{equation}
    \bigl(1-\varepsilon_{nk0}\bigr)
    \Tr_{n-k}\rho^{(n)}
    \leq
    \int
        \sigma^{\otimes k}
        \,\mathrm d\mu_{nk0}(\sigma),\qquad
    \varepsilon_{nk0}
    \leq
    \frac{d^2k}{n}.
\end{equation}
In this case, if \(\rho^{(n)}\) is classical (i.e.\ diagonal in a product basis) or pure, then \(d^2\) can be replaced by \(d\), and the measure can be restricted to classical or pure states, respectively.
\end{theorem}

\zcref{thm:de Finetti mixed states} first of all gives an operator inequality analogue of Renner's exponential de Finetti theorem. Setting \(r=0\) reduces the representation to a mixture of iid states and recovers a standard de Finetti representation, again as an operator inequality. Such operator inequalities imply the closeness of the marginal to the mixture of (almost-)iid states in practically any distance measure. More precisely, for any distance measure $d(\rho, \sigma)$ for which a function $f(\lambda)$ with $f(\lambda) \to 0$ as $\lambda \to 1$ exists, such that $d(\rho, \sigma) \leq f(λ)$ whenever $\sigma \geq λ ρ$, the above operator inequality directly yields an associated de Finetti theorem. These conditions hold in particular for the following choices of $d(\rho, \sigma)$: (i) for any norm $\vertiii{\rho-\sigma}$ and (ii) any of the typically used generalized divergences $\mathbb{D}(ρ\|σ)$. In particular the conditions hold for any quantum R\'enyi divergence as axiomatically characterized in \cite{tomamichel_quantum_2016} and any smoothed version thereof (see \zcref{cor:trace-distance-df-thm}).

In addition to providing a unified proof of the standard and exponential state de Finetti theorems, our result for \(r=0\) significantly improves upon previous classical and quantum de Finetti bounds in max-divergence~\cite{chiribella2011quantum,song2025coded,gavalakis2024finite}: our error scales as \(k/n\). While a \(k/n\) dependence was previously known for trace distance~\cite{christandl2007one} in the quantum case, and for total variation (TV) distance~\cite{diaconis1980finite} and relative entropy~\cite{johnson2025relative,gavalakis2024finite} in the classical case, the best previously known operator inequality (or equivalently max-divergence) bounds scaled as \(k^2/n\), even in the classical setting~\cite{song2025coded,gavalakis2024finite}; see \zcref{sec:background} for details. Since the \(k/n\) dependence is tight classically for TV distance \cite{diaconis1980finite} so is our theorem, and it allows to obtain tight classical and quantum de Finetti theorems in essentially any distance measure.

\begin{remark}
The operator inequality we show for the the standard ($r=0$) de Finetti theorem is essentially already hidden in the proof of \cite[Theorem II.2]{christandl2007one} where it can be obtained from the last unnumbered equation on page 4 after some further optimizations, although this is not mentioned in any way. The proof we give works simultaneously for the standard ($r=0$) and the exponential ($r>0$) case.
\end{remark}

\begin{remark}
Note that our operator-inequality de Finetti representation is different from what are commonly known as \emph{de Finetti reductions} in the literature, although the latter also take the form of operator inequalities. In its standard form~\cite{christandl2009postselection,fawzi2015quantum,nahar2024postselection}, de Finetti reduction states that every permutation-invariant state \(\rho^{(n)}\in\mathcal D(\mathcal H^{\otimes n})\), with \(\dim\mathcal H=d\), satisfies
\begin{equation}
    \rho^{(n)}
    \leq
    \binom{n+d^2-1}{n}
    \int \sigma^{\otimes n}\,\mathrm d\sigma,
\end{equation}
where \(\mathrm d\sigma\) is the induced by Haar measure on pure states in \(\mathcal H\otimes\mathcal H\) after tracing out one copy. It is obtained by bounding the symmetric purification of \(\rho^{(n)}\) by the projector onto \(\operatorname{Sym}^n(\mathcal H\otimes\mathcal H)\) and tracing out the purifying system. This bounds the full \(n\)-copy state with a prefactor growing as \((n+1)^{d^2-1}\), and the bound is universal: it does not depend on the particular permutation-invariant state on the left-hand side. In contrast, our theorem bounds a \(k\)-copy marginal of \(\rho^{(n)}\) with a prefactor tending to one asymptotically, and the de Finetti state on the right-hand side depends on the left-hand side through the probability measure.
\end{remark} 
\smallskip

\noindent
\paragraph*{\normalfont\bfseries
2.\ Improved de Finetti theorem for channels
\normalfont\mdseries
(see \zcref{sec:std-channel-df})}

Besides being of independent interest and potentially useful in other quantum-information-theoretic settings, the operator inequality formulation is also particularly well-suited to studying channel de Finetti theorems. Under the Choi–Jamiołkowski isomorphism \cite{choi1975completely,jamiolkowski1972linear}, operator order between Choi states is equivalent to completely-positive (CP) order between the corresponding maps. This allows us to formulate and prove channel de Finetti representations directly in CP order, rather than first obtaining trace-distance bounds for Choi states and subsequently converting them into diamond-norm bounds for channels, hence avoiding the undesirable \(d^k\) factor in the channel de Finetti theorem of \cite{berta2022semidefinite} (see also the discussion below \zcref{eq:BBFS}).

\begin{theorem}[(Standard de Finetti representation for channels)]
\label{thm:standard-channel-definetti}
Fix \(0\leq k\leq n\), and let \(\mathcal N^{(n)}\in\mathrm{CPTP}(A^n,B^n)\) be permutation-covariant (see \zcref{eq:permutation-covariance}). Let \(\Tr_{n-k}\mathcal N^{(n)}=\mathcal N^{(k)}\in\mathrm{CPTP}(A^k,B^k)\) be the reduced channel defined by \zcref{eq:reduced-channel}. Then there exists a probability measure \(\nu_{nk}\) on \(\mathrm{CPTP}(A,B)\) such that
\begin{equation}
    \bigl(1-\epsilon_{nk}\bigr)
    \Tr_{n-k}\mathcal N^{(n)}
    \leq_{\mathrm{CP}}
    \int
        \Phi^{\otimes k}
        \,\mathrm d\nu_{nk}(\Phi).
    \label{eq:standard-channel-definetti}
\end{equation}
Moreover, \(\epsilon_{nk}\) can be upper bounded by
\begin{equation}
    \epsilon_{nk}
    \leq
    \frac{d_A^2d_B^2k}{n}
    +2d_Ak\sqrt{\frac{d_A^2d_B^2-1}{n+d_A^2d_B^2-1}}
    +\frac{d_A^2k^2}{n+d_A^2d_B^2-1}.
    \label{eq:standard-channel-definetti-error}
\end{equation}
\end{theorem}

\zcref{thm:standard-channel-definetti} is a channel de Finetti representation, with an error scaling as \(k/\sqrt{n}\). The CP-order bound directly implies a corresponding channel de Finetti representation in diamond-norm distance (see \zcref{cor:standard-channel-definetti-diamond}) or any other channel distance that satisfies the conditions $d(\N, \M) \leq f(λ)$ whenever $\M \geq λ\N$ and $f(λ) \to 0$ as $λ \to 1$. In particular this is the case when the distance arises from optimizing a state distance over input states, i.e.\ $d(\N, \M) = \sup_{\nu} d_s((\mathcal{R} \otimes \N)(ν), (\mathcal{R} \otimes \M)(ν))$, where $\mathcal{R}$ is some arbitrary fixed channel and the state distance $d_s$ satisfies the equivalent conditions from above. 

{We also point out that the proof of this result does not rely on the no-signalling condition in \zcref{Eq:no-signalling-last-position}, whereas the proof of the next (exponential) theorem does. Nevertheless, when this condition holds, the reduced channel defined by \zcref{eq:reduced-channel} has a more natural operational interpretation; see the discussions in \zcref{sec:channel-df-setup}.}
\smallskip

\noindent
\paragraph*{\normalfont\bfseries
3.\ Exponential de Finetti theorem for channels
\normalfont\mdseries
(see \zcref{sec:exp-channel-df})}

As in the state setting, one cannot in general obtain an exponentially accurate de Finetti representation by mixtures of tensor-power (iid) channels, which motivates an appropriate notion of almost-iid channels. We call a channel almost-iid along a reference channel if its Choi state is almost-iid along the Choi state of that reference channel; see \zcref{def:almost-iid-channel}. The structural and operational properties of almost-iid channels are developed further in our companion paper~\cite{companion-paper}. Under permutation-covariance and no-signalling conditions, we show that a reduced channel can be approximated by a mixture of almost-iid channels, with an error that decays exponentially under a suitable trade-off between the number of discarded systems and the number of allowed defects.

\begin{theorem}[(Exponential de Finetti representation for channels)]
\label{thm:exp-channel-df}
Fix \(1\leq r\leq k\leq n\), and let \(\mathcal N^{(n)}\in\mathrm{CPTP}(A^n,B^n)\) be both permutation-covariant (see \zcref{eq:permutation-covariance}) and no-signalling (see \zcref{Eq:no-signalling-last-position}). Let \(\Tr_{n-k}\mathcal N^{(n)}=\mathcal N^{(k)}\in\mathrm{CPTP}(A^k,B^k)\) be the reduced channel defined by \zcref{eq:reduced-channel}. Then there exist a probability measure
\(\nu_{nkr}\) on \(\mathrm{CPTP}(A,B)\), and channels \(\mathcal{E}_{\Phi}^{(k,r)}\in\mathrm{CPTP}(A^k,B^k)\) that are \(\binom{k}{r}\)-almost-iid  along \(\Phi\in\mathrm{CPTP}(A,B)\) (see \zcref{def:almost-iid-channel}), such that
\begin{equation}
    \bigl(1-\epsilon_{nkr}\bigr)\Tr_{n-k}\mathcal{N}^{(n)}
    \leq_{\mathrm{CP}}
    \int\mathcal{E}_{\Phi}^{(k,r)}\,\mathrm d\nu_{nkr}(\Phi).
    \label{eq:exp-channel-df}
\end{equation}
Moreover, \(\epsilon_{nkr}\) can be upper bounded by
\begin{equation}
\epsilon_{nkr}
    \leq
    4(n+1)^{3d_A^2d_B^2/2 - 1}
    \exp\!\left[
        -\frac{r+1}{2}
        \ln\!\left(
            \frac{nr^2}{2d_A^4k^3}
        \right)
    \right].
\label{eq:exp-channel-df-error}
\end{equation}
\end{theorem}

The CP-order bound again implies a corresponding channel de Finetti representation in diamond-norm distance (see \zcref{cor:exponential-channel-definetti-diamond}). The representation error decays exponentially with suitable choice of \(n,k,r\), providing a channel analogue of the exponential de Finetti representation for states. A limitation of the present bound is that exponential decay is guaranteed only in regimes where (still) a large fraction of the systems is discarded; see \zcref{rmk:decay-regime} for details. This restriction originates from the additional cost incurred when enforcing the Choi marginal constraint on the individual states in the convex mixture. Consequently, our result does not yet attain the characteristic regime of Renner's exponential de Finetti theorem for states, in which nearly all systems can be retained while the error remains exponentially small. Extending the channel representation to this regime remains an important direction for future work.

\section{Background and previous work}
\label{sec:background}
\subsection{De Finetti theorems for states}
\subsubsection{Classical de Finetti theorems.}
A sequence of random variables $X_1,\dots,X_n$ is \emph{exchangeable} if its joint distribution is invariant under every permutation of the indices. Exchangeability is a much weaker requirement than independence, yet de Finetti's representation theorem~\cite{definetti1931funzione,definetti1937prevision} shows that the two notions are intimately related: an \emph{infinite} exchangeable sequence of binary random variables is distributed as a mixture of independent and identically distributed (iid) sequences, i.e.\ there is a unique probability measure $\mu$ on $[0,1]$ such that $\Pr[X_1^k = x_1^k] = \int Q_p^{\otimes k}(x_1^k)\,d\mu(p)$ for every $k$. Hewitt and Savage~\cite{hewitt1955symmetric} extended this to exchangeable sequences with values in general (compact Hausdorff) spaces, identifying the iid measures as the extreme points of the convex set of exchangeable measures. The theorem is a cornerstone of Bayesian statistics, where it justifies the use of a prior over an unknown parameter as a consequence of exchange symmetry rather than as an additional assumption, and it plays an equally important structural role in probability and information theory: it explains why symmetric objects --- symmetrised codebooks, random permutation ensembles, sampling without replacement, and permutation-invariant test statistics --- behave, to leading order, as if they were iid, which is the regime in which the classical asymptotic tools (the law of large numbers, the asymptotic equipartition property, Sanov's theorem, and the method of types) apply.
 
For \emph{finite} exchangeable sequences the exact representation fails: Diaconis~\cite{diaconis1977finite} gave the simple example of a pair $(X_1,X_2)$ that is uniformly distributed on $\{(0,1),(1,0)\}$, which is exchangeable but cannot be written as a mixture of product distributions. What survives is an approximate statement, usually called a \emph{finite de Finetti theorem}: if $X_1^n$ is exchangeable and $k\ll n$, then the law $P_k$ of the first $k$ variables is close to a mixture $M_{k,\mu_n}=\int Q^{\otimes k}\,d\mu_n(Q)$ of iid laws. Diaconis and Freedman~\cite{diaconis1980finite} obtained the sharp rates in total variation (TV) distance,
\begin{equation}
\label{eq:DF}
\|P_k - M_{k,\mu_n}\|_{\mathrm{TV}} \le \frac{k(k-1)}{2n}
\qquad\text{and}\qquad
\|P_k - M_{k,\mu_n}\|_{\mathrm{TV}} \le \frac{2ck}{n},
\end{equation}
the first holding for exchangeable vectors with values in an \emph{arbitrary} measurable space and the second for a finite alphabet of size $c$; both rates are tight. The natural mixing measure $\mu_n$ in these results is the law of the empirical measure of $X_1^n$, and the proof reduces the problem to the classical comparison between sampling $k$ balls from an urn with and without replacement, i.e.\ between multinomial and hypergeometric laws~\cite{diaconis1980finite,freedman1977remark,diaconis1987dozen}.
 
\subsubsection{Classical de Finetti theorems beyond total variation.}
Total variation is, however, not the only---and often not the most useful---way to measure closeness. In information theory one is typically interested in relative entropy (Kullback--Leibler divergence), in R\'enyi divergences, or in the max-divergence $D_{\max}$, which control TV distance through Pinsker's inequality but not conversely. Finite de Finetti theorems in these stronger measures are more recent. Stam~\cite{stam1978distance} bounded the relative entropy between the hypergeometric and multinomial laws by $(c-1)k(k-1)/\big(2(n-1)(n-k+1)\big)$, a bound later refined by Harremo\"es and Mat\'u\v{s}~\cite{harremoes2020bounds}; combined with the sampling representation of $P_k$ and $M_{k,\mu_n}$ and the joint convexity of relative entropy, this yields a finite de Finetti theorem in relative entropy of optimal order $O(k^2/n^2)$ on finite alphabets~\cite{johnson2025relative,gavalakis2024finite}. Independently, Gavalakis and Kontoyiannis~\cite{gavalakis2021information,gavalakis2023information} gave purely information-theoretic proofs of finite de Finetti theorems in relative entropy (first for binary, then for finite alphabets), and Berta, Gavalakis and Kontoyiannis~\cite{berta2024third} obtained a bound of the form $\frac{k(k-1)}{2(n-k-1)}\,H(X_1)$ valid on arbitrary discrete alphabets, using an argument imported from the quantum information literature. Most relevant for the present work is a family of results that are stronger still: for any exchangeable $X_1^n$ one has the pointwise \emph{domination} inequality
\begin{equation}
\label{eq:classical-domination}
P_k \;\le\; \frac{n^k\,(n-k)!}{n!}\; M_{k,\mu_n},
\end{equation}
which goes back to an observation of Freedman~\cite{freedman1977remark} and was recently turned into a finite de Finetti theorem by Song, Attiah and Yu~\cite{song2025coded} and by Gavalakis, Johnson and Kontoyiannis~\cite{gavalakis2024finite} (who extended it to arbitrary measurable spaces). Inequality~\zcref{eq:classical-domination} is a bound on the max-divergence, $D_\infty(P_k\|M_{k,\mu_n})\le \log\frac{n^k(n-k)!}{n!}\le -\log\big(1-\tfrac{k(k-1)}{2n}\big)$; since $D_{\max}$ dominates every R\'enyi divergence and every $f$-divergence, a single inequality of this type implies a finite de Finetti theorem in \emph{all} of these measures simultaneously, with a bound of order $k^2/n$ that is independent of the alphabet size. In this sense, the domination inequality is the ``master'' form of the classical finite de Finetti theorem.
 
\subsubsection{Quantum de Finetti theorems.}
The quantum analogue of an exchangeable sequence is a \emph{permutation-invariant} (or symmetric) state $\rho^{(n)}$ on $\mathcal{H}^{\otimes n}$, $\mathcal{H}\simeq\mathbb{C}^d$, i.e.\ a state satisfying $U^\pi\rho^{(n)}(U^\pi)^{\dagger}=\rho^{(n)}$ for all permutations $\pi\in S_n$; the analogue of an iid law is a product state $\sigma^{\otimes k}$, and the analogue of a mixture of iid laws is a \emph{de Finetti state} $\int\sigma^{\otimes k}\,d\mu(\sigma)$, where $\mu$ is a probability measure on the state space (for instance the unitarily invariant measure on pure states~\cite{zyczkowski2001induced,bengtsson2007geometry}). The infinite quantum de Finetti theorem---every state on an infinite tensor product that is invariant under finite permutations is a mixture of product states---was established by St{\o}rmer~\cite{stormer1969symmetric} and by Hudson and Moody~\cite{hudson1976locally}, and further developed in the setting of mean-field quantum statistical mechanics by Fannes, Lewis and Verbeure~\cite{fannes1988symmetric} and Raggio and Werner~\cite{raggio1989quantum}. Caves, Fuchs and Schack~\cite{caves2002unknown} gave an elementary proof for finite-dimensional systems together with a Bayesian interpretation of the ``unknown quantum state'', and Fuchs, Schack and Scudo~\cite{fuchs2004finetti} extended the representation to quantum processes.
 
The finite quantum de Finetti theorem is due to K\"onig and Renner~\cite{konig2005finetti} and, in its sharpest and now standard form, to Christandl, K\"onig, Mitchison and Renner~\cite{christandl2007one}: for every permutation-invariant state $\rho^{(n)}$ on $(\mathbb{C}^d)^{\otimes n}$ there is a probability measure $\mu$ on the set of states such that
\begin{equation}
\label{eq:CKMR}
\Big\|\Tr_{n-k}\rho^{(n)} - \int \sigma^{\otimes k}\,d\mu(\sigma)\Big\|_1 \le \frac{2d^2 k}{n},
\end{equation}
and the constant improves to $2dk/n$ (with a measure supported on pure states) when $\rho^{(n)}$ is supported on the symmetric subspace. In contrast with the classical bound~\zcref{eq:DF}, the dependence on the local dimension in~\zcref{eq:CKMR} cannot be removed in general in trace distance~\cite{christandl2007one}, which is one of the reasons why quantum de Finetti theorems are so much more delicate than their classical counterparts. Alternative proofs and variants were given by K\"onig and Mitchison~\cite{konig2009most}, by Christandl and Toner~\cite{christandl2009finite} for conditional probability distributions, and by Arnon-Friedman and Renner~\cite{arnonfriedman2015finetti} in the form of ``de Finetti reductions''; see~\cite{harrow2013church} for a survey of the symmetric-subspace techniques underlying most of these results. Of particular importance for us is the work of Chiribella~\cite{chiribella2011quantum}, who exhibited an exact relation between the finite de Finetti theorem, optimal state estimation and optimal universal cloning: the partial trace of a state supported on the symmetric subspace is expressed as a combination of the measure-and-prepare (estimate-and-reprepare) channel and cloning channels acting on smaller numbers of subsystems. A direct consequence of Chiribella's formula is an \emph{operator} inequality of the form
\begin{equation}
\label{eq:chiribella}
\Tr_{n-k}\rho^{(n)} \;\le\; c_{n,k,d}\int |\phi\rangle\langle\phi|^{\otimes k}\,d\mu_{\rho^{(n)}}(\phi),
\qquad c_{n,k,d}=1+O\!\Big(\frac{k(k+d)}{n}\Big),
\end{equation}
where $\mu_{\rho^{(n)}}$ is the (explicit) measure obtained by optimal estimation on $\rho^{(n)}$. Inequality~\zcref{eq:chiribella} is the quantum analogue of the classical domination inequality~\zcref{eq:classical-domination}: it is a bound on the quantum max-relative entropy $D_{\max}(\Tr_{n-k}\rho^{(n)}\,\|\,\int\phi^{\otimes k}d\mu_{\rho^{(n)}})$ by $\log c_{n,k,d}$ and therefore implies a de Finetti theorem not only in trace distance but in every quantum R\'enyi divergence and, more generally, in any distinguishability measure that is monotone under operator domination; see also~\cite{berta2024third}, where a classical version of this argument is used.
 
\subsubsection{Exponential de Finetti theorems and post-selection.}
A qualitatively different type of statement is Renner's \emph{exponential} de Finetti theorem~\cite{renner2005security,renner2007symmetry}, see \zcref{eq:Renner-exp}. There, one gives up on approximating the reduced state by a mixture of iid states, and instead shows that every permutation-invariant state on $n$ systems has a purification which, after discarding a small number $n-k$ of subsystems, is close to a mixture of \emph{almost-iid}\ states (see \zcref{def:almost-iid-state}). The approximation error is now exponentially small in $n$, instead of decaying only polynomially in $k/n$, and the statement applies to the marginal on almost all $n$ subsystems rather than a marginal on $k\ll n$ of them. This is precisely what is needed for finite-key security proofs of quantum key distribution (QKD) against general (coherent) attacks~\cite{renner2005security}. The price paid is that the ``almost-iid'' states are no longer exactly tensor power, so the ensuing information-theoretic analysis has to be carried out for this larger class of states; a systematic treatment of ``almost-iid'' information theory has recently been given by Mazzola, Sutter and Renner~\cite{mazzola2026almost} and subsequently extended to broader notions in~\cite{girardi2026new,girardi2026quantum,datta2026entropy,datta2026robustness,girardi2026generalised}. The exponential de Finetti theorem was extended to infinite-dimensional systems by Renner and Cirac~\cite{renner2009finetti} and, in the Gaussian setting, by Leverrier~\cite{leverrier2018su}.

\subsubsection{Applications.}
Quantum de Finetti theorems have become an indispensable tool throughout quantum information theory. Besides their original motivation in the security analysis of QKD, where they reduce general attacks to collective (iid) attacks~\cite{renner2005security,renner2007symmetry,christandl2009postselection,nahar2024postselection}, they underpin the Bayesian and the reliable approaches to quantum state and process tomography~\cite{caves2002unknown,fuchs2004finetti,christandl2012reliable}; they explain the convergence of the Doherty--Parrilo--Spedalieri hierarchy of semidefinite programs for entanglement detection based on $k$-extendibility~\cite{doherty2004complete,navascues2009power}, and more generally of semidefinite programming hierarchies for polynomial and bilinear optimisation problems~\cite{doherty2012convergence,berta2022semidefinite}; in complexity theory, de Finetti theorems under local measurements give the best known algorithms for the class $\mathrm{QMA}(2)$ and for related optimisation problems~\cite{brandao2017quantum}; and in mathematical physics they are the mechanism behind the validity of mean-field (Hartree) theory for large bosonic systems and of Bose--Einstein condensation~\cite{raggio1989quantum,lewin2014derivation,rougerie2015finetti}. In all of these applications, the strength of the conclusion is dictated by the distinguishability measure in which the de Finetti approximation is available --- trace distance for hypothesis-testing type tasks, relative entropy and R\'enyi divergences for rates, error exponents and one-shot quantities.

\subsection{De Finetti theorems for channels}

\label{sec:channels-intro}
\subsubsection{Classical channels and conditional probability distributions.}
In classical information theory, a channel is a conditional probability distribution $P_{A^n|X^n}$, often called a \emph{box} in the nonlocality literature. The first de Finetti theorem for such objects is due to Christandl and Toner~\cite{christandl2009finite}, who proved a finite de Finetti theorem for conditional probability distributions that are symmetric under simultaneous permutations of inputs and outputs and satisfy a no-signalling constraint: the marginal on $k$ boxes is close in total variation to a mixture $\int Q_{A|X}^{\otimes k}\,d\mu(Q)$ of iid conditional distributions, with an explicit error vanishing as $k/n\to0$. An infinite version for general probabilistic theories has been established independently by Barrett and Leifer~\cite{barrett2009finetti}.
 
\subsubsection{Quantum channels: exchangeable sequences.}
The infinite quantum de Finetti theorem for channels is due to Fuchs, Schack and Scudo~\cite{fuchs2004finetti} (see also~\cite{fuchs2004unknown}), who introduced the notion of an \emph{exchangeable sequence of channels} $\{\mathcal{N}_{B^n\to \oldbar{B}^n}\}_n$: each $\mathcal{N}_{B^n\to\oldbar{B}^n}$ is permutation-covariant, and the sequence is consistent under discarding, $\Tr_{\oldbar{B}_n}\circ\mathcal{N}_{B^n\to\oldbar{B}^n}=\mathcal{N}_{B^{n-1}\to\oldbar{B}^{n-1}}\circ\Tr_{B_n}$, so that in particular no input can signal to the outputs of other positions. They proved, by mapping the problem to an exchangeable sequence of Choi states subject to the trace-preservation constraint, that such a sequence has a unique representation $\mathcal{N}_{B^k\to\oldbar{B}^k}=\int\mathcal{M}^{\otimes k}\,d\mu(\mathcal{M})$ as a mixture of iid channels, thereby giving a Bayesian justification of quantum-process tomography that parallels the justification of state tomography in~\cite{caves2002unknown}. The corresponding \emph{finite} statement was obtained only much later by Berta, Borderi, Fawzi and Scholz~\cite{berta2022semidefinite}, as a special case of a de Finetti theorem for symmetric states subject to linear constraints. In the channel language, if $\mathcal{N}_{AB^n\to\oldbar{A}\oldbar{B}^n}$ is covariant under permutations of the $n$ $B$-systems and satisfies the two consistency conditions
\begin{equation}
\label{eq:BBFS-conditions}
\Tr_{\oldbar{B}_n}\!\big[\mathcal{N}(\cdot)\big]=\Tr_{\oldbar{B}_n}\!\Big[\mathcal{N}\big(\Tr_{B_n}[\cdot]\otimes\tfrac{\mathbbm{1}_{B_n}}{d_B}\big)\Big],\qquad
\Tr_{\oldbar{A}}\!\big[\mathcal{N}(\cdot)\big]=\Tr_{\oldbar{A}}\!\Big[\mathcal{N}\big(\tfrac{\mathbbm{1}_A}{d_A}\otimes\Tr_{A}[\cdot]\big)\Big],
\end{equation}
then the reduced channel $\mathcal{N}^{(k)}$ on $k$ of the $B$-systems is close to a mixture $\sum_ip_i\,\mathcal{E}^i_{A\to\oldbar{A}}\otimes(\mathcal{D}^i_{B\to\oldbar{B}})^{\otimes k}$ of product channels. The bound is proved at the level of Choi states, where it reads
\begin{equation}
\label{eq:BBFS}
\Big\|\rho_{A\oldbar{A}(B\oldbar{B})^k}-\sum_ip_i\,\sigma^i_{A\oldbar{A}}\otimes(\omega^i_{B\oldbar{B}})^{\otimes k}\Big\|_1
\;\le\;k\,f(B\oldbar{B}|\cdot)\sqrt{(2\ln2)\,\frac{\log(d_Ad_{\oldbar{A}})+(k-1)\log(d_Bd_{\oldbar{B}})}{n-k+1}},
\end{equation}
with $\sigma^i,\omega^i$ Choi states of genuine channels and $f(B\oldbar{B}|\cdot)\le\sqrt{18(d_Bd_{\oldbar{B}})^3}$ a ``distortion'' constant of an informationally complete measurement; the conversion to the diamond norm costs a further factor $d_Ad_B^k$, and the authors leave as an open problem a diamond-norm bound with polynomial dependence on $d_B$ and $k$. The proof is information-theoretic, based on the chain rule for the conditional mutual information in the manner of Brand\~ao and Harrow~\cite{brandao2017quantum}, and it is the finite counterpart of the Fuchs--Schack--Scudo theorem when $A\oldbar{A}$ is trivial. A variant in which one of the two constraints in~\zcref{eq:BBFS-conditions} is dropped yields closeness to one-way LOCC channels instead, and the $k=1$ conditions are the channel analogue of $n$-extendibility: they define \emph{$n$-extendible channels}~\cite{kaur2019extendibility,pankowski2013entanglement}, which converge to entanglement-breaking (measure-and-prepare) channels as $n\to\infty$ and which underlie semidefinite-programming hierarchies for approximate quantum error correction~\cite{berta2022semidefinite} and for the resource theory of unextendibility~\cite{kaur2019extendibility}.

\section{Mathematical preliminaries}
\label{sec:math-prelim}

Let $\mathcal{H}$ denote a finite-dimensional Hilbert space, and $\mathcal{L}(\mathcal{H})$ the set of linear operators acting on $\mathcal{H}$. Let $\mathcal{D}(\mathcal{H})$ denote the set of density matrices on $\mathcal{H}$, i.e., the set of positive semidefinite operators of unit trace. We label different quantum systems by capital Roman letters ($A,B,C$, etc.) and often use these letters interchangeably with the corresponding Hilbert space (i.e., we write $A$ instead of $\mathcal{H}_A$). We also concatenate these letters to denote tensor products of systems, i.e., we write $AR$ for $\mathcal{H}_A \otimes \mathcal{H}_R$. The identity operator in $\mathcal{L}(\mathcal{H})$ is denoted by $\id_{\mathcal{H}}$, whereas the identity map acting on $\mathcal{L}(\mathcal{H})$ is denoted by $\idc_{\mathcal{H}}$. The maximally mixed state on $A$ is denoted by $\tau_A$. We use the shorthand notation \(\phi=|\phi\rangle\!\langle\phi|\) for pure states. For operators
$X,Y\in\mathcal L(A)$, the Hilbert--Schmidt inner product
is defined by
$\langle X,Y\rangle_{\mathrm{HS}}
\coloneqq \Tr(X^\dagger Y)$. The set of unitaries on \(\mathcal{H}\) is denoted by \(\mathsf U(\mathcal{H})\). A quantum channel (usually denoted by $\mathcal{N},\mathcal{E},\Phi$, etc.) is a completely positive trace-preserving map acting on $\mathcal{L}(\mathcal{H})$. The set of quantum channels from $\mathcal{L}(A)$ to $\mathcal{L}(B)$ is denoted by $\mathrm{CPTP}(A,B)$. The normalized Choi state of a channel $\Phi_{A\to B}\in\mathrm{CPTP}(A,B)$ is defined by
\begin{equation}
    J(\Phi_{A\to B})
    \coloneqq
    (\Phi_{A\to B}\otimes\idc_{A'})
    (|\Omega\rangle\langle\Omega|_{AA'}), \qquad
    |\Omega\rangle_{AA'}
    \coloneqq
    \frac{1}{\sqrt{d_A}}
    \sum_{i=1}^{d_A}|i\rangle_A|i\rangle_{A'},
\end{equation}
where $A'\simeq A$, and $\{|i\rangle_A\}_{i=1}^{d_A}$ is a fixed
orthonormal basis of $A$, and
$\{|i\rangle_{A'}\}_{i=1}^{d_A}$ is the corresponding basis of
$A'$. With this normalization, $J(\Phi_{A\to B})$ is a
state on \(A'B\) satisfying $\Tr_BJ(\Phi_{A\to B})=\tau_{A'}$. Conversely, a state $\omega_{A'B}$ is the normalized Choi state of a unique channel from
$A$ to $B$ if and only if $\omega_{A'}=\tau_{A'}$.

A state $\rho_{AB}\in\mathcal D(AB)$ is called an extension of
$\rho_A\in\mathcal D(A)$ if $\Tr_B\rho_{AB}=\rho_A$. Analogously, a
channel $\mathcal N_{A\to BE}$ is called an extension of
$\mathcal N_{A\to B}$ if
$\Tr_E\circ\mathcal N_{A\to BE}=\mathcal N_{A\to B}$. This is
equivalent to the corresponding state-extension property for their
Choi states:
$\Tr_EJ(\mathcal N_{A\to BE})=J(\mathcal N_{A\to B})$.

For Hermitian operators \(X,Y\), we write \(X\leq Y\) if \(Y-X\geq0\), i.e., if \(Y-X\) is positive semidefinite. Analogously, for Hermiticity-preserving linear superoperators \(\mathcal N,\mathcal M:\mathcal L(A)\to\mathcal L(B)\), we write
\(\mathcal N\leq_{\mathrm{CP}}\mathcal M\) if \(\mathcal M-\mathcal N\) is completely positive. Under the Choi isomorphism:
\begin{equation}
    \mathcal N\leq_{\mathrm{CP}}\mathcal M
    \quad\Longleftrightarrow\quad
    J({\mathcal N})\leq J({\mathcal M}).
\end{equation}

For integers \(k\geq 1\), let \(S_k\) denote the symmetric group of degree \(k\); for every \(\pi\in S_k\), we denote by
\(U_{\mathcal H^{\otimes k}}^\pi\) the unitary on \(\mathcal H^{\otimes k}\)
that permutes the \(k\) tensor factors according to \(\pi\). We use \([k]\) to denote the set \(\{1,\ldots,k\}\). For a \(D\)-dimensional Hilbert space \(\mathcal H\) and an integer
\(m\geq 1\), the symmetric subspace is
\begin{equation}
    \operatorname{Sym}^{m}(\mathcal H)
    \coloneq
    \left\{
        \lvert\psi\rangle\in\mathcal H^{\otimes m}:
        U_{\mathcal H^{\otimes m}}^\pi\lvert\psi\rangle
        =
        \lvert\psi\rangle
        \ \text{for all }\pi\in S_m
    \right\}.
\end{equation}
We denote its orthogonal projection and dimension by
\begin{equation}
    P_{\mathrm{sym}}^{(m)}
    =
    \frac{1}{m!}\sum_{\pi\in S_m}
        U_{\mathcal H^{\otimes m}}^\pi,
    \qquad
    d_m
    \coloneq
    \dim\operatorname{Sym}^{m}(\mathcal H)
    =
    \binom{m+D-1}{D-1}
    \leq
    (m+1)^{D-1}.
\end{equation}
We denote by \(\mathrm d\phi\) the normalized Haar measure on the set of pure states in \(\mathcal H\), then \(P_{\mathrm{sym}}^{(m)}\) admits the integral representation \cite{watrous2018theory}:
\begin{equation}
    P_{\mathrm{sym}}^{(m)}=d_m\int
        \lvert\phi\rangle\!\langle\phi\rvert^{\otimes m}
        \,\mathrm d\phi.
    \label{eq:haar-moment-symmetric-subspace}
\end{equation}

\subsection{Almost-iid states and channels}

Three inequivalent notions of almost-iid states have been introduced. In increasing order of generality, these are Mazzola-Sutter-Renner (MSR), Wasserstein, and weakly almost-iid states, and they form a strict hierarchy~\cite{girardi2026new}. In this work, we consider only the MSR notion and henceforth refer to MSR almost-iid states simply as almost-iid states.

For a pure state \(|\phi\rangle\in\mathcal H\) and integers $k\geq1$ and \(0\leq r\leq k\), let
\begin{equation}
    \mathcal V\bigl(\mathcal H^{\otimes k},|\phi\rangle^{\otimes k-r}\bigr)\coloneq \left\{U_{\mathcal H^{\otimes k}}^{\pi}\left(|\phi\rangle^{\otimes k-r}\otimes|\omega^{(r)}\rangle\right):\pi\in S_k,\ |\omega^{(r)}\rangle\in\mathcal H^{\otimes r}\right\}.
\end{equation}

\begin{definition}[(MSR-Almost-iid state~\cite{mazzola2026almost})]
\label{def:almost-iid-state}
Let $\sigma_A\in\mathcal D(A)$, let $k\geq1$, and $0\leq r\leq k$ be integers. A state
$\rho_{A^k}\in\mathcal D(A^k)$ is called an
$\binom{k}{r}$-almost-iid state along the reference state $\sigma_A$ if there exist a
purification $|\phi\rangle_{AE}$ of $\sigma_A$ and an extension
$\rho_{A^kE^k}$ of $\rho_{A^k}$, such that
\begin{equation}
\begin{aligned}
\textnormal{(1)}\quad&
    U^\pi_{(AE)^k}\rho_{A^kE^k}
    \bigl(U^\pi_{(AE)^k}\bigr)^\dagger
    =
    \rho_{A^kE^k},
    \quad\forall\,\pi\in S_k;                                         \\
\textnormal{(2)}\quad&
    \supp(\rho_{A^kE^k})
    \subseteq
    \operatorname{span}
    \mathcal V\!\left(
        \mathcal H_{AE}^{\otimes k},
        |\phi\rangle_{AE}^{\otimes(k-r)}
    \right).
\end{aligned}
\end{equation}
We sometimes use the notation \(\rho_{\sigma}^{(k,r)}\) to indicate that \(\rho\) is $\binom{k}{r}$-almost-iid along $\sigma$.
\end{definition}

When the reference state is pure, the above general definition simplifies as follows: a state \(\rho_{A^k}\in\mathcal D(A^k)\) is \(\binom{k}{r}\)-almost iid along \(\lvert\phi\rangle_A\) if it is permutation invariant and
\begin{equation}
    \supp\!\left(\rho_{A^k}\right)
    \subseteq
    \operatorname{span}
    \mathcal V\!\left(
        \mathcal H_A^{\otimes k},
        \lvert\phi\rangle_A^{\otimes(k-r)}
    \right).
\end{equation}
Indeed, every purification of the pure reference is of the product form \(\lvert\phi\rangle_A\otimes\lvert\eta\rangle_E\). Hence, tracing out \(E^k\) gives the simplified permutation invariance and support condition. Conversely, any state satisfying these conditions admits the trivial extension \(\rho_{A^k}\otimes\lvert\eta\rangle\!\langle\eta\rvert_E^{\otimes k}\).

The projection onto the almost-iid subspace can be constructed explicitly. For \(0\leq j\leq k\) and \(|\phi\rangle\in\mathcal H\), define
\begin{equation}
\label{eq:def-Pi-j}
    \Pi_{\phi,j}^{(k)}\coloneq \sum_{\substack{S\subseteq[k]\\ |S|=j}}|\phi\rangle\!\langle\phi|^{\otimes S^{\mathrm c}}\otimes\bigl(\id-|\phi\rangle\!\langle\phi|\bigr)^{\otimes S}, \qquad \sum_{j=0}^{k}\Pi_{\phi,j}^{(k)}=\id,
\end{equation}
where the tensor factors are placed in the positions indexed by \(S\) and \(S^{\mathrm c}=[k]\setminus S\). Thus \(\Pi_{\phi,j}^{(k)}\) projects onto the subspace with exactly \(j\) tensor factors orthogonal to \(|\phi\rangle\) and remaining factors in
\(\operatorname{span}\{\lvert\phi\rangle\}\). The projection onto \(\operatorname{span}\mathcal V(\mathcal H^{\otimes k},|\phi\rangle^{\otimes k-r})\) is denoted as
\begin{equation}
    P_{\phi}^{(k,r)}\coloneq \sum_{j=0}^{r}\Pi_{\phi,j}^{(k)},\qquad \operatorname{ran}P_{\phi}^{(k,r)}=\operatorname{span}\mathcal V\bigl(\mathcal H^{\otimes k},|\phi\rangle^{\otimes(k-r)}\bigr).
    \label{eq:def-defect-projector-range}
\end{equation}
To see why the range of \(P_{\phi}^{(k,r)}\) coincides with \(\operatorname{span}\mathcal V\bigl(\mathcal H^{\otimes k},|\phi\rangle^{\otimes(k-r)}\bigr)\), decompose
\(\mathcal H=\operatorname{span}\{\lvert\phi\rangle\}\oplus
\{\lvert\phi\rangle\}^{\perp}\) with \(\{\lvert\phi\rangle\}^{\perp}\) being the orthogonal complement. Expanding the \(r\) defect factors of any vector in
\(\mathcal V(\mathcal H^{\otimes k},
\lvert\phi\rangle^{\otimes(k-r)})\) with respect to this decomposition produces only terms with at most \(r\) factors orthogonal to \(\lvert\phi\rangle\), which lie in the range of
\(P_{\phi}^{(k,r)}\). Conversely, every vector with
exactly \(j\leq r\) factors orthogonal to \(\lvert\phi\rangle\) and remaining factors in
\(\operatorname{span}\{\lvert\phi\rangle\}\) contains at least \(k-r\) copies of
\(\lvert\phi\rangle\), and hence belongs \(\operatorname{span}\mathcal V\bigl(\mathcal H^{\otimes k},|\phi\rangle^{\otimes(k-r)}\bigr)\).

The notion of almost-iid states can be extended to define almost-iid quantum channels via the Choi–Jamiołkowski isomorphism:

\begin{definition}[(Almost-iid channel \cite{companion-paper})]
\label{def:almost-iid-channel}
Let $\Phi\in\mathrm{CPTP}(A,B)$, let $k\geq1$, and $0\leq r\leq k$ be integers. A channel
$\mathcal N^{(k)}\in\mathrm{CPTP}(A^k,B^k)$ is called an
$\binom{k}{r}$-almost-iid channel along the reference channel $\Phi$ if its Choi state $J(\mathcal N^{(k)})$ is an
$\binom{k}{r}$-almost-iid state along the Choi state of the reference channel $J(\Phi)$. Here $J(\mathcal N^{(k)})$, initially written as a state on $A'^kB^k$, is regarded as a state on $(A'B)^{k}$ by grouping $A_i'$ with $B_i$ for each $i\in[k]$.
\end{definition}

We mention that a different notion of almost-iid processes, based on the quantum Wasserstein distance of order $1$ \cite{de2021quantum}, was proposed in \cite[Definition 12]{girardi2026quantum}.

\section{Quantum state de Finetti theorems as operator inequalities}
\label{sec:exp-df-states}

\subsection{Setup for state de Finetti theorems}
\label{sec:state-df-setup}
The finite quantum de Finetti representation for states concerns the reduced state \(\rho^{(k)}\) of a permutation-invariant state
\(\rho^{(n)}\in\mathcal D(A^n)\), with
\(\dim A=d\). By the standard result that permutation-invariant states admit symmetric purifications \cite{renner2005security}, it suffices to consider symmetric pure states on the purified Hilbert space \(AE\) with dimension \(D=d^2\), and establish de Finetti representations on the enlarged space. The corresponding statements for permutation-invariant (mixed) states on
the original \(d\)-dimensional system \(A\) then follow by tracing out the purifying system \(E\), which preserves operator inequalities. Moreover, iid pure states
\(\phi_{AE}^{\otimes k}\) reduce to iid mixed states \(\phi_A^{\otimes k}\coloneq (\Tr_E\phi_{AE})^{\otimes k}\), while almost-iid states along \(\phi_{AE}\) reduce to \cite[Remark 2.4 (e)]{mazzola2026almost} almost-iid states
along \(\phi_A\).

Accordingly, for integers \(0\leq k\leq n\), let \(|\Psi^{(n)}\rangle\in\operatorname{Sym}^{n}(\mathcal H)\) be a symmetric pure state, with \(\dim \mathcal H=D\), and denote its reduced state on arbitrary $k$ copies by \(\Psi^{(k)}\coloneq \Tr_{n-k}\Psi^{(n)}\). We will establish de Finetti representations for \(\Psi^{(k)}\).

For notational convenience, for every pure state \(\lvert\phi\rangle\in\mathcal H\), define the
(unnormalized) postselected operator
\begin{equation}
\label{eq:abbreviation-contraction}
    \Psi_{\phi}^{(k)}\coloneq \left(\id\otimes\langle\phi|^{\otimes(n-k)}\right)\Psi^{(n)}\left(\id\otimes|\phi\rangle^{\otimes(n-k)}\right)\in\mathcal L(\mathcal H^{\otimes k}).
\end{equation}
The vector
\(
(\id\otimes\langle\phi\rvert^{\otimes(n-k)})
\lvert\Psi^{(n)}\rangle
\)
is symmetric, since permutations of the remaining \(k\) systems commute with the contraction and leave \(\lvert\Psi^{(n)}\rangle\) invariant.
Hence $\supp\Psi_{\phi}^{(k)}\subseteq\operatorname{Sym}^{k}(\mathcal H)$.

\subsection{Standard and exponential de Finetti theorems for states}
\label{sec:std-exp-state-df}
Proofs of trace-distance de Finetti theorems typically begin with a symmetric pure state \(\Psi^{(n)}\). Projecting \(n-k\) of its subsystems onto \(\lvert\phi\rangle^{\otimes(n-k)}\) leaves an unnormalized state \(\Psi_{\phi}^{(k)}\) on the remaining \(k\) subsystems. This state is then projected onto the almost-iid subspace spanned by vectors in which at least \(k-r\) tensor factors equal \(\lvert\phi\rangle\), while the remaining factors may be arbitrary. The resulting trace-distance error between the reduced state
\begin{equation}
    \Psi^{(k)} \quad \text{and} \quad \int
P_{\phi}^{(k,r)}\Psi_{\phi}^{(k)}P_{\phi}^{(k,r)}\,\mathrm d\phi
\end{equation}
can be controlled by the Haar-averaged mass \(\int\Tr[(\id-P_{\phi}^{(k,r)})\Psi_{\phi}^{(k)}]\,\mathrm d\phi\) outside the almost-iid subspace via the gentle measurement lemma \cite{winter1999coding}. This Haar-averaged mass can be estimated to be exponentially small. See~\cite{vidick2016simple} for a simple proof following this strategy.

To derive an operator inequality, we refine this argument by computing exactly the contribution of each defect sector before taking traces. The resulting operator identities provide us with much richer information than the trace relations, and allow us to treat the standard and exponential de Finetti theorems within a unified argument.

\begin{lemma}\label{lem:haar-defect-coefficient}
For every \(0\leq j\leq k\leq n\), with \(\Pi_{\phi,j}^{(k)}\) defined in \zcref{eq:def-Pi-j},
\begin{equation}
\begin{aligned}
 &P_{\mathrm{sym}}^{(n)}\left(\int\Pi_{\phi,j}^{(k)}\otimes|\phi\rangle\!\langle\phi|^{\otimes(n-k)}\,\mathrm d\phi\right)
 =\left(\int\Pi_{\phi,j}^{(k)}\otimes|\phi\rangle\!\langle\phi|^{\otimes(n-k)}\,\mathrm d\phi\right)P_{\mathrm{sym}}^{(n)}\\
 &\hspace{8em}=\frac{1}{d_n}\frac{\binom{k}{j}}{\binom{n}{j}}\binom{D+j-2}{j}P_{\mathrm{sym}}^{(n)}.
\end{aligned}
\label{eq:haar-defect-coefficient}
\end{equation}
\end{lemma}

\begin{proof}
For \(S\subseteq[k]\) with \(|S|=j\), expanding the tensor power \(\bigl(\id-|\phi\rangle\!\langle\phi|\bigr)^{\otimes S}\) gives
\begin{equation}
 |\phi\rangle\!\langle\phi|^{\otimes S^{\mathrm c}}\otimes\bigl(\id-|\phi\rangle\!\langle\phi|\bigr)^{\otimes S}\otimes|\phi\rangle\!\langle\phi|^{\otimes(n-k)}=\sum_{T\subseteq S}(-1)^{|T|}|\phi\rangle\!\langle\phi|^{\otimes(n-j+|T|)}\otimes\id^{\otimes(j-|T|)},
\end{equation}
where it is understood that the operators on the right-hand side remain in their original positions. For every \(0\leq\ell\leq n\), we have
\begin{equation}
    P_{\mathrm{sym}}^{(n)}\bigl(P_{\mathrm{sym}}^{(n-\ell)}\otimes\id^{\otimes\ell}\bigr)=P_{\mathrm{sym}}^{(n)},
\end{equation}
because every globally symmetric vector is symmetric under local permutations of any chosen \(n-\ell\) factors. Counting the subsets \(S\subseteq[k]\) and \(T\subseteq S\), we therefore obtain
\begin{equation}
\begin{aligned}
P_{\mathrm{sym}}^{(n)}
 \left(
   \int \Pi_{\phi;j}^{(k)}
   \otimes |\phi\rangle\!\langle\phi|^{\otimes(n-k)}
   \,\mathrm d\phi
 \right)&=
 \sum_{\substack{S\subseteq[k]\\ |S|=j}}
 \sum_{T\subseteq S}(-1)^{|T|}
 P_{\mathrm{sym}}^{(n)}
 \left(
   \int
   |\phi\rangle\!\langle\phi|^{\otimes(n-j+|T|)}\, \mathrm d\phi
   \otimes\id^{\otimes(j-|T|)}
 \right)                                                        \\
&=
 \sum_{\substack{S\subseteq[k]\\ |S|=j}}
 \sum_{T\subseteq S}
 \frac{(-1)^{|T|}}{d_{n-j+|T|}}
 P_{\mathrm{sym}}^{(n)}\bigl(P_{\mathrm{sym}}^{(n-j+|T|)}\otimes\id^{\otimes (j-|T|)}\bigr)                                        \\
&=
 \binom{k}{j}
 \sum_{t=0}^{j}(-1)^t\binom{j}{t}
 \frac{1}{d_{n-j+t}}
 P_{\mathrm{sym}}^{(n)},
\end{aligned}
\end{equation}
where the second equality follows from \zcref{eq:haar-moment-symmetric-subspace}.
It remains to evaluate the coefficient. Since for any positive integer $m$, by the integral representation of beta function,
\begin{equation}
    \frac{1}{d_m}=\frac{m!(D-1)!}{(m+D-1)!}=(D-1)\int_0^1x^m(1-x)^{D-2}\,\mathrm dx,
\end{equation}
the binomial theorem gives
\begin{equation}
\begin{aligned}
 \sum_{t=0}^{j}(-1)^t\binom{j}{t}\frac{1}{d_{n-j+t}}
 &= (D-1)\int_0^1x^{n-j}(1-x)^{D-2}\sum_{t=0}^{j}\binom{j}{t}(-x)^t\,\mathrm dx\\
 &= (D-1)\int_0^1x^{n-j}(1-x)^{D+j-2}\,\mathrm dx
 =\frac{1}{d_n}\frac{1}{\binom{n}{j}}\binom{D+j-2}{j}.
\end{aligned}
\end{equation}
This proves one of the equalities in \zcref{eq:haar-defect-coefficient}; the other one follows by taking the Hermitian conjugate.
\end{proof}

For notational convenience, for \(0\leq j\leq k\), define
\begin{equation}
    p_j\coloneq \frac{d_{n-k}}{d_n}\frac{\binom{k}{j}}{\binom{n}{j}}\binom{D+j-2}{j}.
    \label{eq:def-pj}
\end{equation}
These coefficients are chosen such that \(\{p_j\}_{j=0,\cdots,k}\) is a probability distribution. This can be seen by a direct calculation, using the integral representation of beta function and the binomial theorem,
\begin{equation}
\begin{aligned}
 \sum_{j=0}^{k}p_j&=\frac{(n-k+D-1)!}{(D-2)!(n-k)!}
  \sum_{j=0}^{k}\binom{k}{j}
  \frac{(D+j-2)!(n-j)!}{(n+D-1)!} \\
&=\frac{(n-k+D-1)!}{(D-2)!(n-k)!}
  \sum_{j=0}^{k}\binom{k}{j}
  \int_0^1 x^{D+j-2}(1-x)^{n-j}\,\mathrm dx \\
 &=\frac{(n-k+D-1)!}{(D-2)!(n-k)!}\int_0^1x^{D-2}(1-x)^{n-k}\sum_{j=0}^{k}\binom{k}{j}x^j(1-x)^{k-j}\,\mathrm dx\\
 &=\frac{(n-k+D-1)!}{(D-2)!(n-k)!}\int_0^1x^{D-2}(1-x)^{n-k}\,\mathrm dx=1.
\end{aligned}
\end{equation}

A direct consequence of \zcref{lem:haar-defect-coefficient} are the following operator identities. The operator relations, rather than merely scalar relations of the traces, are the key additional
information needed for the subsequent operator inequality argument.

\begin{lemma}\label{lem:contracted-defect-coefficient}
For every \(0\leq j\leq k\leq n\), with \(\{p_j\}_{j=0,\cdots,k}\) defined in \zcref{eq:def-pj},
\begin{align}
    d_{n-k}\int\Pi_{\phi,j}^{(k)}\Psi_{\phi}^{(k)}\,\mathrm d\phi=
 d_{n-k}\int&\Psi_{\phi}^{(k)}\Pi_{\phi,j}^{(k)}\,\mathrm d\phi=p_j\Psi^{(k)},
 \label{eq:lem2-first}
 \\
 d_{n-k}\int\Psi_{\phi}^{(k)}\,\mathrm d\phi&=\Psi^{(k)}.
 \label{eq:lem2-second}
\end{align}
\end{lemma}

\begin{proof}
\zcref{eq:lem2-second} follows from \zcref{eq:lem2-first}, by summing both sides from $j=0$ to $k$, and using the fact that \(\sum_{j=0}^{k}\Pi_{\phi,j}^{(k)}=\id\) and \(\sum_{j=0}^{k}p_j=1\). We now prove \zcref{eq:lem2-first}.

Recall that \(\Psi_{\phi}^{(k)}=\left(\id\otimes\langle\phi|^{\otimes(n-k)}\right)\Psi^{(n)}\left(\id\otimes|\phi\rangle^{\otimes(n-k)}\right)\), we have
\begin{equation}
\begin{aligned}
    \int\Pi_{\phi,j}^{(k)}\Psi_{\phi}^{(k)}\,\mathrm d\phi&=\Tr_{n-k}\!\left[\left(\int\Pi_{\phi,j}^{(k)}\otimes|\phi\rangle\!\langle\phi|^{\otimes(n-k)}\,\mathrm d\phi\right)\Psi^{(n)}\right]\\
    &=\Tr_{n-k}\!\left[\left(\int\Pi_{\phi,j}^{(k)}\otimes|\phi\rangle\!\langle\phi|^{\otimes(n-k)}\,\mathrm d\phi\right)P_{\mathrm{sym}}^{(n)}\Psi^{(n)}\right]\\
    &=\Tr_{n-k}\!\left[\frac{1}{d_n}\frac{\binom{k}{j}}{\binom{n}{j}}\binom{D+j-2}{j}P_{\mathrm{sym}}^{(n)}\Psi^{(n)}\right]
    =\frac{p_j}{d_{n-k}}\Psi^{(k)},
\end{aligned}
\end{equation}
where the second and fourth equality follow from \(\Psi^{(n)}=P_{\mathrm{sym}}^{(n)}\Psi^{(n)}\), the third equality follows from \zcref{lem:haar-defect-coefficient}. The expression with the order between $\Pi_{\phi,j}^{(k)}$ and $\Psi_{\phi}^{(k)}$ reversed follows by taking Hermitian conjugate.
\end{proof}

We can now prove the pure state de Finetti representation theorem.

\begin{theorem}
\label{thm:exp-df-states-operator}
Fix integers \(0\leq r\leq k\leq n\), and let \(|\Psi^{(n)}\rangle\in\operatorname{Sym}^{n}(\mathcal H)\) be a symmetric pure state, with ${\rm{dim}}\, \mathcal H = D$. With \(\{p_j\}_{j=0,\cdots,k}\) defined in \zcref{eq:def-pj}, let \(\epsilon_{nkr}\coloneq \sum_{j=r+1}^{k}p_j\), with the convention $\epsilon_{nkk}=0$. Then there exist a probability measure \(\mu_{nkr}\) on the set of pure states in $\mathcal H$, and states \(\omega_{\phi}^{(k,r)}\) that are \(\binom{k}{r}\)-almost-iid along \(|\phi\rangle \in \mathcal H\), such that
\begin{equation}
    \bigl(1-\epsilon_{nkr}\bigr)\Psi^{(k)}\leq\int\omega_{\phi}^{(k,r)}\,\mathrm d\mu_{nkr}(\phi).
\label{eq:state-dF-op-ineq}
\end{equation}
Moreover, \(\epsilon_{nkr}\) can be upper bounded by
\begin{equation}
   \epsilon_{nkr} \leq (n-k+1)^{D-1} \exp\!\left(-\frac{(n-k)(r+1)}{n}\right).
\end{equation}
\end{theorem}

\begin{proof}
Performing the sum $\sum_{j=0}^r$ on both sides of \zcref{eq:lem2-first} in \zcref{lem:contracted-defect-coefficient}, and using \zcref{eq:def-defect-projector-range}, we obtain
\begin{equation}
 d_{n-k}\int P_{\phi}^{(k,r)}\Psi_{\phi}^{(k)}\,\mathrm d\phi
 =d_{n-k}\int\Psi_{\phi}^{(k)}P_{\phi}^{(k,r)}\,\mathrm d\phi=\bigl(1-\epsilon_{nkr}\bigr)\Psi^{(k)},
 \label{eq:state-cross-term}
\end{equation}
while \zcref{eq:lem2-second} of \zcref{lem:contracted-defect-coefficient} states that 
\begin{equation}
    d_{n-k}\int\Psi_{\phi}^{(k)}\,\mathrm d\phi=\Psi^{(k)}.
\label{eq:state-no-projection-term}
\end{equation}

Consider the following operator identity, where $\lambda\in\mathbb{R}$ is a constant to be optimized over later:
\begin{equation}
\bigl(P_{\phi}^{(k,r)}-\lambda\id\bigr)
 \Psi_{\phi}^{(k)}
 \bigl(P_{\phi}^{(k,r)}-\lambda\id\bigr)=
 P_{\phi}^{(k,r)}
 \Psi_{\phi}^{(k)}
 P_{\phi}^{(k,r)}
 -\lambda P_{\phi}^{(k,r)}\Psi_{\phi}^{(k)}
 -\lambda\Psi_{\phi}^{(k)}P_{\phi}^{(k,r)}
 +\lambda^2\Psi_{\phi}^{(k)}.
\end{equation}
Positivity of the left-hand side, together with \zcref{eq:state-cross-term} and \zcref{eq:state-no-projection-term}, yield (after integrating over $\phi$ with respect to the Haar measure \(\mathrm d\phi\), and multiplying with $d_{n-k}$) 
the operator inequality
\begin{equation}
 0\leq d_{n-k}\int P_{\phi}^{(k,r)}\Psi_{\phi}^{(k)}P_{\phi}^{(k,r)}\,\mathrm d\phi-2\lambda\bigl(1-\epsilon_{nkr}\bigr)\Psi^{(k)}+\lambda^2\Psi^{(k)}.
\end{equation}
To make the inequality as tight as possible, we maximize $2\lambda\bigl(1-\epsilon_{nkr}\bigr)-\lambda^2$, giving the optimal value $\lambda=\bigl(1-\epsilon_{nkr}\bigr)$. This gives
\begin{equation}
\label{eq:intermediate-op-ineq}
 \bigl(1-\epsilon_{nkr}\bigr)^2\Psi^{(k)}\leq d_{n-k}\int P_{\phi}^{(k,r)}\Psi_{\phi}^{(k)}P_{\phi}^{(k,r)}\,\mathrm d\phi.
\end{equation}

Define the normalized states and the probability measure\footnote{If \(\Tr\left(P_{\phi}^{(k,r)}\Psi_{\phi}^{(k)}\right)=0\), simply remove these \(\phi\)'s from the integral as they do not contribute. Also we may assume \(\epsilon_{nkr}<1\), otherwise the theorem follows trivially.}
\begin{equation}
\label{eq:state-prob-measure}
    \omega_{\phi}^{(k,r)}\coloneq \frac{P_{\phi}^{(k,r)}\Psi_{\phi}^{(k)}P_{\phi}^{(k,r)}}{\Tr\left(P_{\phi}^{(k,r)}\Psi_{\phi}^{(k)}\right)},\qquad \mathrm d\mu_{nkr}(\phi)\coloneq \frac{d_{n-k}}{1-\epsilon_{nkr}}\Tr\left(P_{\phi}^{(k,r)}\Psi_{\phi}^{(k)}\right)\,\mathrm d\phi.
\end{equation}
The fact that \(\mu_{nkr}\) is normalized is again the consequence of \zcref{eq:state-cross-term}. Moreover, \(\omega_{\phi}^{(k,r)}\) is by construction permutation invariant and supported on \(\operatorname{span}\mathcal V(\mathcal H^{\otimes k},|\phi\rangle^{\otimes k-r})\), and is therefore \(\binom{k}{r}\)-almost-iid along \(|\phi\rangle\). Finally, by \zcref{eq:intermediate-op-ineq},
\begin{equation}
    \int\omega_{\phi}^{(k,r)}\,\mathrm d\mu_{nkr}(\phi)=\frac{d_{n-k}}{1-\epsilon_{nkr}}\int P_{\phi}^{(k,r)}\Psi_{\phi}^{(k)}P_{\phi}^{(k,r)}\,\mathrm d\phi\geq\bigl(1-\epsilon_{nkr}\bigr)\Psi^{(k)},
\end{equation}
which proves \zcref{eq:state-dF-op-ineq} in the claim. The bound on \(\epsilon_{nkr}\) can be obtained as follows,
\begin{equation}
    \epsilon_{nkr}
=\frac{d_{n-k}}{d_n}
  \sum_{j=r+1}^{k}
  \frac{\binom{k}{j}}{\binom{n}{j}}
  \binom{D+j-2}{j}
\leq \frac{d_{n-k}}{d_n}
  \left(\frac{k}{n}\right)^{r+1}
  \sum_{j=0}^{k}\binom{D+j-2}{j}.
\end{equation}
The sum can be evaluated explicitly by using the Hockey-stick identity\footnote{For \(m\le n\), \(\sum_{i=m}^{n} \binom{i}{m} = \binom{n+1}{m+1}\).} to be $d_k=\binom{D+k-1}{D-1}$. The \((k/n)^{r+1}\) term supplies the exponential decay by observing \(\left(1-x\right)^{a}\le e^{-ax}\) for \(0\le x\le1, a>0\). Together,
\begin{equation}
    \epsilon_{nkr}\le \frac{d_{n-k}d_k}{d_n}\left(\frac{k}{n}\right)^{r+1}
\leq
d_{n-k}\left(1-\frac{n-k}{n}\right)^{r+1}
\leq
  (n-k+1)^{D-1}
  \exp\!\left(-\frac{(n-k)(r+1)}{n}\right),
\end{equation}
proving the theorem.
\end{proof}

We can set \(r=0\) in our \zcref{thm:exp-df-states-operator} and easily recover the standard de Finetti representation using mixture of iid states, in the stronger form of an operator inequality.

\begin{corollary}[(Standard de Finetti representation for states)]
\label{cor:std-df-state}
Fix integers \(0\leq k\leq n\) and let \(|\Psi^{(n)}\rangle\in\operatorname{Sym}^{n}(\mathcal H)\) be a symmetric pure state. Then there exists a probability measure \(\mu_{n}\) on the set of pure states in $\mathcal H$, such that
\begin{equation}
    \left(1-\frac{Dk}{n}\right)\Psi^{(k)}
    \leq
    \int
    |\phi\rangle\!\langle\phi|^{\otimes k}
    \,\mathrm d\mu_{n}(\phi).
\end{equation}
More explicitly, the probability measure can be chosen as
\begin{equation}
\mathrm d\mu_{n}(\phi)
=
d_n\,
\langle\phi|^{\otimes n}
\Psi^{(n)}
|\phi\rangle^{\otimes n}
\,\mathrm d\phi.
\end{equation}
\end{corollary}

\begin{proof}
Set \(r=0\) in \zcref{thm:exp-df-states-operator} and notice that every \(\binom{k}{0}\)-almost-iid state along \(\lvert\phi\rangle\) equals \(\lvert\phi\rangle\!\langle\phi\rvert^{\otimes k}\). Moreover,
\begin{equation}
    1-\epsilon_{n,k,0}
    =p_0
    =\frac{d_{n-k}}{d_n}\ge 1-\frac{Dk}{n},
\end{equation}
where the last bound is a standard estimate, which can be found in~\cite[Corollary~II.3]{christandl2007one}. Setting \(r=0\) in \zcref{eq:state-prob-measure} gives the explicit probability measure, notice that it does not depend on $k$.
\end{proof}

\zcref{thm:de Finetti mixed states} in the introduction follows directly from \zcref{thm:exp-df-states-operator} and \zcref{cor:std-df-state}.

\begin{proof}[Proof of \zcref{thm:de Finetti mixed states}]
Let \(\rho^{(k)}\coloneq\Tr_{n-k}\rho^{(n)}\). By the symmetric-purification
argument in \zcref{sec:state-df-setup}, there exist \(E\simeq\mathcal H\)
and a pure state
\(\Psi^{(n)}\in\operatorname{Sym}^{n}(\mathcal H\otimes E)\), with local
dimension \(D=d^2\), such that
\(\Tr_{E^n}\Psi^{(n)}=\rho^{(n)}\). Applying \zcref{thm:exp-df-states-operator} to
\(\Psi^{(n)}\) gives
\begin{equation}
    (1-\varepsilon_{nkr})\Psi^{(k)}
    \leq
    \int
        \omega_{\phi}^{(k,r)}
        \,\mathrm d\mu_{nkr}(\phi),
    \qquad
    \varepsilon_{nkr}
    \leq
    (n-k+1)^{d^2-1}
    \exp\!\left(
        -\frac{(n-k)(r+1)}{n}
    \right).
\end{equation}
Taking the partial trace over \(E^k\), which preserves the operator inequality, yields
\begin{equation}
    (1-\varepsilon_{nkr})\rho^{(k)}
    \leq
    \int
        \Tr_{E^k}\omega_{\phi}^{(k,r)}
        \,\mathrm d\mu_{nkr}(\phi).
\end{equation}
Writing \(\sigma_{\phi}\coloneq\Tr_E\phi\), \zcref{def:almost-iid-state} implies that
\(\Tr_{E^k}\omega_{\phi}^{(k,r)}\) is
\(\binom{k}{r}\)-almost-iid along \(\sigma_{\phi}\). Pushing the measure forward under \(\phi\mapsto\sigma_{\phi}\), and conditionally averaging over states corresponding to the same \(\sigma_{\phi}\), gives the claimed measure on \(\mathcal D(\mathcal H)\); this averaging preserves the almost-iid property because, for a fixed reference state, the set of almost-iid states is convex \cite[Remark 2.4 (g)]{mazzola2026almost}.

For \(r=0\), \zcref{cor:std-df-state} applied to \(\mathcal H\otimes E\) gives \(\varepsilon_{nk0}\leq Dk/n=d^2k/n\); after tracing out \(E^k\), the
pure iid atoms \(\phi^{\otimes k}\) become mixed iid atoms
\((\Tr_E\phi)^{\otimes k}\). If \(\rho^{(n)}\) is classical, write
\(\rho^{(n)}=\sum_{\boldsymbol x}p_{\boldsymbol x}
\lvert\boldsymbol x\rangle\!\langle\boldsymbol x\rvert\) and consider
the symmetric vector
\(\lvert\Psi^{(n)}\rangle
=\sum_{\boldsymbol x}\sqrt{p_{\boldsymbol x}}
\lvert\boldsymbol x\rangle\in\mathcal H^{\otimes n}\).
Applying \zcref{cor:std-df-state} with \(D=d\), followed by dephasing in
this basis, gives the asserted bounds and a measure supported on
classical states.
\end{proof}

Finally, as we have pointed out in the introduction, the familiar trace-distance de Finetti theorems, both standard and exponential, can be easily recovered from the operator inequality version, in fact one recovers de Finetti theorems for a very wide class of distances.
\begin{corollary}[(de Finetti theorems in other distances)]
\label{cor:trace-distance-df-thm}
Let $d(ρ, σ)$ be any distance measure such that $σ \geq λρ$ implies $d(ρ, σ) \leq f(λ)$ for a function $f(λ)$. Then, under the assumptions of \zcref{thm:exp-df-states-operator}, 
\begin{equation}
    d\left(
       \Psi^{(k)},
       \int\omega_{\phi}^{(k,r)}
       \,\mathrm d\mu_{nkr}(\phi)
    \right)
    \leq
    f\left(1 - (n-k+1)^{D-1}
    \exp\!\left(
       -\frac{(n-k)(r+1)}{n}
    \right)\right).
\end{equation}
And for \(r=0\), from \zcref{cor:std-df-state},
\begin{equation}
d\left(
       \Psi^{(k)}, 
       \int\lvert\phi\rangle\!\langle\phi\rvert^{\otimes k}
       \,\mathrm d\mu_{n}(\phi)
    \right)
    \leq f\left(1 - \frac{Dk}{n}\right).
\end{equation}
In particular this applies when
\begin{enumerate}
    \item $d(ρ, σ) = \vertiii{ρ - σ}$ for any norm $\vertiii{\cdot}$, in which case $f(λ) = (1 - λ) \sup_{ρ, σ \in \mathcal{D}(\mathcal{H}^{\otimes k})} \vertiii{ρ - σ}$. If $\vertiii{\cdot} = \norm{\cdot}{p}$ is a Schatten norm ($p \in [1, \infty]$), then $f_p(λ) = 2^{1/p} (1 - λ)$.
    \item $d(ρ, σ) = \mathbb{D}(ρ\|σ)$ for a divergence $\mathbb{D}$ which is anti-monotonous in the second argument (i.e.\ $σ_1 \leq σ_2$ implies $\mathbb{D}(ρ\|σ_1) \geq \mathbb{D}(ρ\|σ_2)$) and satisfies the scaling property $D(ρ\|λρ) = -\log(λ)$. This is for example the case for any quantum R\'enyi divergence as axiomatically characterized in \cite{tomamichel_quantum_2016}. In this case $f(λ) = - \log(λ)$. 
    \item Smoothed versions of the divergences in 2., in particular $\mathbb{D} = D_{\max}^{ε}$ with smoothing over unnormalized states (in trace-distance or purified distance) in which case $f(λ) = -\log(λ) + \log(1 - ε)$, or $\mathbb{D} = D_H^ε$ in which case $f(λ) = -\log(λ) - \log(1 - ε)$ (see e.g.\ \cite{regula2026tight} for definitions and properties of these quantities). 
    
\end{enumerate}
\end{corollary}

\begin{proof}
For general $d(ρ, σ)$, the operator inequality \zcref{eq:state-dF-op-ineq} in \zcref{thm:exp-df-states-operator} directly gives the desired relations with $λ = 1 - ε_{nkr}$. It remains to show that the given examples for $d(ρ, σ)$ satisfy the required properties. For the norm case, $λρ \leq σ$ implies $σ = λρ + (1 - λ)ω$ for some state $ω$ and so $\vertiii{ρ - σ} = \vertiii{(1 - λ)ρ - (1 - λ)ω} = (1 - λ) \vertiii{ρ - ω} \leq f(λ)$. The divergence case 2.\ is obvious. For 3., remember that $D_{\max}^{ε}(ρ\|λρ) = \inf_{\tilde{ρ} \approx^ε ρ} D_{\max}(\tilde{ρ}\|λρ)$, where the closeness $\tilde{ρ} \approx^ε ρ$ can be defined as either in trace-distance or purified distance, and in both cases $\tilde{ρ} = λρ$ is a feasible (sub-normalized) state and gives the desired bound. For $D_H^ε(ρ\|λρ) = - \log\inf_{0 \leq M \leq \mathbbm{1}}\Set{\Tr(M λ ρ) | \Tr(M ρ)\geq 1 - ε}$, $M = (1 - ε) \mathbbm{1}$ is feasible and also optimal.  
\end{proof}

\section{Quantum channel de Finetti theorems as operator inequalities}
\label{sec:channel-df}

\subsection{Setup for channel de Finetti theorems}
\label{sec:channel-df-setup}
Analogously to state de Finetti theorems, which represent marginals of permutation-invariant states, finite quantum de Finetti theorems for channels concern marginals \(\mathcal N^{(k)}\) of permutation-covariant channels
\(\mathcal N^{(n)}\in\mathrm{CPTP}(A^n,B^n)\), where
\(\dim A=d_A\) and \(\dim B=d_B\). As noted in the introduction, however, while the marginal of a state is unambiguously defined by a partial trace, the marginal of a channel requires additional care: discarding some input and output systems does not generally produce an operationally meaningful channel from the remaining inputs to the remaining outputs because the remaining outputs may still depend on the states at the discarded inputs.

For \(0\leq k\leq n\), the simplest convention is to define the reduced channel by requiring its Choi state to equal the corresponding marginal of the Choi state of the full channel\footnote{For notational convenience, throughout this section we regard the channels as maps from \(A'^{\,n}\) to \(B^n\), where \(A'\simeq A\), so that their Choi states \(J^{(n)}\coloneq J(\mathcal N^{(n)})\) belong to \(\mathcal D((AB)^n)\). Similarly \(J^{(k)}\coloneq J(\mathcal N^{(k)})\in\mathcal D((AB)^k)\).}:
\begin{equation}
    J^{(k)}_{A^kB^k}
    =
    \Tr_{A_{k+1}\cdots A_nB_{k+1}\cdots B_n}
    J^{(n)}_{A^nB^n}.
    \label{eq:reduced-choi-state}
\end{equation}
At the channel level, this amounts to feeding the maximally mixed state \(\tau_{A'}\) into each discarded input and tracing out the corresponding outputs, defining a reduced channel from \(A_1\cdots A_k\) to \(B_1\cdots B_k\): \(\forall \rho_{A'^k}\in\mathcal{D}(A'^k)\),
\begin{equation}
    \mathcal N^{(k)}_{A'^k\to B^k}(\rho_{A'^k})
    \coloneq
    \Tr_{B_{k+1}\cdots B_n}
    \mathcal N^{(n)}_{A'^n\to B^n}
    \left(
        \rho_{A'^k}\otimes\tau_{A'}^{\otimes(n-k)}
    \right).
    \label{eq:reduced-channel}
\end{equation}
Although natural from the Choi-state perspective, this prescription singles out the maximally mixed state without an operational reason. A more meaningful approach is to impose a no-signalling condition from the discarded inputs to the remaining outputs. The remaining outputs then depend only on the remaining inputs, so the same reduced channel is obtained irrespective of the states supplied at the discarded input positions.

For \(n\in\mathbb N\), let \(\mathcal N^{(n)}\in\mathrm{CPTP}(A'^n,B^n)\). For \(S\subseteq[n]\), denote \(A'_S=\bigotimes_{j\in S}A'_j\), \(B_S=\bigotimes_{j\in S}B_j\). We impose the following conditions on the channel \(\mathcal N^{(n)}\) under consideration:
\begin{enumerate}[label=(\roman*)]
    \item \emph{Permutation covariance:} \(\forall\pi\in S_n, \rho_{A'^n}\in\mathcal{D}(A'^n)\),
    \begin{equation}
        \mathcal N^{(n)}
        \left(
            U_{A'^n}^{\pi}\rho_{A'^n}
            \bigl(U_{A'^n}^{\pi}\bigr)^{\dagger}
        \right)
        =
        U_{B^n}^{\pi}\mathcal N^{(n)}(\rho_{A'^n})
        \bigl(U_{B^n}^{\pi}\bigr)^{\dagger};
        \label{eq:permutation-covariance}
    \end{equation}

    \item \emph{No-signalling:} \(\forall\rho_{A'^n}\in\mathcal{D}(A'^n)\),
    \begin{equation}
        \Tr_{B_n}\mathcal N^{(n)}(\rho_{A'^n})
        =
        \Tr_{B_n}\mathcal N^{(n)}
        \left(
            \Tr_{A'_n}(\rho_{A'^n})\otimes\tau_{A'_n}
        \right),
        \label{Eq:no-signalling-last-position}
    \end{equation}
    i.e., \(A'_n\) does not signal to \(B_{1}\cdots B_{n-1}\).
\end{enumerate}

At the Choi-state level, permutation covariance of the channel is equivalent to permutation invariance of the Choi state:
\begin{equation}
    \bigl(U_{A^n}^{\pi}\otimes U_{B^n}^{\pi}\bigr)
    J^{(n)}
    \bigl(U_{A^n}^{\pi}\otimes U_{B^n}^{\pi}\bigr)^\dagger
    =
    J^{(n)},
    \qquad \forall\pi\in S_n,
    \label{eq:choi-permutation-invariance}
\end{equation}
while the no-signalling condition is equivalent to
\begin{equation}
    \Tr_{B_n}J^{(n)}
    =
    \Tr_{A_nB_n}J^{(n)}\otimes\tau_{A_n}.
\end{equation}

The no-signalling condition stated above only assumes that the input supplied to \(A'_n\) cannot influence the reduced output on \(B_{[n-1]}\). We now show that this, together with permutation covariance, implies the corresponding no-signalling condition on arbitrary subsets of systems. For \(j\in[n]\), introduce the replacement map \(\mathcal R_j\) that replaces the \(j\)-th input by the maximally mixed state:
\begin{equation}
    \mathcal R_j(\rho_{A'^n})
    \coloneq
    \Tr_{A'_j}(\rho_{A'^n})\otimes\tau_{A'_j}.
\end{equation}
For each \(j\), let \(\sigma=(j\,n)\) be the permutation exchanging \(j\) and \(n\), use the abbreviation \(U_{A'}=U_{A'^n}^{\sigma}\) and \(U_B= U_{B^n}^{\sigma}\). Then, for any \(X_{B^n}\in\mathcal{L}(B^n)\),
\begin{equation}
    \Tr_{B_n}\!\left(U_BX_{B^n}U_B^\dagger\right)
    =
    \Tr_{B_j}(X_{B^n}),
    \qquad
    \mathcal R_n\!\left(U_{A'}\rho U_{A'}^\dagger\right)
    =
    U_{A'}\mathcal R_j(\rho)U_{A'}^\dagger.
\end{equation}
Applying permutation covariance and \zcref{Eq:no-signalling-last-position} gives, \(\forall \rho_{A'^n}\in\mathcal{D}(A'^n)\),
\begin{equation}
\label{eq:no-sig-for-j}
\begin{aligned}
    \Tr_{B_j}\mathcal N^{(n)}(\rho_{A'^n})
    &=
    \Tr_{B_n}\!\left[
        U_B\mathcal N^{(n)}(\rho_{A'^n})U_B^\dagger
    \right] \\
    &=
    \Tr_{B_n}\mathcal N^{(n)}
    \left(U_{A'}\rho_{A'^n} U_{A'}^\dagger\right) \\
    &=
    \Tr_{B_n}\mathcal N^{(n)}
    \left(
        \mathcal R_n
        \left(U_{A'}\rho_{A'^n} U_{A'}^\dagger\right)
    \right) \\
    &=
    \Tr_{B_n}\mathcal N^{(n)}
    \left(
        U_{A'}\mathcal R_j(\rho_{A'^n})U_{A'}^\dagger
    \right) \\
    &=
    \Tr_{B_n}\!\left[U_B \mathcal N^{(n)}
    \left(
        \mathcal R_j(\rho_{A'^n})
    \right)U_{B}^\dagger\right] \\
    &=
    \Tr_{B_j}\mathcal N^{(n)}
    \left(\mathcal R_j(\rho_{A'^n})\right).
\end{aligned}
\end{equation}
Operationally this means no-signalling from \(A_j\) to \(B_{[n]\setminus j}\). Let \(S=\{j_1,\ldots,j_s\}\subseteq[n]\) and apply \zcref{eq:no-sig-for-j} successively to \(j_1,\ldots,j_s\) yields
\begin{equation}
\begin{aligned}
    \Tr_{B_S}\mathcal N^{(n)}(\rho_{A'^n})
    =
    \Tr_{B_{S\setminus \{j_s\}}}\Tr_{B_{j_s}}\mathcal N^{(n)}
    \bigl(\rho_{A'^n}\bigr)
    &=
    \Tr_{B_{S}}\mathcal N^{(n)}
    \bigl(\mathcal R_{j_s}(\rho_{A'^n}\bigr)\bigr)
    =
    \cdots\\
    &=
    \Tr_{B_S}\mathcal N^{(n)}
    \bigl(\mathcal R_{j_1}\cdots\mathcal R_{j_s}(\rho_{A'^n})\bigr)\\
    &=
    \Tr_{B_S}\mathcal N^{(n)}
    \left(
        \Tr_{A'_S}(\rho_{A'^n})\otimes\tau_{A'_S}
    \right).
\end{aligned}
\label{eq:no-signalling-arbitrary-subset}
\end{equation}
This establishes the strongest no-signalling: for any \(S\subseteq [n]\), the inputs \(A'_S\) cannot signal to the complementary outputs \(B_{[n]\setminus S}\). At the Choi-state level, this translates to
\begin{equation}
    \Tr_{B_S}J^{(n)}
    =
    \Tr_{A_SB_S}J^{(n)}\otimes\tau_{A_S},
    \qquad \forall S\subseteq[n].
    \label{eq:choi-no-signalling-subsets}
\end{equation}

Importantly, under the no-signalling condition, the maximally mixed extension in the definiton of reduced channel (\zcref{eq:reduced-channel}) is merely a convenient choice. Indeed, for any state \(\omega_{A'^n}\), applying \zcref{eq:no-signalling-arbitrary-subset} with \(S=\{k+1,\ldots,n\}\) gives
\begin{equation}
    \Tr_{B_{k+1}\cdots B_n}
    \mathcal N^{(n)}(\omega_{A'^n})
    =
    \Tr_{B_{k+1}\cdots B_n}
    \mathcal N^{(n)}
    \left(
        \Tr_{A'_{k+1}\cdots A'_n}(\omega_{A'^n})
        \otimes\tau_{A'}^{\otimes(n-k)}
    \right)
    =
    \mathcal N^{(k)}
    \left(
        \Tr_{A'_{k+1}\cdots A'_n}(\omega_{A'^n})
    \right).
\label{eq:reduced-channel-arbitrary-extension}
\end{equation}
Consequently, the reduced state on the remaining outputs depends only on the reduced state on the remaining inputs, even if the latter is extended arbitrarily (possibly entangled) onto the discarded inputs. This highlights the importance of the no-signalling assumption, without which the definition of reduced channel would become operationally unnatural.

The no-signalling condition \zcref{Eq:no-signalling-last-position} is a significant restriction on the set of channels to which the de Finetti theorems apply; it is not hard to construct channels which are permutation covariant but do not satisfy the no-signalling constraint. For example, one can take the channel that applies a random permutation to the systems
\begin{equation}
    \mathcal N^{(n)}\!\left(\rho^{(n)}\right)
    \coloneq
    \frac{1}{n!}
    \sum_{\pi\in S_n}
    U_{B^n}^{\pi}\rho^{(n)}
    \bigl(U_{B^n}^{\pi}\bigr)^\dagger ,
\end{equation}
and it is not hard to see that for a product input \(\rho^{(n)}=\rho_1\otimes\cdots\otimes\rho_n\), the reduced output \(\Tr_{B_n}\mathcal N^{(n)}(\rho^{(n)})\) depends on \(\rho_n\).

The no-signalling condition is trivially satisfied by mixtures of iid channels, however it is also not hard to construct examples of channels which satisfy the no-signalling condition and which are not mixture of iid channels:
\begin{enumerate}
\item \emph{Replacer channels.} Let
    \begin{equation}
        \mathcal N^{(n)}\!\left(\rho^{(n)}\right)
        \coloneq
        \Tr\!\left(\rho^{(n)}\right)\sigma^{(n)},
    \end{equation}
    where \(\sigma^{(n)}\in\mathcal D(B^n)\) is permutation invariant but not a mixture of iid states, for example a GHZ state. This channel is permutation covariant and no-signalling because its output is independent of the input, but it cannot be a mixture of iid channels, since it maps every product input to the entangled state \(\sigma^{(n)}\).

    \item \emph{Symmetrized non-iid product channels.} Given channels \(\Lambda_1,\ldots,\Lambda_n\in\mathrm{CPTP}(A,B)\), let
    \begin{equation}
        \mathcal N^{(n)}
        \coloneq
        \frac{1}{n!}
        \sum_{\pi\in S_n}
        \mathcal U_{B^n}^{\pi}
        \circ
        \bigl(\Lambda_1\otimes\cdots\otimes\Lambda_n\bigr)
        \circ
        \mathcal U_{A^n}^{\pi^{-1}},
    \end{equation}
    where \(\mathcal U_{\mathcal H^{\otimes n}}^\pi(X)\coloneq U_{\mathcal H^{\otimes n}}^\pi X(U_{\mathcal H^{\otimes n}}^\pi)^\dagger\). Each term is a product of local channels, so \(\mathcal N^{(n)}\) is no-signalling, while averaging over permutations enforces permutation covariance. To see that it need not be a mixture of iid channels, take the \(\Lambda_i\) to be Pauli channels \(\Lambda_i(\rho)=
    \sum_{s\in\{I,X,Y,Z\}}
        p_i(s)\,\sigma_s\rho\sigma_s^\dagger\). Then \(\mathcal N^{(n)}\) is a Pauli channel that applies a Pauli string \(\boldsymbol s\in\{I,X,Y,Z\}^n\) according to the exchangeable distribution \(p(\boldsymbol s)
    =
    \frac{1}{n!}
    \sum_{\pi\in S_n}
        \prod_{i=1}^{n}p_i\!\left(s_{\pi(i)}\right)\). This distribution need not be a mixture of iid distributions. On the other hand, a mixture of iid Pauli channels would imply \(p(\boldsymbol s)\) to be a mixture of iid distributions.

    \item \emph{Entangled environments.} Let \(\Gamma_{AE\to B}\) be a channel and let \(\omega_{E^n}\) be permutation invariant. Then
    \begin{equation}
        \mathcal N^{(n)}\!\left(\rho_{A^n}\right)
        \coloneq
        (\Gamma_{AE \to B})^{\otimes n}
        \left(
            \rho_{A^n}
            \otimes
            \omega_{E^n}
        \right)
    \end{equation}
    is permutation covariant and no-signalling. If \(\omega_{E^n}\) is entangled across the $n$ copies of $E$, it is possible for the output \(\mathcal N^{(n)}\!\left(\rho_{A^n}\right)\) to be entangled even if the input \(\rho_{A^n}\) is separable. This is impossible for a mixture of iid channels, which maps separable inputs to separable outputs.

\end{enumerate}

The channels defined by \zcref{eq:reduced-channel} form a consistent finite sequence, in the sense that for every \(0\leq k\leq m\leq n\),
\begin{equation}
\begin{aligned}
    \mathcal N^{(k)}
    \circ\Tr_{A'_{k+1}\cdots A'_m}
    &=
    \Tr_{B_{k+1}\cdots B_m}
    \circ\mathcal N^{(m)},\\
    \mathcal N^{(k)}
    \left(
        U_{A'^k}^{\pi}\rho_{A'^k}
        \bigl(U_{A'^k}^{\pi}\bigr)^\dagger
    \right)
    &=
    U_{B^k}^{\pi}\mathcal N^{(k)}(\rho_{A'^k})
    \bigl(U_{B^k}^{\pi}\bigr)^\dagger.
\end{aligned}
\label{eq:reduced-channel-consistency}
\end{equation}
These conditions are precisely the finite versions of the extendibility and symmetry conditions used to define an exchangeable sequence of quantum operations in the infinite de Finetti theorem of Fuchs, Schack, and Scudo~\cite{fuchs2004finetti}. If a sequence
\(\{\mathcal N^{(m)}\}_{m\in\mathbb N}\) satisfies these conditions for every \(m\in\mathbb N\), their theorem claims the exact de Finetti representation \(\mathcal N^{(m)}=\int
        \mathcal E_{A'\to B}^{\otimes m}\,
        \mathrm d\nu(\mathcal E)\). Our aim is to establish a finite-\(n\) version of this representation.

\subsection{Standard de Finetti theorem for channels}
\label{sec:std-channel-df}

Choi states are generally mixed, so we will work in the symmetric purification subspace. Permutation covariance of \(\mathcal N^{(n)}\) implies permutation invariance of \(J^{(n)}\in\mathcal D((AB)^n)\), and hence \(J^{(n)}\) admits a purification \(\lvert\Psi^{(n)}\rangle\in\operatorname{Sym}^{n}(ABE)\), where \(E\simeq AB\) and \(\Tr_{E^n}\Psi^{(n)}=J^{(n)}\). The local dimension is \(D=\dim(ABE)=d_A^2d_B^2\). 

The idea is to apply the standard state de Finetti \zcref{cor:std-df-state} to \(\Psi^{(n)}\) and then trace out \(E^k\), obtaining iid atoms of the form \(\phi_{AB}^{\otimes k}\). These atoms are states but need not be Choi states, since generally \(\phi_A\neq\tau_A\). Rather than restricting the integration variable \(\phi_{ABE}\) to satisfy the Choi constraint, we impose the constraint directly on each atom: for every \(\phi_{AB}\), including those that are not Choi, we construct a Choi state \(\hat\phi_{AB}\) that completes \(\phi_{AB}\) in the sense that \(\hat\phi_{AB}\geq c_{\phi}\phi_{AB}\) with completion cost \(c_{\phi}\) (see \zcref{lem:single-copy-choi-completion}). Since an iid tensor-power state along a Choi state is also a Choi state, this single-copy Choi completion can be lifted to multiple-copy and supply an iid Choi state \(\hat\phi_{AB}^{\otimes k}\).

\begin{lemma}[(Single-copy Choi-completion)]
\label{lem:single-copy-choi-completion}
Let \(\rho_{AB}\) be a bipartite state with reduced state \(\rho_A\), and define
\(c_\rho\coloneq\frac{1}{d_A\|\rho_A\|_\infty}\). Then \(d_A^{-1}\leq c_\rho\leq 1\), and \(\widehat{\rho}_{AB}\coloneq c_\rho\rho_{AB}+\bigl(\tau_A-c_\rho\rho_A\bigr)\otimes\tau_B\) is a Choi state, i.e., \(\widehat{\rho}_{A}=\tau_A\), satisfying \(\widehat{\rho}_{AB}\geq c_\rho\rho_{AB}\geq 0\).
\end{lemma}

\begin{proof}
By construction \(\widehat{\rho}_{A}=\tau_A\). The bound on \(c_\rho\) follows trivially from \(\frac{1}{d_A}\leq\|\rho_A\|_\infty\leq 1\). Moreover,
\begin{equation}
    c_\rho\rho_A
    \leq c_\rho\|\rho_A\|_\infty\id_A
    =\frac{\id_A}{d_A}
    =\tau_A.
\end{equation}
Therefore \(\tau_A-c_\rho\rho_A\geq0\), and thus \( \widehat{\rho}_{AB}-c_\rho\rho_{AB}=\bigl(\tau_A-c_\rho\rho_A\bigr)\otimes\tau_B\geq0\).
\end{proof}

After the Choi completion, the main task reduces to bounding the averaged completion cost \(\int c_{\phi}^{-k}\,\mathrm d\mu_n(\phi)\). Intuitively, because the state \(\Psi^{(k)}\) being represented arises from \(\Psi^{(n)}\) which satisfies the Choi constraint, the measure \(\mu_n\) should concentrate near those \(\phi\) whose \(A\)-marginals are close to \(\tau_A\), and for which the completion cost is small. To make this intuition quantitative, we use \zcref{lem:fixed-marginal-measure-independence}, whose proof relies on \zcref{lem:fixed-marginal-haar-moment}. The latter is a fixed-marginal generalization of the Haar-moment identity in \zcref{eq:haar-moment-symmetric-subspace} and will also play an important role in the exponential channel de Finetti construction later. It reformulates the constructions of~\cite[Lemma~3.1]{fawzi2015quantum} and~\cite[Lemma~10]{nahar2024postselection} to suit our purpose; for completeness, we provide an elementary proof of \zcref{lem:fixed-marginal-haar-moment} in~\zcref{App:Proof:fixed-marginal-haar-moment}.

\begin{lemma}
\label{lem:fixed-marginal-haar-moment}
Let \(A\) and \(R\) be finite-dimensional Hilbert spaces with \(d_A\leq d_R\), fix a basis and let \(|\theta\rangle_{AR}\coloneq\sum_{i=1}^{d_A}|i\rangle_A|i\rangle_R\), and \(\mathrm dU\) be the Haar measure on \(\mathsf U(R)\). For \(m\in\mathbb N\), let \(P_{\mathrm{sym}}^{(m)}\) be the projection to \(\operatorname{Sym}^{m}(AR)\) and define \(K_{A^m}\coloneq \left(\Tr_{R^m}P_{\mathrm{sym}}^{(m)}\right)^{-1}>0\). Then
\begin{equation}
\label{eq:constrained-Haar}
    \int_{\mathsf U(R)}
    \bigl[(\id_A\otimes U_R)|\theta\rangle\!\langle\theta|_{AR}(\id_A\otimes U_R^\dagger)\bigr]^{\otimes m}\,\mathrm dU=(K_{A^m}\otimes\id_{R^m})P_{\mathrm{sym}}^{(m)}
       =P_{\mathrm{sym}}^{(m)}(K_{A^m}\otimes\id_{R^m}).
\end{equation}
Consequently, for any pure state \(|\phi\rangle_{AR}\) and \(\phi_A\coloneq\Tr_R|\phi\rangle\!\langle\phi|_{AR}\),
\begin{equation}
    \int_{\mathsf U(R)}
    \bigl[(\id_A\otimes U_R)|\phi\rangle\!\langle\phi|_{AR}(\id_A\otimes U_R^\dagger)\bigr]^{\otimes m}\,\mathrm dU
    =\bigl(\phi_A^{\otimes m}K_{A^m}\otimes\id_{R^m}\bigr)P_{\mathrm{sym}}^{(m)}.
    \label{eq:fixed-marginal-haar-moment}
\end{equation}
\end{lemma}

Although the measure \(\mu_n\) from \zcref{cor:std-df-state} depends on \(\Psi^{(n)}\), \zcref{lem:fixed-marginal-measure-independence} below shows that expectations of functions of the \(A\)-marginal \(\phi_A\) alone do not. They may therefore be evaluated using a convenient product state; see \zcref{App:Proof:lem:fixed-marginal-measure-independence} for the proof.

\begin{lemma}
\label{lem:fixed-marginal-measure-independence}
Let \(A\) and \(R\) be finite-dimensional Hilbert spaces with \(d_A\leq d_R\). For any \(n\in\mathbb N\) and symmetric pure state \(|\Psi^{(n)}\rangle\in\operatorname{Sym}^{n}(AR)\) satisfying \(\Psi_{A^n}^{(n)}=\tau_A^{\otimes n}\), consider the probability measure \(\mu_n\) on the set of pure states \(|\phi\rangle_{AR}\) defined by
\begin{equation}
    \mathrm d\mu_{n}(\phi)\coloneq d_n\,\langle\phi|^{\otimes n}\Psi^{(n)}|\phi\rangle^{\otimes n}\,\mathrm d\phi.
\end{equation}
Then for any bounded measurable function \(f\) of the reduced state \(\phi_A=\Tr_R|\phi\rangle\!\langle\phi|_{AR}\), we have
\begin{equation}
    \int f(\phi_A)\,\mathrm d\mu_{n}(\phi)
    =d_n\int f(\phi_A)
    \left|\langle\psi|\phi\rangle\right|^{2n}\mathrm d\phi,
    \label{eq:replace-symmetric-purification}
\end{equation}
where \(|\psi\rangle_{AR}\) can be any fixed pure state satisfying \(\psi_A=\tau_A\).
\end{lemma}

With the preceding lemmas in place, we can prove the channel representation theorem.

\begin{proof}[Proof of \zcref{thm:standard-channel-definetti}]
Let \(E\simeq AB\), \(D=\dim(ABE)=d_A^2d_B^2\), and let
\(|\Psi^{(n)}\rangle\in\operatorname{Sym}^n(ABE)\) be a symmetric purification of \(J^{(n)}=J(\mathcal{N}^{(n)})\). Its satisfies the Choi constraint \(\Psi_{A^n}^{(n)}=\tau_A^{\otimes n}\). Let
\(\Psi^{(k)}\coloneq\Tr_{(ABE)^{n-k}}\Psi^{(n)}\), then \(\Tr_{E^k}\Psi^{(k)}=\Tr_{(AB)^{n-k}}J^{(n)}=J^{(k)}=J(\mathcal{N}^{(k)})\), see \zcref{eq:reduced-choi-state}.

Our starting point is the standard de Finetti representation for states in \zcref{cor:std-df-state},
\begin{equation}
    \left(1-\frac{Dk}{n}\right)\Psi^{(k)}
    \leq
    \int \phi_{ABE}^{\otimes k}\,\mathrm d\mu_n(\phi), \qquad \mathrm d\mu_n(\phi)
    =d_n\langle\phi|^{\otimes n}\Psi^{(n)}
    |\phi\rangle^{\otimes n}\,\mathrm d\phi.
\end{equation}
Taking the partial trace over \(E^k\) yields \(\left(1-\frac{Dk}{n}\right)J^{(k)}\leq\int \phi_{AB}^{\otimes k}\,\mathrm d\mu_n(\phi)\). For every \(\phi\), define \(c_\phi\coloneq\frac{1}{d_A\|\phi_A\|_\infty}\) and \(\widehat{\phi}_{AB}\coloneq c_\phi\phi_{AB}+\bigl(\tau_A-c_\phi\phi_A\bigr)\otimes\tau_B\). By \zcref{lem:single-copy-choi-completion},
\(\widehat{\phi}_{AB}\) is a Choi state and \(\widehat{\phi}_{AB}\geq c_\phi\phi_{AB}\). Tensoring this inequality \(k\) times gives \(\phi_{AB}^{\otimes k}\leq c_\phi^{-k}\widehat{\phi}_{AB}^{\otimes k}\). Consequently,
\begin{equation}
\label{eq:post-choi-completion}
    \left(1-\frac{Dk}{n}\right)J^{(k)}
    \leq
    \int\bigl(d_A\|\phi_A\|_\infty\bigr)^k
    \widehat{\phi}_{AB}^{\otimes k}\,
    \mathrm d\mu_n(\phi)
    =Z_{nk}\int\widehat{\phi}_{AB}^{\otimes k}\,
    \mathrm d\nu_{nk}(\phi),
\end{equation}
where we have defined
\begin{equation}
d\nu_{nk}(\phi)\coloneq \frac{\bigl(d_A\|\phi_A\|_\infty\bigr)^k \mathrm d\mu_n(\phi)}{Z_{nk}},
\qquad Z_{nk}
    \coloneq
    \int\bigl(d_A\|\phi_A\|_\infty\bigr)^k
    \,\mathrm d\mu_n(\phi).
\end{equation}

Since \(d_A\|\phi_A\|_\infty\geq1\), we have \(Z_{nk}\geq1\). It remains to find an explicit upper bound on \(Z_{nk}\) in terms of \(k\) and \(n\). Apply \zcref{lem:fixed-marginal-measure-independence} with \(R=BE\) and
\(f(\phi_A)=(d_A\|\phi_A\|_\infty)^k\). For any fixed pure state
\(|\psi\rangle_{ABE}\) satisfying \(\psi_A=\tau_A\), we obtain
\begin{equation}
    Z_{nk}
    =
    d_n\int
    \bigl(d_A\|\phi_A\|_\infty\bigr)^k
    |\langle\psi|\phi\rangle|^{2n}\,\mathrm d\phi.
\end{equation}
We bound the integrand to be only dependent on the overlap: by monotonicity of \(p\)-norm and data-processing of trace norm under partial trace, we have
\begin{equation}
    \|\phi_A-\tau_A\|_\infty\leq\|\phi_A-\tau_A\|_1
    \leq
    \|\phi_{ABE}-\psi_{ABE}\|_1
    =
    2\sqrt{1-|\langle\psi|\phi\rangle|^2}.
\end{equation}
By triangle inequality \(\|\phi_A\|_\infty\le\|\phi_A-\tau_A\|_\infty+\frac{1}{d_A}\), it follows that
\begin{equation}
Z_{nk}\le
d_n\int
    \left(1+2d_A\sqrt{1-|\langle\psi|\phi\rangle|^2}\right)^{k}
    |\langle\psi|\phi\rangle|^{2n}\,\mathrm d\phi
\end{equation}
The factor \(d_n|\langle\psi|\phi\rangle|^{2n}\mathrm d\phi\) defines a probability
measure that tilts the Haar measure, exponentially favoring states with a large overlap with \(\psi\). For these typical states, the bracket is close to 1. Atypical states having substantially smaller overlap can produce a larger bracket value, but their contribution is exponentially suppressed by the factor \(|\langle\psi|\phi\rangle|^{2n}\).

To be precise, consider the (squared) fidelity between two pure states \(F_{\phi}\coloneq|\langle\psi|\phi\rangle|^2\). For a Haar-random pure state
\(|\phi\rangle\) in a \(D\)-dimensional Hilbert space and a fixed
\(|\psi\rangle\) in the same Hilbert space, the probability distribution of the fidelity \(F\in[0,1]\) is \((D-1)(1-F)^{D-2}\). This is a standard result on average fidelity between random quantum states, which has long been known and studied in random matrix theory, mostly by Życzkowski and others; see, e.g., \cite[Equation~(15)]{zyczkowski2005average}. Equivalently, \(F\sim\operatorname{Beta}(1,D-1)\) follows a  Beta distribution.\footnote{A simple derivation follows from generating a Haar-random state as
\(|\phi\rangle=(\sum_{i=1}^{D}|z_i|^2)^{-1/2}\sum_{i=1}^{D}z_i|i\rangle\),
where the \(z_i\) are independent complex Gaussian variables~\cite{zyczkowski2001induced},\cite[Section~7.6]{bengtsson2007geometry}. One may take \(|\psi\rangle=|1\rangle\), so that
\(|\langle\psi|\phi\rangle|^2=|z_1|^2/\sum_{i=1}^{D}|z_i|^2\). For complex Gaussian variables, \(|z_i|^2\) are independent gamma variables, and the normalized ratio can be calculated to follow a beta distribution.} Hence,
\begin{equation}
\begin{aligned}
Z_{nk}
&\leq d_n(D-1)\int_0^1
    \left(1+2d_A\sqrt{1-F}\right)^k
    F^n(1-F)^{D-2}\,\mathrm dF \\
&=  \int_0^1
    \left(1+2d_A\sqrt U\right)^k
    \frac{U^{D-2}(1-U)^n}{\mathrm B(D-1,n+1)}\,\mathrm dU,
\end{aligned}
\end{equation}
where in the second line we have changed variable \(U\coloneq1-F\) and we have used beta-gamma function relation to get \(d_n(D-1)=\frac{\Gamma(n+D)}{\Gamma(n+1)\Gamma(D-1)}=\mathrm B(D-1,n+1)^{-1}\). Notice the function
\(\mathrm B(D-1,n+1)^{-1}U^{D-2}(1-U)^n\) is the \(\operatorname{Beta}(D-1,n+1)\) PDF. Consequently,
\begin{equation}
Z_{nk}
\leq\mathbb E\left(1+2d_A\sqrt U\right)^k
\leq\mathbb E\exp\left(2d_Ak\sqrt U\right),
\qquad U\sim\operatorname{Beta}(D-1,n+1).
\end{equation}
A technical estimate on this expectation, which we provide in \zcref{App:Proof:eq-Z-estimate}, gives
\begin{equation}
\label{eq:Z-estimate}
Z_{nk}\leq\mathbb E\exp\left(2d_Ak\sqrt U\right)\leq
\exp\left(
2d_Ak\sqrt{\frac{D-1}{n+D-1}}+\frac{d_A^2k^2}{n+D-1}
\right).
\end{equation}

Finally, combining with \zcref{eq:post-choi-completion}, we have
\begin{equation}
    \bigl(1-\epsilon_{nk}\bigr)J^{(k)}
    \leq
    \int\widehat{\phi}_{AB}^{\otimes k}\,
    \mathrm d\nu_{nk}(\phi), \quad 1-\epsilon_{nk}\coloneq \left(1-\frac{Dk}{n}\right)Z_{nk}^{-1}.
\end{equation}
Using \(Z_{nk}\geq1\) and \(1-Z_{nk}^{-1}\leq\ln Z_{nk}\), we conclude that
\begin{equation}
    \epsilon_{nk}=1-\left(1-\frac{Dk}{n}\right)Z_{nk}^{-1}
    \leq
    \ln Z_{nk}+\frac{Dk}{n}
    \leq
    \frac{Dk}{n}+2d_Ak\sqrt{\frac{D-1}{n+D-1}}+\frac{d_A^2k^2}{n+D-1}.
\end{equation}

For each \(\phi\in ABE\), let \(\Phi\in\mathrm{CPTP}(A,B)\) be the channel whose Choi state is \(\widehat{\phi}_{AB}\), and, with a slight abuse of notation, let \(\nu_{nk}\) also denote the pushforward of the measure \(\nu_{nk}\), on the set of pure states in \(ABE\), under the map \(\phi\mapsto\Phi\). Then,
\begin{equation}
    \bigl(1-\epsilon_{nk}\bigr)J(\mathcal{N}^{(k)})
    \leq
    \int J(\Phi)^{\otimes k}\,
    \mathrm d\nu_{nk}(\Phi),
\end{equation}
which implies the CP-order \zcref{eq:standard-channel-definetti}. Substituting \(D=d_A^2d_B^2\) completes the proof.
\end{proof}

\begin{remark}
The error \(\epsilon_{nk}\) in \zcref{thm:standard-channel-definetti}, for fixed \(d_A,d_B\), is \(O(k/\sqrt{n})\), and therefore vanishes whenever \(k=o(\sqrt{n})\). In particular, if \(k=O(n^{1/2-\delta})\) for any \(0<\delta\le1/2\), the error decays polynomially as \(O(n^{-\delta})\). To the best of our knowledge, this is the first finite channel de Finetti theorem with error polynomial in \(n,k\) and local dimensions, even in the substantially stronger form of an operator inequality. In contrast, the channel de Finetti theorem of \cite[Corollary~4.7]{berta2022semidefinite} is formulated in diamond norm, and its error bound contains a factor exponential in \(k\). We have therefore solved the open problem raised by the authors of whether a bound with polynomial dependence on the dimensions and \(n,k\) could be obtained. 
\end{remark}

A more familiar diamond-norm de Finetti theorem can be easily recovered from the operator inequality version.

\begin{corollary}[(Diamond-norm standard de Finetti)]
\label{cor:standard-channel-definetti-diamond}
Under the assumptions of \zcref{thm:standard-channel-definetti},
\begin{equation}
    \left\|
        \Tr_{n-k}\mathcal N^{(n)}
        -
        \int
            \Phi^{\otimes k}
            \,\mathrm d\nu_{nk}(\Phi)
    \right\|_{\diamond}
    \leq
    \frac{2d_A^2d_B^2k}{n}
    +4d_Ak
        \sqrt{
            \frac{d_A^2d_B^2-1}
                 {n+d_A^2d_B^2-1}
        }                                                    
    +\frac{2d_A^2k^2}
           {n+d_A^2d_B^2-1}.
\end{equation}
\end{corollary}

\begin{proof}
Define the mixture of iid channels
\begin{equation}
    \mathcal M^{(k)}
    \coloneq
    \int
        \Phi^{\otimes k}
        \,\mathrm d\nu_{nk}(\Phi).
\end{equation}
The CP-order inequality in \zcref{thm:standard-channel-definetti} implies that \(\exists \ \mathcal R^{(k)}\in\mathrm{CPTP}(A^k,B^k)\) such that
\begin{equation}
    \mathcal M^{(k)}
    =
    (1-\epsilon_{nk})\Tr_{n-k}\mathcal N^{(n)}
    +
    \epsilon_{nk}\mathcal R^{(k)}.
\end{equation}
Consequently,
\begin{equation}
    \left\|
        \Tr_{n-k}\mathcal N^{(n)}
        -
        \mathcal M^{(k)}
    \right\|_{\diamond}
    =
    \epsilon_{nk}
    \left\|
        \Tr_{n-k}\mathcal N^{(n)}
        -
        \mathcal R^{(k)}
    \right\|_{\diamond} 
    \leq
    2\epsilon_{nk},
\end{equation}
where we used that the diamond-norm distance between two channels is at most \(2\). The explicit estimate follows from \zcref{eq:standard-channel-definetti-error}.
\end{proof}

\subsection{Exponential de Finetti theorem for channels}
\label{sec:exp-channel-df}

Motivated by Renner's exponential de Finetti theorem for states~\cite{renner2007symmetry}, we now relax the atoms in the convex mixture from iid channels to almost-iid channels (see \zcref{def:almost-iid-channel}), allowing up to \(r\) defects among the \(k\) retained copies, in exchange for an exponentially decaying representation error. Our present bound in \zcref{thm:exp-channel-df}, however, achieves this decay only when a sufficiently large fraction of the systems is discarded and therefore does not cover the regime in which most systems are retained. See \zcref{rmk:decay-regime} for details.

The setup is the same as in the standard case: we consider the symmetric purification \(\lvert\Psi^{(n)}\rangle\in\operatorname{Sym}^{n}(ABE)\) of \(J^{(n)}\in\mathcal D((AB)^n)\), where \(E\simeq AB\) and \(\Tr_{E^n}\Psi^{(n)}=J^{(n)}\). The local dimension is \(D=\dim(ABE)=d_A^2d_B^2\). Let \(P_{\mathrm{sym}}^{(m)}\) denote the projection onto
\(\operatorname{Sym}^{m}(ABE)\), with \(d_m=\dim\operatorname{Sym}^{m}(ABE)\). 

The single-copy Choi completion \zcref{lem:single-copy-choi-completion} used for the standard de Finetti \zcref{thm:standard-channel-definetti} does not directly extend to the exponential setting. There, completing a single-copy reference state \(\phi_{AB}\) to a Choi state \(\widehat{\phi}_{AB}\) is sufficient for completing the corresponding iid state \(\phi_{AB}^{\otimes k}\) to an iid Choi state \(\widehat{\phi}_{AB}^{\otimes k}\). In the exponential setting, one considers almost iid states (see \zcref{def:almost-iid-state}) as opposed to iid states. Even when the reference state is Choi, an almost-iid state along that reference need not itself be Choi due to defects\footnote{Indeed, one can construct sequence of states \(\omega_{(AB)^k}\) that are \(\binom{k}{1}\)-almost-iid along a fixed Choi state, but \(\frac{1}{2}\lVert\omega_{A^k}-\tau_A^{\otimes k}\rVert_1\) is bounded away from 0 as \(k\to\infty\).}, thus a single-copy Choi completion for the reference does not imply the almost iid state satisfies the Choi constraint. The difficulty with the exponential channel de Finetti theorem originates from the fact that the Choi constraint and the almost-iid constraint generally do not commute: projecting onto an almost-iid subspace can disturb the input marginal, while applying a Choi completion after the projection need not preserve the almost-iid support. One has to add the Choi constraint to both the reference state and the almost iid state.

We first address the Choi constraint on the reference state by incorporating it directly into the integration: rather than integrating over all pure states on \(ABE\), we integrate only over those whose marginal on \(A\) is maximally mixed. More precisely, let
\begin{equation}
    \mathcal P_{ABE}(\tau_A)
    \coloneq
    \left\{
        \hat\phi_{ABE}=|\hat\phi\rangle\!\langle\hat\phi|_{ABE}:
        \Tr_{BE}\hat\phi_{ABE}=\tau_A
    \right\}.
\end{equation}
Equivalently, \(\mathcal P_{ABE}(\tau_A)\) is the set of purifications of \(\tau_A\) with purifying system \(BE\). In particular, for every \(\hat\phi\in\mathcal P_{ABE}(\tau_A)\), the state \(\hat\phi_{AB}\coloneq\Tr_E\hat\phi_{ABE}\) is a normalized Choi state corresponding to a single-copy quantum channel from \(A'\) to \(B\).

To equip \(\mathcal P_{ABE}(\tau_A)\) with a natural probability measure analogous to the Haar measure in the unconstrained case, fix any \(\hat\phi_0\in\mathcal P_{ABE}(\tau_A)\). By the unitary equivalence of purifications, every \(\hat\phi\in\mathcal P_{ABE}(\tau_A)\) can be written as
\begin{equation}
    \hat\phi
    =
    (\id_A\otimes U_{BE})\hat\phi_0
    (\id_A\otimes U_{BE}^{\dagger})
\end{equation}
for some \(U_{BE}\in\mathsf U(BE)\). The normalized Haar measure \(\mathrm dU\) on \(\mathsf U(BE)\) therefore induces a probability measure \(\mathrm d\hat\phi\) on \(\mathcal P_{ABE}(\tau_A)\): for every possibly operator-valued function \(f\) of \(\hat\phi\), we define
\begin{equation}
    \int_{\mathcal P_{ABE}(\tau_A)}
        f(\hat\phi)\,\mathrm d\hat\phi
    \coloneq
    \int_{\mathsf U(BE)}
        f\!\left(
            (\id_A\otimes U_{BE})\hat\phi_0
            (\id_A\otimes U_{BE}^{\dagger})
        \right)\mathrm dU.
    \label{eq:fixed-marginal-haar-measure}
\end{equation}
By the right invariance of Haar measure, this definition is independent of the choice of reference purification \(\hat\phi_0\).
This construction is inspired by \cite{nahar2024postselection} whose authors considered de Finetti reduction with fixed marginal, but we apply it to de Finetti representation.

We first require a constrained analogue of \zcref{eq:haar-moment-symmetric-subspace}, in which the integration should be taken with respect to \(\mathrm d\hat\phi\) over \(\mathcal P_{ABE}(\tau_A)\). The following lemma provides precisely this relation as a direct consequence of \zcref{lem:fixed-marginal-haar-moment}.

\begin{lemma}
\label{lem:constrained-haar-identity}
For \(m\in\mathbb N\), let \(P_{\mathrm{sym}}^{(m)}\) be the projection onto
\(\operatorname{Sym}^{m}(ABE)\), and \(K_{A^m}\coloneq
    \left(
        \Tr_{(BE)^m}P_{\mathrm{sym}}^{(m)}
    \right)^{-1}\). Then
\begin{equation}
    \int_{\mathcal P_{ABE}(\tau_A)}
        |\hat\phi\rangle\!\langle\hat\phi|^{\otimes m}\,\mathrm d\hat\phi
    =
    d_A^{-m}
    \bigl(K_{A^m}\otimes\id_{(BE)^m}\bigr)
    P_{\mathrm{sym}}^{(m)}
    =
    d_A^{-m}
    P_{\mathrm{sym}}^{(m)}
    \bigl(K_{A^m}\otimes\id_{(BE)^m}\bigr).
    \label{eq:constrained-haar-identity}
\end{equation}
\end{lemma}

\begin{proof}
Apply \zcref{lem:fixed-marginal-haar-moment} with \(R=BE\). Since any \(\hat\phi_0\in\mathcal P_{ABE}(\tau_A)\) satisfies
\(\hat\phi_{0,A}=d_A^{-1}\id_A\), it gives
\begin{equation}
    \int_{\mathcal P_{ABE}(\tau_A)}
        \hat\phi^{\otimes m}\,\mathrm d\hat\phi
    =\int_{\mathsf U(BE)}
        \left(
            (\id_A\otimes U_{BE})\hat\phi_0
            (\id_A\otimes U_{BE}^{\dagger})
        \right)^{\otimes m}\mathrm dU
    =
    d_A^{-m}
    \bigl(
        K_{A^m}\otimes\id_{(BE)^m}
    \bigr)
    P_{\mathrm{sym}}^{(m)}.
\end{equation}
The commutation
\(
[K_{A^m}\otimes\id_{(BE)^m},P_{\mathrm{sym}}^{(m)}]=0
\) was also established in \zcref{lem:fixed-marginal-haar-moment}.
\end{proof}

The key difference from \zcref{eq:haar-moment-symmetric-subspace} is that the coefficient in front of \(P_{\mathrm{sym}}^{(m)}\) is no longer a scalar but an operator \(K_{A^m}\otimes\id_{(BE)^m}\). Similarly, the scalar coefficients \(p_j\) appearing in the state de Finetti argument are replaced here by operator coefficients on \(A^k\). We now define and normalize these coefficients.

For any fixed subset \(R\subseteq[n]\), let \(P_{\mathrm{sym}}^{(R)}\) be the projection onto the symmetric subspace of \(\bigotimes_{i\in R} A_iB_iE_i\), and denote \(K_{A^R}=
    \left(
        \Tr_{(BE)^R}P_{\mathrm{sym}}^{(R)}
    \right)^{-1}\in\mathcal{L}(A^R)\). When it is clear which systems we are considering, we write \(P_{\mathrm{sym}}^{(\ell)}\) and \(K_{A^\ell}\) for \(|R|=\ell\). By \zcref{lem:constrained-haar-identity}, 
\begin{equation}
    \int_{\mathcal P_{ABE}(\tau_A)}
        \hat\phi^{\otimes R}\,\mathrm d\hat\phi
    =
    d_A^{-|R|}
    \bigl(
        K_{A^R}\otimes\id_{(BE)^R}
    \bigr)
    P_{\mathrm{sym}}^{(R)}.
    \label{eq:constrained-haar-subset}
\end{equation}

For \(0\leq k\leq n\) and an arbitrary fixed pure state
\(\hat\psi=|\hat\psi\rangle\!\langle\hat\psi|
\in\mathcal P_{ABE}(\tau_A)\), let
\(F_{\hat\phi}\coloneq
|\langle\hat\phi|\hat\psi\rangle|^2\). Define
\begin{equation}
\begin{aligned}
    a_{nk}
    \coloneq
    \int_{\mathcal P_{ABE}(\tau_A)}
        F_{\hat\phi}^{\,n-k}\,\mathrm d\hat\phi                                 
    &=
    \Tr\left[
        \hat\psi^{\otimes(n-k)}
        \int_{\mathcal P_{ABE}(\tau_A)}
            \hat\phi^{\otimes(n-k)}\,\mathrm d\hat\phi
    \right]                                                                             \\
    &=
    d_A^{-(n-k)}
    \Tr\left[
        K_{A^{n-k}}\hat\psi_A^{\otimes(n-k)}
    \right]
    =
    d_A^{-2(n-k)}\Tr K_{A^{n-k}}.
\end{aligned}
\label{eq:ank-definition}
\end{equation}
The last equality follows from \(\hat\psi_A=\tau_A\). In particular, \(a_{nk}>0\) is a constant independent of \(\hat\psi\). For \(0\leq j\leq k\) and an arbitrary fixed pure state
\(\hat\psi\in\mathcal P_{ABE}(\tau_A)\), define \(L_j\in\mathcal L(A^k)\) by
\begin{equation}
    L_j
    \coloneq
    \frac{d_A^k}{a_{nk}}
    \int_{\mathcal P_{ABE}(\tau_A)}
        F_{\hat\phi}^{\,n-k}
        \Tr_{(BE)^k}
        \left[
            \Pi_{\hat\phi,j}^{(k)}
            \hat\psi^{\otimes k}
        \right]
        \mathrm d\hat\phi.
    \label{eq:Lj-definition}
\end{equation}
Recall (from the proof of \zcref{lem:haar-defect-coefficient}) that
\begin{equation}
    \Pi_{\hat\phi,j}^{(k)}
    =
    \sum_{\substack{S\subseteq[k]\\|S|=j}}
    \ \sum_{T\subseteq S}
        (-1)^{|T|}
        \hat\phi^{\otimes(([k]\setminus S)\cup T)}
        \otimes\id_{(ABE)^{S\setminus T}}.
    \label{eq:Pi-inclusion-exclusion-channel}
\end{equation}
For each \(S\subseteq[k]\) and \(T\subseteq S\), let \(R_{S,T}\coloneq([n]\setminus S)\cup T\) with \(|R_{S,T}|=n-j+|T|\). Substituting \zcref{eq:Pi-inclusion-exclusion-channel} into
\zcref{eq:Lj-definition}, applying \zcref{eq:constrained-haar-subset} to the copies indexed by \(R_{S,T}\), therefore gives an alternative expression for \(L_j\),
\begin{equation}
    L_j
    = \frac{1}{a_{nk}}
    \sum_{\substack{S\subseteq[k]\\|S|=j}}
    \ \sum_{T\subseteq S}
        (-1)^{|T|}
        d_A^{-(2n-k-j+|T|)}                                    
    \Tr_{A_{k+1}\cdots A_n}
        \left[
            K_{A^{R_{S,T}}}
            \otimes\id_{A^{S\setminus T}}
        \right].
\label{eq:Lj-K-expansion}
\end{equation}
This expression contains no reference to \(\hat\psi\), proving that \(L_j\)
is independent of its choice. It also shows that \(L_j\) is Hermitian, since
each \(K_{A^{R_{S,T}}}\) is positive. The operators \(L_j\) sum to identity, analogous to the scalars \(p_j\) summing to one. Indeed,
using \(\sum_{j=0}^k\Pi_{\hat\phi,j}^{(k)}=\id_{(ABE)^k}\),
\begin{equation}
\begin{aligned}
    \sum_{j=0}^{k}L_j
    &=
    \frac{d_A^k}{a_{nk}}
    \int_{\mathcal P_{ABE}(\tau_A)}
        F_{\hat\phi}^{\,n-k}
        \Tr_{(BE)^k}\hat\psi^{\otimes k}
        \,\mathrm d\hat\phi                                      \\
    &=
    \frac{d_A^k}{a_{nk}}
    \int_{\mathcal P_{ABE}(\tau_A)}
        F_{\hat\phi}^{\,n-k}
        \tau_A^{\otimes k}\,\mathrm d\hat\phi
    =
    \id_{A^k}.
\end{aligned}
\label{eq:Lj-normalization}
\end{equation}
We also compute the trace of \(L_j\) for future reference. Using \(\Tr(\hat\phi\hat\psi)=F_{\hat\phi}\),
\begin{equation}
    \Tr\left[
        \Pi_{\hat\phi,j}^{(k)}
        \hat\psi^{\otimes k}
    \right]
    =
    \binom{k}{j}
    F_{\hat\phi}^{\,k-j}
    \left(1-F_{\hat\phi}\right)^j.
\end{equation}
Taking the trace in \zcref{eq:Lj-definition} therefore gives
\begin{equation}
    \Tr L_j
    =
    \frac{d_A^k}{a_{nk}}
    \binom{k}{j}
    \int_{\mathcal P_{ABE}(\tau_A)}
        F_{\hat\phi}^{\,n-j}
        \left(1-F_{\hat\phi}\right)^j
        \,\mathrm d\hat\phi.
    \label{eq:trace-Lj}
\end{equation}
In particular, summing over \(j\) and performing the binomial sum gives
\begin{equation}
    \sum_{j=0}^{k}\Tr L_j
    =
    \frac{d_A^k}{a_{nk}}
    \int_{\mathcal P_{ABE}(\tau_A)}
        F_{\hat\phi}^{\,n-k}
        \left(
            F_{\hat\phi}+1-F_{\hat\phi}
        \right)^k
        \,\mathrm d\hat\phi
    =
    d_A^k,
\end{equation}
consistent with the normalization of \(L_j\).

For every \(\hat\phi\in\mathcal P_{ABE}(\tau_A)\), we define the (unnormalized) postselected operator
\begin{equation}
    \Psi_{\hat\phi}^{(k)}
    \coloneq
    \bigl(
        \id_{(ABE)^k}\otimes
        \langle\hat\phi|^{\otimes(n-k)}
    \bigr)
    \Psi^{(n)}
    \bigl(
        \id_{(ABE)^k}\otimes
        |\hat\phi\rangle^{\otimes(n-k)}
    \bigr)\in\mathcal{L}\left((ABE)^{\otimes k}\right),
    \label{eq:channel-conditional-state}
\end{equation}
supported on $\operatorname{Sym}^{k}(ABE)$, see the comment below \zcref{eq:abbreviation-contraction}.

The following lemma is the constrained analogue of \zcref{lem:contracted-defect-coefficient} and provides the key ingredient needed to extend the state de Finetti argument to channels.

\begin{lemma}
\label{lem:channel-projection-identities}
For every \(0\leq j\leq k\leq n\), let \(a_{nk}>0\) and \(L_j\in\mathcal{L}(A^k)\) be defined as
in \zcref{eq:ank-definition} and \zcref{eq:Lj-definition}, respectively. Then
\begin{align}
    \frac{1}{a_{nk}}
    \int_{\mathcal P_{ABE}(\tau_A)}
        \Tr_{E^k}
        \left[
            \Pi_{\hat\phi,j}^{(k)}
            \Psi_{\hat\phi}^{(k)}
        \right]
        \mathrm d\hat\phi
    &=
    \bigl(L_j\otimes\id_{B^k}\bigr)J^{(k)},
    \label{eq:channel-projected-left}
    \\
    \frac{1}{a_{nk}}
    \int_{\mathcal P_{ABE}(\tau_A)}
        \Tr_{E^k}
        \left[
            \Psi_{\hat\phi}^{(k)}
            \Pi_{\hat\phi,j}^{(k)}
        \right]
        \mathrm d\hat\phi
    &=
    J^{(k)}\bigl(L_j\otimes\id_{B^k}\bigr),
    \label{eq:channel-projected-right}
    \\
    \frac{1}{a_{nk}}
    \int_{\mathcal P_{ABE}(\tau_A)}
        \Tr_{E^k}\Psi_{\hat\phi}^{(k)}
        \,\mathrm d\hat\phi
    &=
    J^{(k)}.
    \label{eq:channel-conditional-resolution}
\end{align}
\end{lemma}

\begin{proof}
All integrals below are over \(\mathcal P_{ABE}(\tau_A)\), we will hence omit it. For
\(S\subseteq[k]\) with \(|S|=j\) and \(T\subseteq S\), denote
\(R_{S,T}=([n]\setminus S)\cup T\). By \zcref{eq:Pi-inclusion-exclusion-channel} and \zcref{eq:channel-conditional-state},
\begin{equation}
\begin{aligned}
    &\int
        \Tr_{E^k}
        \left[
            \Pi_{\hat\phi,j}^{(k)}
            \Psi_{\hat\phi}^{(k)}
        \right]
        \mathrm d\hat\phi                                      \\
    &=
    \int
        \Tr_{E^k}
        \left[
            \Pi_{\hat\phi,j}^{(k)}
            \Tr_{(ABE)_{k+1}\cdots(ABE)_n}
            \left[
                \left(
                    \id_{(ABE)^k}
                    \otimes\hat\phi^{\otimes(n-k)}
                \right)
                \Psi^{(n)}
            \right]
        \right]
        \mathrm d\hat\phi                                      \\
    &=\sum_{\substack{S\subseteq[k]\\|S|=j}}
    \ \sum_{T\subseteq S}
        (-1)^{|T|}
        \Tr_{E^k(ABE)_{k+1}\cdots(ABE)_n}
        \left[
            \left(
                \int\hat\phi^{\otimes R_{S,T}}
                \mathrm d\hat\phi
                \otimes\id_{(ABE)^{S\setminus T}}
            \right)
            \Psi^{(n)}
        \right].
\end{aligned}
\end{equation}
Applying \zcref{eq:constrained-haar-subset} to the copies indexed by \(R_{S,T}\), the resulting projection \(P_{\mathrm{sym}}^{(R_{S,T})}\) may be removed because
\(\Psi^{(n)}\) is supported on \(\operatorname{Sym}^{n}(ABE)\). Hence
\begin{equation}
\begin{aligned}
    &\int
        \Tr_{E^k}
        \left[
            \Pi_{\hat\phi,j}^{(k)}
            \Psi_{\hat\phi}^{(k)}
        \right]
        \mathrm d\hat\phi                                      \\
    &=
    \sum_{\substack{S\subseteq[k]\\|S|=j}}
    \ \sum_{T\subseteq S}
        (-1)^{|T|}
        d_A^{-|R_{S,T}|}
        \Tr_{E^k(ABE)_{k+1}\cdots(ABE)_n}
        \left[
            \left(
                K_{A^{R_{S,T}}}
                \otimes\id_{A^{S\setminus T}(BE)^n}
            \right)
            \Psi^{(n)}
        \right]                                                 \\
    &=
    \sum_{\substack{S\subseteq[k]\\|S|=j}}
    \ \sum_{T\subseteq S}
        (-1)^{|T|}
        d_A^{-|R_{S,T}|}
        \Tr_{A_{k+1}\cdots A_nB_{k+1}\cdots B_n}
        \left[
            \left(
                K_{A^{R_{S,T}}}
                \otimes\id_{A^{S\setminus T}B^n}
            \right)
            J^{(n)}
        \right]                                                 \\
    &=
    \sum_{\substack{S\subseteq[k]\\|S|=j}}
    \ \sum_{T\subseteq S}
        (-1)^{|T|}
        d_A^{-(|R_{S,T}|+n-k)}
        \left(
            \Tr_{A_{k+1}\cdots A_n}
            \left[
                K_{A^{R_{S,T}}}
                \otimes\id_{A^{S\setminus T}}
            \right]
            \otimes\id_{B^k}
        \right)
        J^{(k)}                                                  \\
    &=
    a_{nk}
    \left(L_j\otimes\id_{B^k}\right)J^{(k)}.
\end{aligned}
\label{eq:projected-conditional-integral}
\end{equation}
The second equality uses \(\Tr_{E^n}\Psi^{(n)}=J^{(n)}\), and the third
uses the no-signalling requirement,
\(\Tr_{B_{k+1}\cdots B_n}J^{(n)}
=\tau_A^{\otimes(n-k)}\otimes J^{(k)}\), from \zcref{eq:choi-no-signalling-subsets}. The last equality follows from
\(|R_{S,T}|=n-j+|T|\) and the expression for \(L_j\) in \zcref{eq:Lj-K-expansion}. Dividing by \(a_{nk}\) proves \zcref{eq:channel-projected-left}.

Since \(L_j\) is Hermitian, taking the adjoint of \zcref{eq:channel-projected-left} gives \zcref{eq:channel-projected-right}. Finally, summing \zcref{eq:channel-projected-left} over \(j\), using
\(\sum_{j=0}^{k}\Pi_{\hat\phi,j}^{(k)}=\id_{(ABE)^k}\) and
\(\sum_{j=0}^{k}L_j=\id_{A^k}\) proves \zcref{eq:channel-conditional-resolution}.
\end{proof}

Many of the technical difficulties with the proof of the exponential channel de Finetti comes from the fact that it is much harder to deal with the operators \(L_j\) than the scalars \(p_j\). We have an explicit expression for \(p_j\), but only an integral expression \zcref{eq:Lj-definition} or a double summation expression \zcref{eq:Lj-K-expansion} for \(L_j\). While we don't have an easy expression of these operators\footnote{It is possible that a tighter exponential de Finetti theorem can be obtained by a more careful study of the operators \(L_j\), just like the explicit expression for \(p_j\) is crucial for refining the state de Finetti theorem.}, we can still bound their operator norm, which will be useful for error estimate later. The proof of the following lemma can be found in \zcref{App:Proof:lem:ank-Lj-bounds}.

\begin{lemma}
\label{lem:ank-Lj-bounds}
For every \(0\leq j\leq k\leq n\), let \(a_{nk}>0\) and \(L_j\in\mathcal{L}(A^k)\) be defined as
in \zcref{eq:ank-definition} and \zcref{eq:Lj-definition}, respectively, and recall that \(d_{n-k}=\dim\operatorname{Sym}^{n-k}(ABE)\). Then
\begin{equation}
    \frac{1}{d_{n-k}}
    \leq a_{nk}\leq 1,
    \label{eq:ank-bounds}
\end{equation}
and
\begin{equation}
    \lVert L_j\rVert_\infty
    \leq
    \binom{k}{j}d_A^{j/2}
    \left(
        \frac{1}{a_{nk}}
        \int_{\mathcal P_{ABE}(\tau_A)}
            F_{\hat\phi}^{\,n-j}
            (1-F_{\hat\phi})^j
            \,\mathrm d\hat\phi
    \right)^{1/2}.
    \label{eq:Lj-operator-norm-bound}
\end{equation}
\end{lemma}

We will now explicitly construct the convex mixture of almost-iid Choi states that are used to approximate the Choi state \(J^{(k)}\) of the reduced channel. Before presenting the full proof, we briefly explain the difficulty and how we address them. The construction closely follows the state analogue, in which one considers the projected operators
\begin{equation}
   P_{\phi}^{(k,r)}\Psi_{\phi}^{(k)}P_{\phi}^{(k,r)}, 
\end{equation}
see, e.g., \zcref{eq:state-prob-measure}. The natural modification in the channel setting would be to replace \(\phi\) by \(\hat{\phi}\) and consider
\begin{equation}
    P_{\hat\phi}^{(k,r)}\Psi_{\hat\phi}^{(k)}P_{\hat\phi}^{(k,r)},
\end{equation}
so that the reference states are by construction Choi, representing reference channels. 

As we have stressed at the beginning of this section, however, although \(\hat\phi_{AB}\) is a Choi state, a state that is almost-iid along a Choi state need not itself
satisfy the Choi constraint, since the defects will contaminate the input marginal. 

To deal with this, we apply an additional Haar twirling \(V\) on the purifying system \(E\) and define
\begin{equation}
    \omega_{\hat\phi}^{(k)} \sim \int_{\mathsf U(E)}
    \bigl(\id_{AB}\otimes V_E^\dagger\bigr)^{\otimes k}
    \Psi_{(\id_{AB}\otimes V_E)\hat\phi(\id_{AB}\otimes V^\dagger_E)}^{(k)}
    \bigl(\id_{AB}\otimes V_E\bigr)^{\otimes k}
    \,\mathrm dV.
\end{equation}
This is to ensure \(\omega_{\hat\phi}^{(k)}\) satisfies the Choi constraint while \(\Psi_{\hat\phi}^{(k)}\) does not. We then consider
\begin{equation}
    P_{\hat\phi}^{(k,r)}\omega_{\hat\phi}^{(k)}P_{\hat\phi}^{(k,r)}.
\end{equation}
The projection will inevitably disturb the input marginal, causing a deviation from the Choi constraint. But since \(\omega_{\hat\phi}^{(k)}\) is enforced to be Choi, this deviation can be characterized by how much the projection disturbs the state:
\begin{equation}
        \Tr_{(BE)^k}\!\left(
            P_{\hat\phi}^{(k,r)}
            \omega_{\hat\phi}^{(k)}
            P_{\hat\phi}^{(k,r)}
        \right)
        -\tau_A^{\otimes k}
    =
        \Tr_{(BE)^k}\!\left(
            P_{\hat\phi}^{(k,r)}
            \omega_{\hat\phi}^{(k)}
            P_{\hat\phi}^{(k,r)}
            -\omega_{\hat\phi}^{(k)}
        \right).
\end{equation}
We expect this disturbance to be exponentially small, for the same reason that the almost-iid projection causes only exponentially small disturbance in the state case, see, e.g., \zcref{eq:state-cross-term}. If the deviation were exactly zero, then we are already done. However, if the deviation is small, we can still repair it by paying an exponentially vanishing cost. We will now formalize the above intuition into the following two propositions.

\begin{proposition}
\label{prop:projected-channel-df}
Fix integers \(0\leq r\leq k\leq n\), and let \(J^{(n)}\) be the Choi state of a permutation-covariant and no-signalling channel \(\mathcal N^{(n)}_{A'^n\to B^n}\in\mathrm{CPTP}(A'^n,B^n)\). Its reduced channel from \(A'^k\) to \(B^k\) is defined via \zcref{eq:reduced-channel}, with Choi state \(J^{(k)}\). With \(a_{nk}\) and \(\{L_j\}_{j=0}^k\) defined
in \zcref{eq:ank-definition} and \zcref{eq:Lj-definition}, respectively, let
\begin{equation}
    L_{>r}\coloneq\sum_{j=r+1}^k L_j,
    \qquad
    \eta_{nkr}
    \coloneq
    \left(
        1+2d_k\lVert L_{>r}\rVert_\infty
    \right)^{-1},
    \label{eq:eta-nkr}
\end{equation}
where \(d_k=\dim\operatorname{Sym}^k(ABE)\), with the convention \(L_{>k}=0\). Then there exist a probability measure \(\mu_{nkr}\) on \(\mathcal P_{ABE}(\tau_A)\) and, for every \(\hat\phi\in\mathcal P_{ABE}(\tau_A)\), a permutation-invariant state
\(\omega_{\hat\phi}^{(k)}\) on \((ABE)^k\) satisfying \(\omega_{\hat\phi,A^k}^{(k)}=\tau_A^{\otimes k}\), such that
\begin{equation}
    \eta_{nkr}J^{(k)}
    \leq
    \int_{\mathcal P_{ABE}(\tau_A)}
        \Tr_{E^k}\!\left[
            P_{\hat\phi}^{(k,r)}
            \omega_{\hat\phi}^{(k)}
            P_{\hat\phi}^{(k,r)}
        \right]
        \mathrm d\mu_{nkr}(\hat\phi),
    \label{eq:projected-channel-df}
\end{equation}
and
\begin{equation}
    \int_{\mathcal P_{ABE}(\tau_A)}
        \Tr\!\left[
            \bigl(\id_{(ABE)^k}-P_{\hat\phi}^{(k,r)}\bigr)
            \omega_{\hat\phi}^{(k)}
        \right]
        \mathrm d\mu_{nkr}(\hat\phi)
        \leq
        \eta_{nkr}d_{n-k}
        \left(\frac{k}{n}\right)^{r+1}.
    \label{eq:projected-channel-tail}
\end{equation}
\end{proposition}

\begin{proof}
Recall that \(\lvert\Psi^{(n)}\rangle\) is a symmetric purification of \(J^{(n)}\), and \(\Psi^{(k)}=
    \Tr_{(ABE)_{k+1}\cdots(ABE)_n}\Psi^{(n)}\) with \(\Tr_{E^k}\Psi^{(k)}=J^{(k)}\). For every \(\hat\phi\in\mathcal P_{ABE}(\tau_A)\), let \(\Psi_{\hat\phi}^{(k)}=
    \langle\hat\phi|^{\otimes(n-k)}
    \Psi^{(n)}
    |\hat\phi\rangle^{\otimes(n-k)}\).

\textbf{Step 1: Construction of Choi-compatible operators.}
We apply a twirling on \(E\) so that Choi constraint is enforced before the almost-iid projection. For
\(V\in\mathsf U(E)\), let
\begin{equation}
    |\hat\phi_V\rangle
    \coloneq
    (\id_{AB}\otimes V_E)|\hat\phi\rangle
\end{equation}
and define
\begin{equation}
    \widetilde\omega_{\hat\phi}^{(k)}
    \coloneq
    \int_{\mathsf U(E)}
        (\id_{AB}\otimes V_E^\dagger)^{\otimes k}
        \Psi_{\hat\phi_V}^{(k)}
        (\id_{AB}\otimes V_E)^{\otimes k}
        \,\mathrm dV.
\end{equation}
It is easy to see each \(\widetilde\omega_{\hat\phi}^{(k)}\) is permutation invariant. We now verify its input marginal. For any fixed \(\hat\phi\), \zcref{lem:fixed-marginal-haar-moment} applied with \(A\to AB\) and \(R\to E\) gives
\begin{equation}
    \int_{\mathsf U(E)}
        \hat\phi_V^{\otimes(n-k)}
        \,\mathrm dV
    =
    \left(
    \hat\phi_{AB}^{\otimes(n-k)}
        K_{(AB)^{n-k}}
        \otimes\id_{E^{n-k}}
    \right)
    P_{\mathrm{sym}}^{(n-k)}, \quad K_{(AB)^{n-k}}
    =
    \left(
        \Tr_{E^{n-k}}P_{\mathrm{sym}}^{(n-k)}
    \right)^{-1}.
\end{equation}
Using \(\Tr_{E^n}\Psi^{(n)}=J^{(n)}\) and the no-signalling identity
\(\Tr_{B^k}J^{(n)}
=\tau_A^{\otimes k}\otimes J^{(n-k)}\) gives
\begin{equation}
\begin{aligned}
    \widetilde\omega_{\hat\phi,A^k}^{(k)}
    &=\int_{\mathsf U(E)}\Tr_{(BE)^k} \Tr_{(ABE)^{n-k}}
    \left[
        \id_{(ABE)^k}\otimes|\hat\phi_V\rangle\!\langle\hat\phi_V|^{\otimes(n-k)}\Psi^{(n)}
    \right] \,\mathrm dV
    \\
    &=
    \Tr_{(BE)^k(ABE)^{n-k}}
    \left[
        \left(
            \id_{(ABE)^k}
            \otimes
            \hat\phi_{AB}^{\otimes(n-k)}
            K_{(AB)^{n-k}}
            \otimes\id_{E^{n-k}}
        \right)
        \Psi^{(n)}
    \right]
    \\
    &=
    \Tr_{(AB)^{n-k}}\Tr_{B^k}
    \left[
        \left(
            \id_{A^kB^k}
            \otimes
            \hat\phi_{AB}^{\otimes(n-k)}
            K_{(AB)^{n-k}}
        \right)
        J^{(n)}
    \right]
    \\
    &=
    \Tr\!\left[
    \hat\phi_{AB}^{\otimes(n-k)}
        K_{(AB)^{n-k}}
        J^{(n-k)}
    \right]
    \tau_A^{\otimes k}
\equiv
    \Tr\!\left[\widetilde\omega_{\hat\phi}^{(k)}\right]
    \tau_A^{\otimes k}.
\end{aligned}
\end{equation}
Thus, up to normalization, 
\(\widetilde\omega_{\hat\phi}^{(k)}\) satisfies the Choi constraint.

\textbf{Step 2: Projection identities after twirling.}
We next show that the additional twirl does not change the projection
identities established in \zcref{lem:channel-projection-identities}. Indeed,
\begin{equation}
    \Pi_{\hat\phi_V,j}^{(k)}
    =
    (\id_{AB}\otimes V_E)^{\otimes k}
    \Pi_{\hat\phi,j}^{(k)}
    (\id_{AB}\otimes V_E^\dagger)^{\otimes k},
\end{equation}
and the measure \(\mathrm d\hat\phi\) is invariant under
\(\hat\phi\mapsto\hat\phi_V\). Hence, for every \(0\leq j\leq k\),
\begin{equation}
\begin{aligned}
    \frac{1}{a_{nk}}
    \int_{\mathcal P_{ABE}(\tau_A)}
        \Tr_{E^k}
        \left[
            \Pi_{\hat\phi,j}^{(k)}
            \widetilde\omega_{\hat\phi}^{(k)}
        \right]
        \,\mathrm d\hat\phi
    &=
    (L_j\otimes\id_{B^k})J^{(k)},                                      \\
    \frac{1}{a_{nk}}
    \int_{\mathcal P_{ABE}(\tau_A)}
        \Tr_{E^k}
        \left[
            \widetilde\omega_{\hat\phi}^{(k)}
            \Pi_{\hat\phi,j}^{(k)}
        \right]
        \,\mathrm d\hat\phi
    &=
    J^{(k)}(L_j\otimes\id_{B^k}),                                      \\
    \frac{1}{a_{nk}}
    \int_{\mathcal P_{ABE}(\tau_A)}
        \Tr_{E^k}\widetilde\omega_{\hat\phi}^{(k)}
        \,\mathrm d\hat\phi
    &=
    J^{(k)}.
\end{aligned}
\end{equation}
Since
\(\id-P_{\hat\phi}^{(k,r)}
=\sum_{j=r+1}^k\Pi_{\hat\phi,j}^{(k)}\), summing the first two
identities over \(j=r+1,\cdots,k\) gives
\begin{equation}
\label{eq:channel-projection-identities-twirl}
\begin{aligned}
    \frac{1}{a_{nk}}
    \int_{\mathcal P_{ABE}(\tau_A)}
        \Tr_{E^k}
        \left[
            \bigl(\id_{(ABE)^k}-P_{\hat\phi}^{(k,r)}\bigr)
            \widetilde\omega_{\hat\phi}^{(k)}
        \right]
        \,\mathrm d\hat\phi
    &=
    (L_{>r}\otimes\id_{B^k})J^{(k)},                                   \\
    \frac{1}{a_{nk}}
    \int_{\mathcal P_{ABE}(\tau_A)}
        \Tr_{E^k}
        \left[
            \widetilde\omega_{\hat\phi}^{(k)}
            \bigl(\id_{(ABE)^k}-P_{\hat\phi}^{(k,r)}\bigr)
        \right]
        \,\mathrm d\hat\phi
    &=
    J^{(k)}(L_{>r}\otimes\id_{B^k}).
\end{aligned}
\end{equation}
Taking the trace of the third identity yields
\begin{equation}
\label{eq:channel-identity-twirl}
    \int_{\mathcal P_{ABE}(\tau_A)}
        \Tr\widetilde\omega_{\hat\phi}^{(k)}
        \,\mathrm d\hat\phi
    =
    a_{nk}.
\end{equation}

\textbf{Step 3: Controlling the cross terms.}
Consider the operator identity
\begin{equation}
\label{eq:op-identity}
\begin{aligned}
    P_{\hat\phi}^{(k,r)}
    \widetilde\omega_{\hat\phi}^{(k)}
    P_{\hat\phi}^{(k,r)}
    &=
    \widetilde\omega_{\hat\phi}^{(k)}
    -
    \bigl(\id-P_{\hat\phi}^{(k,r)}\bigr)
    \widetilde\omega_{\hat\phi}^{(k)}
    -
    \widetilde\omega_{\hat\phi}^{(k)}
    \bigl(\id-P_{\hat\phi}^{(k,r)}\bigr)
    +
    \bigl(\id-P_{\hat\phi}^{(k,r)}\bigr)
    \widetilde\omega_{\hat\phi}^{(k)}
    \bigl(\id-P_{\hat\phi}^{(k,r)}\bigr)
    \\
    &\geq
    \widetilde\omega_{\hat\phi}^{(k)}
    -
    \bigl(\id-P_{\hat\phi}^{(k,r)}\bigr)
    \widetilde\omega_{\hat\phi}^{(k)}
    -
    \widetilde\omega_{\hat\phi}^{(k)}
    \bigl(\id-P_{\hat\phi}^{(k,r)}\bigr),
\end{aligned}
\end{equation}
where in the second line we have thrown away the last positive term. Consequently, taking the partial trace and the integral, using \zcref{eq:channel-projection-identities-twirl},
\begin{equation}
    \frac{1}{a_{nk}}
    \int_{\mathcal P_{ABE}(\tau_A)}
        \Tr_{E^k}
        \left[
            P_{\hat\phi}^{(k,r)}
            \widetilde\omega_{\hat\phi}^{(k)}
            P_{\hat\phi}^{(k,r)}
        \right]
        \,\mathrm d\hat\phi
    \geq\;
    J^{(k)}
    -
    (L_{>r}\otimes\id_{B^k})J^{(k)}
    -
    J^{(k)}(L_{>r}\otimes\id_{B^k}).
\end{equation}
It remains to dominate the last two cross terms by a mixture of Choi states. We show in \zcref{App:Proof:cross-term-domination} that
\begin{equation}
\label{eq:cross-term-domination}
    (L_{>r}\otimes\id_{B^k})J^{(k)}
    +
    J^{(k)}(L_{>r}\otimes\id_{B^k})
    \leq
    2d_k\|L_{>r}\|_\infty
    \int_{\mathcal P_{ABE}(\tau_A)}
        \hat\phi_{AB}^{\otimes k}\,\mathrm d\hat\phi.
\end{equation}
Together this yields
\begin{equation}
\label{eq:constrained-op-ineq}
   a_{nk}J^{(k)}
   \leq \int_{\mathcal P_{ABE}(\tau_A)}
        \left(\Tr_{E^k}
        \left[
            P_{\hat\phi}^{(k,r)}
            \widetilde\omega_{\hat\phi}^{(k)}
            P_{\hat\phi}^{(k,r)}
        \right]+2a_{nk}d_k\|L_{>r}\|_\infty \hat\phi_{AB}^{\otimes k}
        \right)
        \,\mathrm d\hat\phi.
\end{equation}

\textbf{Step 4: State and measure normalization.}
Define normalized states and probability measure:
\begin{equation}
    \omega_{\hat\phi}^{(k)}
    \coloneq
    \frac{
        \widetilde\omega_{\hat\phi}^{(k)}
        +
        2a_{nk}d_k\|L_{>r}\|_\infty
        \hat\phi^{\otimes k}
    }{
        \Tr\widetilde\omega_{\hat\phi}^{(k)}
        +
        2a_{nk}d_k\|L_{>r}\|_\infty
    }, \quad 
    \mathrm d\mu_{nkr}(\hat\phi)
    \coloneq
    \frac{
        \Tr\widetilde\omega_{\hat\phi}^{(k)}
        +
        2a_{nk}d_k\|L_{>r}\|_\infty
    }{
        a_{nk}
        \bigl(1+2d_k\|L_{>r}\|_\infty\bigr)
    }
    \,\mathrm d\hat\phi.
\end{equation}
The fact that \(\mu_{nkr}\) is a probability measure follows from \zcref{eq:channel-identity-twirl}. Moreover, each \(\omega_{\hat\phi}^{(k)}\) is permutation invariant
and satisfies \(\omega_{\hat\phi,A^k}^{(k)}=\tau_A^{\otimes k}\).

Since the iid state is trivially also almost-iid, i.e., \(P_{\hat\phi}^{(k,r)}\hat\phi^{\otimes k}
P_{\hat\phi}^{(k,r)}=\hat\phi^{\otimes k}\), we obtain
\begin{equation}
\begin{aligned}
    &\int_{\mathcal P_{ABE}(\tau_A)}
        \Tr_{E^k}
        \left[
            P_{\hat\phi}^{(k,r)}
            \omega_{\hat\phi}^{(k)}
            P_{\hat\phi}^{(k,r)}
        \right]
        \,\mathrm d\mu_{nkr}(\hat\phi)
    \\
    &\quad=
    \frac{1}{
        a_{nk}\bigl(1+2d_k\|L_{>r}\|_\infty\bigr)
    }
    \int_{\mathcal P_{ABE}(\tau_A)}
        \Tr_{E^k}
        \left[
            P_{\hat\phi}^{(k,r)}
            \left(
                \widetilde\omega_{\hat\phi}^{(k)}
                +
                2a_{nk}d_k\|L_{>r}\|_\infty
                \hat\phi^{\otimes k}
            \right)
            P_{\hat\phi}^{(k,r)}
        \right]
        \,\mathrm d\hat\phi
    \\
    &\quad\geq
    \frac{1}{1+2d_k\|L_{>r}\|_\infty}\,J^{(k)}
    =
    \eta_{nkr}J^{(k)},
\end{aligned}
\end{equation}
where the inequality follows from \zcref{eq:constrained-op-ineq}, proving the first statement in the proposition. This is analogous to the state de Finetti operator inequality.

\textbf{Step 5: Weight outside the almost-iid subspace.}
Finally, we estimate the average mass discarded by the projection:
\begin{equation}
\begin{aligned}
    &\qquad\int_{\mathcal P_{ABE}(\tau_A)}
        \Tr\!\left[
            \bigl(\id_{(ABE)^k}-P_{\hat\phi}^{(k,r)}\bigr)
            \omega_{\hat\phi}^{(k)}
        \right]
        \,\mathrm d\mu_{nkr}(\hat\phi)
    \\
    &\quad=
    \frac{1}{
        a_{nk}\bigl(1+2d_k\|L_{>r}\|_\infty\bigr)
    }
    \int_{\mathcal P_{ABE}(\tau_A)}
        \Tr\!\left[
            \bigl(\id_{(ABE)^k}-P_{\hat\phi}^{(k,r)}\bigr)
            \widetilde\omega_{\hat\phi}^{(k)}
        \right]
        \,\mathrm d\hat\phi
    \\
    &\quad=
    \eta_{nkr}
    \Tr\!\left[
        (L_{>r}\otimes\id_{B^k})J^{(k)}
    \right]
    =
    \eta_{nkr}d_A^{-k}\Tr L_{>r}.
\end{aligned}
\end{equation}
The previously derived \zcref{eq:trace-Lj} for \(\Tr L_j\) gives
\begin{equation}
\begin{aligned}
    d_A^{-k}\Tr L_{>r}
    &=
    \frac{1}{a_{nk}}
    \sum_{j=r+1}^{k}
        \binom{k}{j}
        \int_{\mathcal P_{ABE}(\tau_A)}
            F_{\hat\phi}^{\,n-j}
            (1-F_{\hat\phi})^j
            \,\mathrm d\hat\phi
    \\
    &=
    \frac{1}{a_{nk}}
    \sum_{j=r+1}^{k}
        \frac{\binom{k}{j}}{\binom{n}{j}}
        \int_{\mathcal P_{ABE}(\tau_A)}
            \binom{n}{j}F_{\hat\phi}^{\,n-j}
            (1-F_{\hat\phi})^j
            \,\mathrm d\hat\phi
    \\
    &\leq
    \frac{1}{a_{nk}}
    \frac{\binom{k}{r+1}}{\binom{n}{r+1}}\int_{\mathcal P_{ABE}(\tau_A)}
            \sum_{j=0}^{n}\binom{n}{j}F_{\hat\phi}^{\,n-j}
            (1-F_{\hat\phi})^j
            \,\mathrm d\hat\phi
    \\
    &\leq
    d_{n-k}
    \left(\frac{k}{n}\right)^{r+1},
\end{aligned}
\end{equation}
where the first inequality follows from \(\frac{\binom{k}{j}}{\binom{n}{j}}\) being monotonically decreasing as \(j\) increases, and relaxing the sum from \(\sum_{j=r+1}^k\) to \(\sum_{j=0}^n\). The last inequality uses the binomial theorem, \(a_{nk}\geq d_{n-k}^{-1}\) from \zcref{lem:ank-Lj-bounds}, and
\begin{equation}
    \frac{\binom{k}{r+1}}{\binom{n}{r+1}}
    =
    \prod_{i=0}^{r}\frac{k-i}{n-i}
    \leq
    \left(\frac{k}{n}\right)^{r+1}.
\end{equation}
Multiplying by \(\eta_{nkr}\) proves the second assertion.
\end{proof}

As explained before \zcref{prop:projected-channel-df}, it remains to
repair the Choi constraint disturbed by the almost-iid projections, so
that the right-hand side of \zcref{eq:projected-channel-df} becomes a
convex mixture of almost-iid Choi states along Choi reference states. This is
more subtle than the single-copy completion in
\zcref{lem:single-copy-choi-completion}: completing one copy and taking
tensor powers immediately produces an iid Choi state, whereas an
almost-iid state along a Choi reference state need not itself be Choi. The
following proposition instead performs the completion directly within
the multi-copy almost-iid subspace. Crucially, the state before
projection is exactly Choi, so its entire input marginal deviation from maximally mixed is caused by the projection and can be controlled by the discarded weight.

\begin{proposition}[(Choi completion preserving almost-iid structure)]
\label{prop:choi-completion}
Fix integers \(0\leq r\leq k\), let
\(\hat\phi\in\mathcal P_{ABE}(\tau_A)\), and let
\(\omega^{(k)}\) be a permutation-invariant state on \((ABE)^k\)
satisfying \(\omega_{A^k}^{(k)}=\tau_A^{\otimes k}\). Define the discarded mass of \(\omega^{(k)}\) under almost-iid projection along \(\hat\phi\), and a combinatorial factor,
\begin{equation}
    \delta_{kr}
    \coloneq
    \Tr\!\left[
        \bigl(\id_{(ABE)^k}-P_{\hat\phi}^{(k,r)}\bigr)
        \omega^{(k)}
    \right], \qquad N_{kr}\coloneq\sum_{j=0}^{\min\{2r,k\}}
        \binom{k}{j}(d_A^2-1)^j.
\end{equation}
Then there exists a \(\binom{k}{r}\)-almost-iid state \(\widehat\omega^{(k)}\) along \(\hat\phi\) on \((ABE)^k\), satisfying \(\widehat\omega_{A^k}^{(k)}=\tau_A^{\otimes k}\), such that
\begin{equation}
    P_{\hat\phi}^{(k,r)}
    \omega^{(k)}
    P_{\hat\phi}^{(k,r)}
    \leq
    \left(
        1+
        \sqrt{\delta_{kr}
        N_{kr}}
    \right)^2
    \widehat\omega^{(k)}.
\end{equation}
\end{proposition}

\begin{proof}
Let
\(\{F_\alpha\}_{\alpha=0}^{d_A^2-1}\) be an orthogonal unitary
operator basis of \(\mathcal{L}(A)\), for example the discrete Weyl operators \cite[section 4.1.2]{watrous2018theory}, with the following normalization convention
\begin{equation}
    F_0=\id_A,
    \qquad
    \Tr(F_\alpha^\dagger F_\beta)
    =d_A\delta_{\alpha,\beta}.
\end{equation}
For a string
\(\boldsymbol\alpha=(\alpha_1,\ldots,\alpha_k)\), let
\(F_{\boldsymbol\alpha}\coloneq F_{\alpha_1}\otimes\cdots\otimes F_{\alpha_k}\), and let \(\operatorname{wt}(\boldsymbol\alpha)\) denote the weight, namely the number of nonzero entries, of \(\boldsymbol\alpha\).

\textbf{Step 1: Expanding the deviation from Choi-ness.}
By definition, \(P_{\hat\phi}^{(k,r)}\) projects to the subspace spanned by
vectors of the form
\begin{equation}
    |\hat\phi\rangle^{\otimes S^{c}}\otimes |\xi_S\rangle,
    \qquad
    S\subseteq[k],\quad |S|\leq r,\quad
    |\xi_S\rangle\in(ABE)^S.
\end{equation}
Consequently,
\(P_{\hat\phi}^{(k,r)}\omega^{(k)}P_{\hat\phi}^{(k,r)}\)
is a linear combination of operators of the form
\begin{equation}
    \left(
        |\hat\phi\rangle^{\otimes S^{c}}\otimes|\xi_S\rangle
    \right)
    \left(
        \langle\hat\phi|^{\otimes T^{c}}\otimes\langle \eta_T|
    \right),
    \qquad S,T\subseteq[k], \quad |S|,|T|\leq r.
\end{equation}
For every subsystem indexed by \(i\notin S\cup T\), both the ket and the
bra equal \(\hat\phi\). Since \(\hat\phi_A=\tau_A\), tracing out
\((BE)^k\) therefore gives a linear combination of operators of the form
\begin{equation}
    \tau_{A^{(S\cup T)^c}}
    \otimes X_{A^{S\cup T}}, \qquad X_{A^{S\cup T}}\in\mathcal L(A^{S\cup T}).
\end{equation}
Its expansion in the \(\{F_{\boldsymbol\alpha}\}_{\boldsymbol\alpha}\) basis can contain non-identity factors only on the copies in \(S\cup T\), and hence only strings of weight at most
\begin{equation}
    \operatorname{wt}(\boldsymbol\alpha)\leq|S\cup T|
    \leq |S|+|T|
    \leq 2r.
\end{equation}
Thus, combining with the trivial bound \(\operatorname{wt}(\boldsymbol\alpha)\leq k\), we can expand the deviation as
\begin{equation}
\Tr_{(BE)^k}\!\left[
    P_{\hat\phi}^{(k,r)}
    \omega^{(k)}
    P_{\hat\phi}^{(k,r)}
\right]
-\tau_A^{\otimes k}
=
\sum_{\substack{\boldsymbol\alpha:
    \operatorname{wt}(\boldsymbol\alpha)\leq\min\{2r,k\}}}
    h_{\boldsymbol\alpha}F_{\boldsymbol\alpha}
=
\sum_{\substack{\boldsymbol\alpha:
    \operatorname{wt}(\boldsymbol\alpha)\leq\min\{2r,k\}}}
    \overline{h_{\boldsymbol\alpha}}
    F_{\boldsymbol\alpha}^{\dagger},
\label{eq:choi-marginal-deviation-expansion}
\end{equation}
where
\begin{equation}
    h_{\boldsymbol\alpha}
    =
    d_A^{-k}
    \Tr\!\left[
        F_{\boldsymbol\alpha}^{\dagger}
        \left(
            \Tr_{(BE)^k}\!\left[
                P_{\hat\phi}^{(k,r)}
                \omega^{(k)}
                P_{\hat\phi}^{(k,r)}
            \right]
            -\tau_A^{\otimes k}
        \right)
    \right].
\end{equation}

\textbf{Step 2: Complete the marginal to be Choi.}
For every \(\boldsymbol\alpha\) with
\(\operatorname{wt}(\boldsymbol\alpha)\leq\min\{2r,k\}\), partition
the non-identity tensor factors of \(F_{\boldsymbol\alpha}\) into two
sets of size at most \(r\), and write
\begin{equation}
    F_{\boldsymbol\alpha}
    =
    X_{\boldsymbol\alpha}Y_{\boldsymbol\alpha},
    \qquad
    X_{\boldsymbol\alpha},Y_{\boldsymbol\alpha}\in\mathcal L(A^k),
\end{equation}
where \(X_{\boldsymbol\alpha}Y_{\boldsymbol\alpha}\) is matrix multiplication, not tensor product, and both \(X_{\boldsymbol\alpha}\) and \(Y_{\boldsymbol\alpha}\) belong to the product operator basis containing at most \(r\) non-identity factors\footnote{This can always be done and the choice is not unique. As a toy example, let \(k=5, r=2\), and consider \(F_{\boldsymbol\alpha}=\id\otimes F_1\otimes F_1\otimes F_1\otimes F_1\), one may choose \(X_{\boldsymbol\alpha}=\id\otimes F_1\otimes F_1\otimes \id\otimes\id\) and \(Y_{\boldsymbol\alpha}=\id\otimes\id\otimes\id\otimes F_1\otimes F_1\).}.

Define
\begin{equation}
    |x_{\boldsymbol\alpha}\rangle
    \coloneq
    \bigl(X_{\boldsymbol\alpha}\otimes\id_{(BE)^k}\bigr)
    |\hat\phi\rangle^{\otimes k},\qquad
    |y_{\boldsymbol\alpha}\rangle
    \coloneq
    \bigl(Y_{\boldsymbol\alpha}^\dagger\otimes\id_{(BE)^k}\bigr)
    |\hat\phi\rangle^{\otimes k}.
\end{equation}
Since \(X_{\boldsymbol\alpha}\) and \(Y_{\boldsymbol\alpha}\) act nontrivially on at most \(r\)
copies, both \(|x_{\boldsymbol\alpha}\rangle\) and \(|y_{\boldsymbol\alpha}\rangle\) differ from
\(|\hat\phi\rangle^{\otimes k}\) on at most \(r\) copies. Hence
\begin{equation}
    |x_{\boldsymbol\alpha}\rangle,|y_{\boldsymbol\alpha}\rangle
    \in\operatorname{ran}P_{\hat\phi}^{(k,r)}.
\end{equation}
Moreover, since \(\hat\phi_A=\tau_A\) and \(X_{\boldsymbol\alpha},Y_{\boldsymbol\alpha}\) are
unitary,
\begin{equation}
\begin{aligned}
    \Tr_{(BE)^k}|x_{\boldsymbol\alpha}\rangle\!\langle x_{\boldsymbol\alpha}|
    &=
    X_{\boldsymbol\alpha}\tau_A^{\otimes k}X_{\boldsymbol\alpha}^\dagger
    =\tau_A^{\otimes k},                                                \\
    \Tr_{(BE)^k}|y_{\boldsymbol\alpha}\rangle\!\langle y_{\boldsymbol\alpha}|
    &=
    Y_{\boldsymbol\alpha}^\dagger\tau_A^{\otimes k}Y_{\boldsymbol\alpha}
    =\tau_A^{\otimes k},                                                \\
    \Tr_{(BE)^k}|x_{\boldsymbol\alpha}\rangle\!\langle y_{\boldsymbol\alpha}|
    &=
    X_{\boldsymbol\alpha}\tau_A^{\otimes k}Y_{\boldsymbol\alpha}
    =d_A^{-k}X_{\boldsymbol\alpha} Y_{\boldsymbol\alpha}
    =d_A^{-k}F_{\boldsymbol\alpha}, \\
    \Tr_{(BE)^k}|y_{\boldsymbol\alpha}\rangle\!\langle x_{\boldsymbol\alpha}|
    &=Y_{\boldsymbol\alpha}^\dagger\tau_A^{\otimes k}X_{\boldsymbol\alpha}^\dagger=d_A^{-k}Y_{\boldsymbol\alpha}^\dagger X_{\boldsymbol\alpha}^\dagger=d_A^{-k}F_{\boldsymbol\alpha}^\dagger.
\end{aligned}
\label{eq:xy-partial-traces}
\end{equation}
Since only \(\boldsymbol\alpha\) with
\(h_{\boldsymbol\alpha}\neq0\) contribute to the sum in
\zcref{eq:choi-marginal-deviation-expansion}, define
\begin{equation}
    |z_{\boldsymbol\alpha}\rangle
    \coloneq
    \sqrt{\frac{|h_{\boldsymbol\alpha}|}{2}}\,
    |x_{\boldsymbol\alpha}\rangle
    -
    \frac{\overline{h_{\boldsymbol\alpha}}}
         {\sqrt{2|h_{\boldsymbol\alpha}|}}\,
    |y_{\boldsymbol\alpha}\rangle \quad \in\operatorname{ran}P_{\hat\phi}^{(k,r)}.
\end{equation}
Using \zcref{eq:xy-partial-traces}, we obtain
\begin{equation}
\begin{aligned}
\Tr_{(BE)^k}
    |z_{\boldsymbol\alpha}\rangle
    \!\langle z_{\boldsymbol\alpha}|
&=
    \frac{|h_{\boldsymbol\alpha}|}{2}
    \left(
        \Tr_{(BE)^k}
        |x_{\boldsymbol\alpha}\rangle
        \!\langle x_{\boldsymbol\alpha}|
        +
        \Tr_{(BE)^k}
        |y_{\boldsymbol\alpha}\rangle
        \!\langle y_{\boldsymbol\alpha}|
    \right)                                                            \\
&\qquad
    -\frac{h_{\boldsymbol\alpha}}{2}
        \Tr_{(BE)^k}
        |x_{\boldsymbol\alpha}\rangle
        \!\langle y_{\boldsymbol\alpha}|
    -\frac{\overline{h_{\boldsymbol\alpha}}}{2}
        \Tr_{(BE)^k}
        |y_{\boldsymbol\alpha}\rangle
        \!\langle x_{\boldsymbol\alpha}|                               \\
& =
    |h_{\boldsymbol\alpha}|\tau_A^{\otimes k}
    -
    \frac{1}{2d_A^k}
    \left(
        h_{\boldsymbol\alpha}F_{\boldsymbol\alpha}
        +
        \overline{h_{\boldsymbol\alpha}}
        F_{\boldsymbol\alpha}^{\dagger}
    \right).
\end{aligned}
\end{equation}
Rearranging and using
\zcref{eq:choi-marginal-deviation-expansion}, we obtain
\begin{equation}
\label{eq:correction-non-sym}
\begin{aligned}
\Tr_{(BE)^k}\!\left[
    P_{\hat\phi}^{(k,r)}
    \omega^{(k)}
    P_{\hat\phi}^{(k,r)}
    +
    d_A^k
    \sum_{\substack{\boldsymbol\alpha:
        \operatorname{wt}(\boldsymbol\alpha)\leq\min\{2r,k\}}}
    |z_{\boldsymbol\alpha}\rangle
    \!\langle z_{\boldsymbol\alpha}|
\right]                                                               
=
\left(
    1+
    d_A^k
    \sum_{\substack{\boldsymbol\alpha:
        \operatorname{wt}(\boldsymbol\alpha)\leq\min\{2r,k\}}}
    |h_{\boldsymbol\alpha}|
\right)
\tau_A^{\otimes k}.
\end{aligned}
\end{equation}
We enforce permutation invariance by symmetrizing the correction term and define
\begin{equation}
\label{eq:choi-correction-sym}
\begin{aligned}
    \Delta^{(k)}
    \coloneq
    \frac{d_A^k}{k!}
    \sum_{\pi\in S_k}
    U_{(ABE)^k}^{\pi}
    \left(
        \sum_{\substack{\boldsymbol\alpha:
            \operatorname{wt}(\boldsymbol\alpha)
            \leq\min\{2r,k\}}}
        |z_{\boldsymbol\alpha}\rangle
        \!\langle z_{\boldsymbol\alpha}|
    \right)
    \bigl(U_{(ABE)^k}^{\pi}\bigr)^\dagger .
\end{aligned}
\end{equation}
Since every \(|z_{\boldsymbol\alpha}\rangle\in\operatorname{ran}P_{\hat\phi}^{(k,r)}\), and
\(P_{\hat\phi}^{(k,r)}\) commutes with every permutation, we have
\begin{equation}
    \Delta^{(k)}\geq0,
    \qquad
    P_{\hat\phi}^{(k,r)}
    \Delta^{(k)}
    P_{\hat\phi}^{(k,r)}
    =
    \Delta^{(k)},
    \qquad
    U_{(ABE)^k}^{\pi}
    \Delta^{(k)}
    \bigl(U_{(ABE)^k}^{\pi}\bigr)^\dagger
    =
    \Delta^{(k)}, \quad\forall \pi\in S_k.
\end{equation}
Both \(P_{\hat\phi}^{(k,r)}\omega^{(k)}P_{\hat\phi}^{(k,r)}\) and \(\tau_A^{\otimes k}\) are permutation invariant, symmetrizing \zcref{eq:correction-non-sym} gives
\begin{equation}
    \Tr_{(BE)^k}\!\left[
        P_{\hat\phi}^{(k,r)}
        \omega^{(k)}
        P_{\hat\phi}^{(k,r)}
        +
        \Delta^{(k)}
    \right]
    =
    (1+\lambda)\tau_A^{\otimes k}.
\end{equation}
where we have defined the completion cost
\begin{equation}
    \lambda
    \coloneq
    d_A^k
    \sum_{\substack{\boldsymbol\alpha:
        \operatorname{wt}(\boldsymbol\alpha)
        \leq\min\{2r,k\}}}
    |h_{\boldsymbol\alpha}|.
\end{equation}
We may thus define the completed state
\begin{equation}
    \widehat\omega^{(k)}
    \coloneq
    \frac{1}{1+\lambda}\left(P_{\hat\phi}^{(k,r)}
        \omega^{(k)}
        P_{\hat\phi}^{(k,r)}
        +
        \Delta^{(k)}\right).
\end{equation}
The preceding argument shows that \(\widehat\omega^{(k)}\) is normalized, permutation invariant, satisfies the Choi constraint \(\widehat\omega_{A^k}^{(k)}=\tau_A^{\otimes k}\), and \(\operatorname{supp}\widehat\omega^{(k)}\subseteq\operatorname{ran}P_{\hat\phi}^{(k,r)}\). Thus, by definition, \(\widehat\omega^{(k)}\) is a
\(\binom{k}{r}\)-almost-iid state along \(\hat\phi\). Finally, since
\(\Delta^{(k)}\geq0\),
\begin{equation}
    P_{\hat\phi}^{(k,r)}
    \omega^{(k)}
    P_{\hat\phi}^{(k,r)}
    \leq
    (1+\lambda)\widehat\omega^{(k)}.
\label{eq:choi-completion-before-factor-estimate}
\end{equation}

\textbf{Step 3: Estimate the completion cost.}
It remains to bound
\begin{equation}
\label{eq:def-halpha-lambda}
    \lambda=d_A^k
    \sum_{\substack{\boldsymbol\alpha:
        \operatorname{wt}(\boldsymbol\alpha)
        \leq\min\{2r,k\}}}
    |h_{\boldsymbol\alpha}|,\qquad
    h_{\boldsymbol\alpha}
    =
    d_A^{-k}
    \Tr\!\left[
        F_{\boldsymbol\alpha}^{\dagger}
        \left(
            \Tr_{(BE)^k}\!\left[
                P_{\hat\phi}^{(k,r)}
                \omega^{(k)}
                P_{\hat\phi}^{(k,r)}
            \right]
            -\tau_A^{\otimes k}
        \right)
    \right].
\end{equation}
in terms of the discarded mass \(\delta_{kr}\). This estimate, to be found in \zcref{App:Proof:lambda-bound}, gives,
\begin{equation}
\label{eq:lambda-bound}
\lambda
\leq
2\sqrt{
    \delta_{kr}
    N_{kr}
}
+
\delta_{kr}
N_{kr}.
\end{equation}
The idea behind it is simple: since \(\omega_{A^k}^{(k)}=\tau_A^{\otimes k}\), the deviation term in \(h_{\boldsymbol\alpha}\) is directly related to the disturbance \(P_{\hat\phi}^{(k,r)}\omega^{(k)}P_{\hat\phi}^{(k,r)}-\omega^{(k)}\) caused by the projection, and therefore naturally associated with \(\delta_{kr}\), while the combinatorial factor \(N_{kr}\) comes from Cauchy--Schwarz when summing over \(\boldsymbol\alpha:\operatorname{wt}(\boldsymbol\alpha)\leq\min\{2r,k\}\). 

Combining this with
\zcref{eq:choi-completion-before-factor-estimate} completes the proof.
\end{proof}

It is now straightforward to combine \zcref{prop:projected-channel-df} and \zcref{prop:choi-completion} to prove the main theorem.

\begin{proof}[Proof of \zcref{thm:exp-channel-df}]
Let \(E\simeq AB\), \(D=\dim(ABE)=d_A^2d_B^2\), and let
\(|\Psi^{(n)}\rangle\in\operatorname{Sym}^n(ABE)\) be a symmetric purification of \(J^{(n)}=J(\mathcal{N}^{(n)})\). Its satisfies the Choi constraint \(\Psi_{A^n}^{(n)}=\tau_A^{\otimes n}\). Let
\(\Psi^{(k)}\coloneq\Tr_{(ABE)^{n-k}}\Psi^{(n)}\), then \(\Tr_{E^k}\Psi^{(k)}=\Tr_{(AB)^{n-k}}J^{(n)}=J^{(k)}=J(\mathcal{N}^{(k)})\), see \zcref{eq:reduced-choi-state}.

Let \(\eta_{nkr}\), \(\mu_{nkr}\), and
\(\{\omega_{\hat\phi}^{(k)}\}_{\hat\phi}\) be supplied by
\zcref{prop:projected-channel-df}, and define
\begin{equation}
    \delta_{kr}^{\hat\phi}
    \coloneq
    \Tr\!\left[
        \bigl(
            \id_{(ABE)^k}-P_{\hat\phi}^{(k,r)}
        \bigr)
        \omega_{\hat\phi}^{(k)}
    \right].
\end{equation}
Then
\begin{equation}
    \eta_{nkr}J^{(k)}
    \leq
    \int_{\mathcal P_{ABE}(\tau_A)}
        \Tr_{E^k}\!\left[
            P_{\hat\phi}^{(k,r)}
            \omega_{\hat\phi}^{(k)}
            P_{\hat\phi}^{(k,r)}
        \right]
        \mathrm d\mu_{nkr}(\hat\phi),
    \label{eq:exp-channel-df-proof-projected}
\end{equation}
and
\begin{equation}
    \int_{\mathcal P_{ABE}(\tau_A)}
        \delta_{kr}^{\hat\phi}
        \,\mathrm d\mu_{nkr}(\hat\phi)
    \leq
    \eta_{nkr}d_{n-k}
    \left(\frac{k}{n}\right)^{r+1}.
    \label{eq:exp-channel-df-proof-tail}
\end{equation}
By \zcref{prop:choi-completion}, for every \(\hat\phi\) there exists
a \(\binom{k}{r}\)-almost-iid Choi state \(\widehat\omega_{\hat\phi}^{(k)}\) along \(\hat\phi\), so that
\begin{equation}
    P_{\hat\phi}^{(k,r)}
    \omega_{\hat\phi}^{(k)}
    P_{\hat\phi}^{(k,r)}
    \leq
    \left(
        1+\sqrt{N_{kr}\delta_{kr}^{\hat\phi}}
    \right)^2
    \widehat\omega_{\hat\phi}^{(k)}, \qquad N_{kr}=
    \sum_{j=0}^{\min\{2r,k\}}
        \binom{k}{j}(d_A^2-1)^j.
    \label{eq:exp-channel-df-proof-completion}
\end{equation}

Define the re-weighted probability measure
\begin{equation}
\mathrm d\nu_{nkr}(\hat\phi)
    \coloneq
    \frac{1}{Z_{nkr}}\left(
            1+\sqrt{N_{kr}\delta_{kr}^{\hat\phi}}
        \right)^2\,
    \mathrm d\mu_{nkr}(\hat\phi),
\label{eq:exp-channel-df-proof-measure}
\end{equation}
where
\begin{equation}
    Z_{nkr}
    \coloneq
    \int_{\mathcal P_{ABE}(\tau_A)}
        \left(
            1+\sqrt{N_{kr}\delta_{kr}^{\hat\phi}}
        \right)^2
        \mathrm d\mu_{nkr}(\hat\phi).
\end{equation}
By Jensen's inequality and \zcref{eq:exp-channel-df-proof-tail},
\begin{equation}
\begin{aligned}
    1\leq Z_{nkr}
    &=
    1
    +2\sqrt{N_{kr}}
        \int
            (\delta_{kr}^{\hat\phi})^{1/2}
            \,\mathrm d\mu_{nkr}(\hat\phi)
    +N_{kr}
        \int
            \delta_{kr}^{\hat\phi}
            \,\mathrm d\mu_{nkr}(\hat\phi)                              \\
    &\leq
    \left(
        1+
        \sqrt{
            \eta_{nkr}d_{n-k}N_{kr}
            \left(\frac{k}{n}\right)^{r+1}
        }
    \right)^2.
\end{aligned}
\label{eq:exp-channel-df-proof-normalization}
\end{equation}
Combining
\zcref{eq:exp-channel-df-proof-projected} and \zcref{eq:exp-channel-df-proof-completion}
gives
\begin{equation}
    \eta_{nkr}J^{(k)}
    \leq
    Z_{nkr}
    \int_{\mathcal P_{ABE}(\tau_A)}
        \Tr_{E^k}\widehat\omega_{\hat\phi}^{(k)}
        \,\mathrm d\nu_{nkr}(\hat\phi).
    \label{eq:exp-channel-df-proof-renormalized}
\end{equation}

Let \(1-\epsilon_{nkr}\coloneq \eta_{nkr}/Z_{nkr}\). It remains to estimate \(\epsilon_{nkr}\). Using \((1+x)^{-2}\geq1-2x\) for \(x\geq0\) and \zcref{eq:exp-channel-df-proof-normalization}, we have
\begin{equation}
    \epsilon_{nkr}
    \leq
    1-\eta_{nkr}
    +
    2\sqrt{d_{n-k}N_{kr}}
    \left(\frac{k}{n}\right)^{(r+1)/2}.
    \label{eq:channel-df-error-preliminary}
\end{equation}
The factor \(1-\eta_{nkr}\) can be estimated using \zcref{lem:ank-Lj-bounds}, giving (see \zcref{App:Proof:eta-exponential-bound}),
\begin{equation}
    1-\eta_{nkr}
    \leq
    2d_k\sqrt{d_{n-k}(k-r)}
    \exp\!\left[
        -\frac{r+1}{2}
        \ln\!\left(
            \frac{n(r+1)}{e d_Ak^2}
        \right)
    \right].
    \label{eq:eta-exponential-bound}
\end{equation}
The second factor can be estimated to be (see \zcref{App:Proof:Nkr-exponential-bound}),
\begin{equation}
    2\sqrt{d_{n-k}N_{kr}}
    \left(\frac{k}{n}\right)^{(r+1)/2}
    \leq
    2\sqrt{d_{n-k}}
    \exp\!\left[
        -\frac{r+1}{2}
        \ln\!\left(
            \frac{nr^2}{2d_A^4k^3}
        \right)
    \right].
\label{eq:Nkr-exponential-bound}
\end{equation}

Finally, combining \zcref{eq:eta-exponential-bound} and \zcref{eq:Nkr-exponential-bound}, using \(d_k\leq d_n\leq(n+1)^{D-1}\), and noticing that
\begin{equation}
    \frac{n(r+1)}{e d_Ak^2}
    \geq
    \frac{nr^2}{2d_A^4k^3},
\end{equation}
we get
\begin{equation}
\begin{aligned}
    \epsilon_{nkr}
    &\leq
    2\sqrt{d_{n-k}}
    \left(
        1+d_k\sqrt{k-r}
    \right)                                                           
    \exp\!\left[
        -\frac{r+1}{2}
        \ln\!\left(
            \frac{nr^2}{2d_A^4k^3}
        \right)
    \right]\\
    &\leq 4(n+1)^{\frac{3D-2}{2}}\exp\!\left[
        -\frac{r+1}{2}
        \ln\!\left(
            \frac{nr^2}{2d_A^4k^3}
        \right)
    \right].
\end{aligned}
\label{eq:channel-df-combined-error}
\end{equation}

From \zcref{eq:exp-channel-df-proof-renormalized}, we have
\begin{equation}
    (1-\epsilon_{nkr})J^{(k)}
    \leq
    \int_{\mathcal P_{ABE}(\tau_A)}
        \Tr_{E^k}\widehat\omega_{\hat\phi}^{(k)}
        \,\mathrm d\nu_{nkr}(\hat\phi).
\end{equation}

For each \(\hat\phi\in\mathcal P_{ABE}(\tau_A)\), let \(\Phi\in\mathrm{CPTP}(A,B)\) be the channel with Choi state \(\hat\phi_{AB}\), and let \(\mathcal E_{\Phi}^{(k,r)}\in\mathrm{CPTP}(A^k,B^k)\) be the channel with Choi state \(\Tr_{E^k}\widehat\omega_{\hat\phi}^{(k)}\). By \zcref{def:almost-iid-channel}, \(\mathcal E_{\Phi}^{(k,r)}\)is \(\binom{k}{r}\)-almost-iid along \(\Phi\). Push \(\nu_{nkr}\) forward under \(\hat\phi\mapsto\Phi\) and, when several purifications correspond to the same \(\Phi\), conditionally average the associated channels; this preserves the almost-iid property because the set of almost-iid channels along a fixed reference channel is convex. Reusing \(\nu_{nkr}\) to denote the resulting measure on \(\mathrm{CPTP}(A,B)\), then
\begin{equation}
 (1-\epsilon_{nkr})J(\mathcal{N}^{(k)})
    \leq
    \int
        J\left(\mathcal E_{\Phi}^{(k,r)}\right)
        \,\mathrm d\nu_{nkr}(\Phi),
\end{equation}
which implies the CP-order \zcref{eq:exp-channel-df}. Substituting \(D=d_A^2d_B^2\) completes the proof.
\end{proof}

\begin{corollary}[(Diamond-norm exponential de Finetti)]
\label{cor:exponential-channel-definetti-diamond}
Under the assumptions of \zcref{thm:exp-channel-df},
\begin{equation}
\left\|
    \Tr_{n-k}\mathcal N^{(n)}
    -
    \int
        \mathcal E_{\Phi}^{(k,r)}
        \,\mathrm d\nu_{nkr}(\Phi)
\right\|_{\diamond}
\leq
8(n+1)^{\frac{3d_A^2d_B^2-2}{2}}
\exp\!\left[
    -\frac{r+1}{2}
    \ln\!\left(
        \frac{nr^2}{2d_A^4k^3}
    \right)
\right].
\end{equation}
\end{corollary}

\begin{proof}
Same as \zcref{cor:standard-channel-definetti-diamond}.
\end{proof}

\begin{remark}
\label{rmk:decay-regime}
Let \(k\sim n^\alpha\) and \(r\sim n^\beta\), with \(0<\beta<\alpha<1\). Then \(nr^2/k^3\sim n^{1+2\beta-3\alpha}\). Hence, whenever \(1+2\beta-3\alpha>0\), for large enough \(n\), the error decays exponentially. For example, taking \(k\sim n^{3/4}\) and \(r\sim n^{2/3}\) gives \(\epsilon_{nkr}\sim\exp\!\left(-n^{2/3}\ln n\right)\). However, the drawback of \zcref{thm:exp-channel-df} is that it does not allow us to keep most of the subsystems. Indeed, if \(k\sim n-o(n)\), then \(nr^2/k^3\leq n/k\to1\) as \(n\to\infty\), so that the exponent will eventually become positive, and the error bound will explode instead of decaying. We leave the question of finding an exponential de Finetti theorem for channels with exponential decay while discarding only a negligible number of systems open.
\end{remark}

\section*{Acknowledgements}
BB and ND are supported by the Engineering and Physical Sciences Research Council [Grant Ref: EP/Y028732/1]. LY is also supported by the Engineering and Physical Sciences Research Council. The authors thank Mario Berta, Lampros Gavalakis and Ioannis Kontoyiannis for helpful discussions.

\textbf{Statement on AI usage.} OpenAI’s GPT-5.6 Sol was used to assist with brainstorming, literature search, and exploring proof strategies during the development of this work, and was used to transcribe handwritten proofs into \LaTeX. GPT-6 Astra later helped improve the estimation of several error bounds in the appendices, and Claude Fable 5.1 was used to improve parts of the introduction. The organization and exposition of the results are due to the authors, who take full responsibility for the contents in the manuscript.
\bibliography{main}
\appendix
\addtocontents{toc}{\protect\hideappendixsubsections}
\renewcommand{\thesubsection}{\thesection.\arabic{subsection}}
\section{Proof of technical statements in~\zcref{sec:std-channel-df}}

\subsection{Proof of~\zcref{lem:fixed-marginal-haar-moment}}
\label{App:Proof:fixed-marginal-haar-moment}
\begin{proof}
First we argue that the inverse in \(K_{A^m}\) is well defined. Since \(d_A^{-m/2}|\theta\rangle_{AR}^{\otimes m}\) is a unit vector in
\(\operatorname{Sym}^{m}(AR)\), we have
\begin{equation}
P_{\mathrm{sym}}^{(m)}=d_A^{-m}\bigl(|\theta\rangle\!\langle\theta|_{AR}\bigr)^{\otimes m}+P_{\perp}
    \geq d_A^{-m}
    \bigl(|\theta\rangle\!\langle\theta|_{AR}\bigr)^{\otimes m},
\end{equation}
where \(P_{\perp}\geq0\) is the projection onto the orthogonal complement
of \(\operatorname{span}\{|\theta\rangle^{\otimes m}\}\) within the symmetric subspace. Taking the partial trace over \(R^m\) and using
\(\Tr_R|\theta\rangle\!\langle\theta|_{AR}=\id_A\) gives a strictly positive marginal \(\Tr_{R^m}P_{\mathrm{sym}}^{(m)}\geq \tau_{A^m}>0\), so that \(K_{A^m}\) is well defined.

Next we show that \(K_{A^m}\) commutes with permutations on \(A^{\otimes m}\). For \(\pi\in S_m\), let \(W_\pi^A\) and \(W_\pi^R\) denote the corresponding
unitary permutation operators on \(A^m\) and \(R^m\), respectively. Since
\(P_{\mathrm{sym}}^{(m)}\) is invariant under
\(W_\pi^A\otimes W_\pi^R\),
\begin{equation}
    W_\pi^A\bigl(\Tr_{R^m}P_{\mathrm{sym}}^{(m)}\bigr)(W_\pi^A)^\dagger
    =\Tr_{R^m}\!\left[
      (W_\pi^A\otimes W_\pi^R)P_{\mathrm{sym}}^{(m)}
      (W_\pi^A\otimes W_\pi^R)^\dagger
      \right]
    =\Tr_{R^m}P_{\mathrm{sym}}^{(m)}.
\end{equation}
Hence \(K_{A^m}\) commutes with every \(W_\pi^A\). Using
\(P_{\mathrm{sym}}^{(m)}
 =\frac{1}{m!}\sum_{\pi\in S_m}W_\pi^A\otimes W_\pi^R\), it follows that
\begin{equation}
\label{eq:comm-K}
    (K_{A^m}\otimes\id_{R^m})P_{\mathrm{sym}}^{(m)}
    =P_{\mathrm{sym}}^{(m)}(K_{A^m}\otimes\id_{R^m}).
\end{equation}

For brevity, set \(T_m\coloneq\int_{\mathsf U(R)}
    \bigl[(\id_A\otimes U_R)|\theta\rangle\!\langle\theta|_{AR}
    (\id_A\otimes U_R^\dagger)\bigr]^{\otimes m}\,\mathrm dU\). We next show that
\begin{equation}
\label{eq:comm-unitary}
    [T_m,V_A^{\otimes m}\otimes\id_{R^m}]=0
    \qquad\text{for every }V_A\in\mathsf U(A).
\end{equation}
Fix \(V_A\in\mathsf U(A)\), and extend \(V_A^{\mathsf T}\) acting on \(\operatorname{span}\{|i\rangle_R:1\leq i\leq d_A\}\), to a unitary
\(\widetilde V_R\in\mathsf U(R)\). Recall the transpose trick \((V_A\otimes\id_R)|\theta\rangle_{AR}=(\id_A\otimes\widetilde V_R)|\theta\rangle_{AR}\). Consequently,
\begin{equation}
\begin{aligned}
 &(V_A^{\otimes m}\otimes\id_{R^m})T_m
   ((V_A^\dagger)^{\otimes m}\otimes\id_{R^m}) \\
 &\quad=\int_{\mathsf U(R)}
   \bigl[(V_A\otimes U_R)|\theta\rangle\!\langle\theta|_{AR}
   (V_A^\dagger\otimes U_R^\dagger)\bigr]^{\otimes m}\,\mathrm dU \\
 &\quad=\int_{\mathsf U(R)}
   \bigl[(\id_A\otimes U_R\widetilde V_R)
   |\theta\rangle\!\langle\theta|_{AR}
   (\id_A\otimes\widetilde V_R^\dagger U_R^\dagger)\bigr]^{\otimes m}
   \,\mathrm dU
   =T_m,
\end{aligned}
\end{equation}
where the last equality follows from the right invariance of Haar measure. This proves \zcref{eq:comm-unitary}. By~\cite[Theorem~7.15]{watrous2018theory}, commutation with all tensor-power unitaries is equivalent to commutation with all tensor-power linear operators, we thus conclude
\begin{equation}
\label{eq:comm-all}
    [T_m,X_A^{\otimes m}\otimes\id_{R^m}]=0
    \qquad\text{for every }X_A\in\mathcal L(A).
\end{equation}

Let \(\phi_{AR}=|\phi\rangle\!\langle\phi|_{AR}\) be any pure state. Since
\(|\phi\rangle_{AR}\) and \((\sqrt{\phi_A}\otimes\id_R)|\theta\rangle_{AR}\) are both purifications of \(\phi_A\), there exists \(V_R\in\mathsf U(R)\) such that \(|\phi\rangle_{AR}=(\sqrt{\phi_A}\otimes V_R)|\theta\rangle_{AR}\). Thus
\begin{equation}
\label{eq:rotated-purification}
\begin{aligned}
 &\quad\int_{\mathsf U(R)}
 \bigl[(\id_A\otimes U_R)|\phi\rangle\!\langle\phi|_{AR}
 (\id_A\otimes U_R^\dagger)\bigr]^{\otimes m}\,\mathrm dU \\
 &=\left(\sqrt{\phi_A}\otimes\id_{R}\right)^{\otimes m}
 \int_{\mathsf U(R)}
 \bigl[(\id_A\otimes U_RV_R)|\theta\rangle\!\langle\theta|_{AR}
 (\id_A\otimes V_R^\dagger U_R^\dagger)\bigr]^{\otimes m}\,\mathrm dU
 \left(\sqrt{\phi_A}\otimes\id_{R}\right)^{\otimes m}\\
&=\left(\phi_A^{\otimes m}\otimes\id_{R^m}\right)T_m,
\end{aligned}
\end{equation}
where in the last step we have used the right invariance of Haar measure and the commutation relation \zcref{eq:comm-all}. Therefore, \zcref{eq:fixed-marginal-haar-moment} in the lemma is a consequence of \zcref{eq:constrained-Haar}. To prove the latter, take the partial trace of \(\int|\phi\rangle\!\langle\phi|_{AR}^{\otimes m}\,\mathrm d\phi=\frac{1}{d_m}P_{\mathrm{sym}}^{(m)}\), giving
\begin{equation}
    \int\phi_A^{\otimes m}\,\mathrm d\phi
    =\frac{1}{d_m}\Tr_{R^m}P_{\mathrm{sym}}^{(m)}
    =\frac{1}{d_m}K_{A^m}^{-1}.
\end{equation}
Integrating \zcref{eq:rotated-purification} on both sides over \(\mathrm d\phi\), and using \(U_R^{\otimes m}P_{\mathrm{sym}}^{(m)}(U^{\dagger}_R)^{\otimes m}=P_{\mathrm{sym}}^{(m)}\), yields
\begin{equation}
    \frac{1}{d_m}P_{\mathrm{sym}}^{(m)}
    =\frac{1}{d_m}(K_{A^m}^{-1}\otimes\id_{R^m})T_m.
\end{equation}
It follows that
\begin{equation}
    T_m=(K_{A^m}\otimes\id_{R^m})P_{\mathrm{sym}}^{(m)}
       =P_{\mathrm{sym}}^{(m)}(K_{A^m}\otimes\id_{R^m}),
\end{equation}
where the second equality follows from \zcref{eq:comm-K}. This proves \zcref{eq:constrained-Haar}, and hence \zcref{eq:fixed-marginal-haar-moment} in the lemma.
\end{proof}

\subsection{Proof of~\zcref{lem:fixed-marginal-measure-independence}}
\label{App:Proof:lem:fixed-marginal-measure-independence}
\begin{proof}
For any \(U\in\mathsf U(R)\), we have \(\Tr_R\!\left[(\id_A\otimes U_R)\phi_{AR}(\id_A\otimes U_R^\dagger)\right]=\phi_A\), so denote \(|\phi'\rangle=(\id_A\otimes U_R)|\phi\rangle\), we have \(\phi'_A=\phi_A\), then
\begin{equation}
\begin{aligned}
    \int f(\phi_A)\,\mathrm d\mu_n(\phi)
    &=
    d_n\int f(\phi_A)
    \Tr\!\left[\Psi^{(n)}|\phi\rangle\!\langle\phi|^{\otimes n}\right]\mathrm d\phi \\
    &=
    d_n\int f(\phi'_A)
    \Tr\!\left[\Psi^{(n)}|\phi'\rangle\!\langle\phi'|^{\otimes n}\right]\mathrm d\phi \\
    &=
    d_n\int f(\phi_A)
    \Tr\!\left[
        \Psi^{(n)}
        \bigl((\id_A\otimes U_R)\phi_{AR}
        (\id_A\otimes U_R^\dagger)\bigr)^{\otimes n}
    \right]\mathrm d\phi,
\end{aligned}
\end{equation}
where the second line follows from left invariance of Haar measure.
Integrating both sides over \(U\) with respect to the Haar measure on \(\mathsf U(R)\) gives
\begin{equation}
\begin{aligned}
    \int f(\phi_A)\,\mathrm d\mu_n(\phi)
    &=
    d_n\int_{\mathsf U(R)}\int f(\phi_A)
    \Tr\!\left[
        \Psi^{(n)}
        \bigl((\id_A\otimes U_R)\phi_{AR}
        (\id_A\otimes U_R^\dagger)\bigr)^{\otimes n}
    \right]\mathrm d\phi\,\mathrm dU \\
    &=
    d_n\int f(\phi_A)
    \Tr\!\left[
        \Psi^{(n)}
        \int_{\mathsf U(R)}
        \bigl((\id_A\otimes U_R)\phi_{AR}
        (\id_A\otimes U_R^\dagger)\bigr)^{\otimes n}
        \mathrm dU
    \right]\mathrm d\phi .
\end{aligned}
\end{equation}
By \zcref{lem:fixed-marginal-haar-moment}, this becomes
\begin{equation}
\begin{aligned}
    \int f(\phi_A)\,\mathrm d\mu_n(\phi)
    &=
    d_n\int f(\phi_A)
    \Tr\!\left[
        \Psi^{(n)}
        \bigl(\phi_A^{\otimes n}K_{A^n}\otimes\id_{R^n}\bigr)
        P_{\mathrm{sym}}^{(n)}
    \right]\mathrm d\phi \\
    &=
    d_n\int f(\phi_A)
    \Tr\!\left[
        \Psi^{(n)}
        \bigl(\phi_A^{\otimes n}K_{A^n}\otimes\id_{R^n}\bigr)
    \right]\mathrm d\phi \\
    &=
    \frac{d_n}{d_A^n}
    \int f(\phi_A)
    \Tr\!\left[\phi_A^{\otimes n}K_{A^n}\right]\mathrm d\phi .
\end{aligned}
\end{equation}
In the second line, we used
\(P_{\mathrm{sym}}^{(n)}\Psi^{(n)}=\Psi^{(n)}\), and in the third line
we used \(\Psi_{A^n}^{(n)}=\tau_A^{\otimes n}\).

The key observation is that the resulting expression is now independent of the particular \(|\Psi^{(n)}\rangle\). We may therefore replace \(\Psi^{(n)}\) by
\(\psi^{\otimes n}\), where \(|\psi\rangle_{AR}\) can be any pure state satisfying \(\psi_A=\tau_A\). It follows that
\begin{equation}
    \int f(\phi_A)\,\mathrm d\mu_n(\phi)
    =
    d_n\int f(\phi_A)
    \langle\phi|^{\otimes n}
    \psi^{\otimes n}
    |\phi\rangle^{\otimes n}\,\mathrm d\phi
    =
    d_n\int f(\phi_A)
    |\langle\psi|\phi\rangle|^{2n}\mathrm d\phi ,
\end{equation}
which proves the claim.
\end{proof}

\subsection{Proof of \zcref{eq:Z-estimate}}
\label{App:Proof:eq-Z-estimate}
We would like to bound \(\mathbb E\exp\left(2d_Ak\sqrt U\right)\), with \(U\sim\operatorname{Beta}(D-1,n+1)\). 

Expanding the exponential gives
\begin{equation}
    \mathbb E\exp\left(2d_Ak\sqrt U\right)=\sum_{j=0}^{\infty}\frac{(2d_Ak)^j}{j!}\mathbb E U^{j/2}=\sum_{j=0}^{\infty}\frac{(2d_Ak)^j}{j!}\frac{\mathrm B(D-1+j/2,n+1)}{\mathrm B(D-1,n+1)}.
\end{equation}
The ratio of two beta functions can be expressed in terms of ratio of gamma functions. Gautschi's inequality implies the bound \(\Gamma(n+D)/\Gamma(n+D+s)\leq(n+D-1)^{-s}\) for \(s\geq0\). Set \(t\coloneq 2d_Ak/\sqrt{n+D-1}\), we obtain
\begin{equation}
\begin{aligned}
\mathbb E\exp\left(2d_Ak\sqrt U\right)
&=\sum_{j=0}^{\infty}\frac{(2d_Ak)^j}{j!}
  \frac{\Gamma(D-1+j/2)\Gamma(n+D)}
       {\Gamma(D-1)\Gamma(n+D+j/2)} \\
&\leq \frac{1}{\Gamma(D-1)}
  \sum_{j=0}^{\infty}\frac{t^j}{j!}\Gamma(D-1+j/2)
=\frac{1}{\Gamma(D-1)}\int_0^\infty x^{D-2}e^{-x+t\sqrt{x}}\,\mathrm dx,
\end{aligned}
\end{equation}
where the last step follows from the definition of gamma function. For any \(0<v<1\), the inequality \(t\sqrt{x}\leq vx+t^2/(4v)\) therefore gives
\begin{equation}
\mathbb E\exp\left(2d_Ak\sqrt U\right)
\leq \frac{e^{t^2/(4v)}}{\Gamma(D-1)}\int_0^\infty x^{D-2}e^{-(1-v)x}\,\mathrm dx =e^{t^2/(4v)}(1-v)^{-(D-1)}.
\end{equation}
Choosing \(v\coloneq t/(2\sqrt{D-1}+t)\) and using
\(\log(1+x)\leq x\), we find
\begin{equation}
\begin{aligned}
\log\mathbb E\exp\left(2d_Ak\sqrt U\right)
&\leq \frac{t\sqrt{D-1}}{2}+\frac{t^2}{4}
 +(D-1)\log\left(1+\frac{t}{2\sqrt{D-1}}\right)\\
&\leq t\sqrt{D-1}+\frac{t^2}{4}=2d_Ak\sqrt{\frac{D-1}{n+D-1}}
+\frac{d_A^2k^2}{n+D-1}.
\end{aligned}
\end{equation}
\zcref{eq:Z-estimate} is proved by taking the exponential,
\begin{equation}
\mathbb E\exp\left(2d_Ak\sqrt U\right)\leq
\exp\left(
2d_Ak\sqrt{\frac{D-1}{n+D-1}}
+\frac{d_A^2k^2}{n+D-1}
\right).
\end{equation}

\section{Proof of technical statements in~\zcref{sec:exp-channel-df}}

\subsection{Proof of~\zcref{lem:ank-Lj-bounds}}
\label{App:Proof:lem:ank-Lj-bounds}
\begin{proof}
We first prove the bound on \(a_{nk}\). Recall that
\begin{equation}
    a_{nk}\coloneq
    \int_{\mathcal P_{ABE}(\tau_A)}
        F_{\hat\phi}^{\,n-k}\,\mathrm d\hat\phi                                 
    =
    d_A^{-2(n-k)}\Tr K_{A^{n-k}}.
\end{equation}
The upper bound follows from \(0\leq F_{\hat\phi}\leq 1\). For the lower bound, let \(\{\lambda_1,\ldots,\lambda_{d_A^{n-k}}\}\) be the strictly positive eigenvalues of \(K_{A^{n-k}}^{-1}\). By its definition, \(K_{A^{n-k}}^{-1}=\Tr_{(BE)^{n-k}}P_{\mathrm{sym}}^{(n-k)}\), so
\begin{equation}
    \sum_{\alpha=1}^{d_A^{n-k}}\lambda_\alpha
    =
    \Tr K_{A^{n-k}}^{-1}
    =
    \Tr P_{\mathrm{sym}}^{(n-k)}
    =
    d_{n-k}.
\end{equation}
The Cauchy--Schwarz inequality therefore gives
\begin{equation}
    d_A^{2(n-k)}
    =
    \left(
        \sum_{\alpha=1}^{d_A^{n-k}}1
    \right)^2
    \leq
    \left(
        \sum_{\alpha=1}^{d_A^{n-k}}\lambda_\alpha
    \right)
    \left(
        \sum_{\alpha=1}^{d_A^{n-k}}\lambda_\alpha^{-1}
    \right)
    =
    d_{n-k}\Tr K_{A^{n-k}},
\end{equation}
proving \zcref{eq:ank-bounds}.

We next prove the operator-norm bound. Recall that
\begin{equation}
    L_j
    \coloneq
    \frac{d_A^k}{a_{nk}}
    \int_{\mathcal P_{ABE}(\tau_A)}
        F_{\hat\phi}^{\,n-k}
        \Tr_{(BE)^k}
        \left[
            \Pi_{\hat\phi,j}^{(k)}
            \hat\psi^{\otimes k}
        \right]
        \mathrm d\hat\phi,
\end{equation}
and the definition does not depend on the choice of \(\hat\psi\in\mathcal P_{ABE}(\tau_A)\). For
\(\hat\phi\in\mathcal P_{ABE}(\tau_A)\), let \(z_{\hat\phi}\coloneq\langle\hat\phi|\hat\psi\rangle\) so that \(F_{\hat\phi}=|z_{\hat\phi}|^2\). By the definition of \(\Pi_{\hat\phi,j}^{(k)}\),
\begin{equation}
    \Tr_{(BE)^k}
    \left[
        \Pi_{\hat\phi,j}^{(k)}
        \hat\psi^{\otimes k}
    \right]
    =
    \sum_{\substack{S\subseteq[k]\\ |S|=j}}
        \left(z_{\hat\phi}
    \Tr_{BE}|\hat\phi\rangle\!\langle\hat\psi|\right)^{\otimes S^c}
        \otimes
        \left(\Tr_{BE}
    \left[
        \bigl(
            |\hat\psi\rangle
            -z_{\hat\phi}|\hat\phi\rangle
        \bigr)
        \langle\hat\psi|
    \right]\right)^{\otimes S}.
\end{equation}
We estimate the operator norm of each term. For arbitrary vectors \(|u\rangle,|v\rangle\in ABE\) and unit vectors \(|a\rangle,|b\rangle\in A\), Cauchy--Schwarz gives
\begin{equation}
\label{eq:partial-trace-cs}
    \left|
        \left\langle a\middle|
        \Tr_{BE}|u\rangle\!\langle v|
        \middle|b\right\rangle
    \right|^2
    =
    \left|
    \langle v|
        (|b\rangle\otimes\id)(\langle a|\otimes\id)
    |u\rangle
    \right|^2
    \leq
    \langle u|
        (|a\rangle\!\langle a|\otimes\id)
    |u\rangle\cdot
    \langle v|
        (|b\rangle\!\langle b|\otimes\id)
    |v\rangle.
\end{equation}
Applying \zcref{eq:partial-trace-cs} with
\(|u\rangle=|\hat\phi\rangle\) and
\(|v\rangle=|\hat\psi\rangle\), we obtain
\begin{align}
    \left|
        \left\langle a\middle|
        \Tr_{BE}|\hat\phi\rangle\!\langle\hat\psi|
        \middle|b
        \right\rangle
    \right|^2
    &\leq
    \langle a|\hat\phi_A|a\rangle
    \langle b|\hat\psi_A|b\rangle
    =
    \frac{1}{d_A^2}.
\end{align}
Consequently, by recalling that \(\lVert M\rVert_\infty
=\sup_{\lVert a\rVert=\lVert b\rVert=1}
|\langle a|M|b\rangle|\), we obtain
\begin{equation}
    \lVert z_{\hat\phi}
    \Tr_{BE}|\hat\phi\rangle\!\langle\hat\psi|\rVert_\infty
    \leq
    \frac{|z_{\hat\phi}|}{d_A}
    =
    \frac{\sqrt{F_{\hat\phi}}}{d_A}.
\end{equation}
Applying \zcref{eq:partial-trace-cs} with
\(|u\rangle=|\hat\psi\rangle-z_{\hat\phi}|\hat\phi\rangle\) and
\(|v\rangle=|\hat\psi\rangle\) gives
\begin{equation}
    \left|
        \left\langle a\middle|
        \Tr_{BE}
    \left[
        \bigl(
            |\hat\psi\rangle
            -z_{\hat\phi}|\hat\phi\rangle
        \bigr)
        \langle\hat\psi|
    \right]
        \middle|b
        \right\rangle
    \right|^2
    \leq
    \frac{\lVert|\hat\psi\rangle-z_{\hat\phi}|\hat\phi\rangle\rVert^2}{d_A}
    =
    \frac{1-F_{\hat\phi}}{d_A}.
\end{equation}
Consequently,
\begin{equation}
    \left\Vert \Tr_{BE}
    \left[
        \bigl(
            |\hat\psi\rangle
            -z_{\hat\phi}|\hat\phi\rangle
        \bigr)
        \langle\hat\psi|
    \right]\right\rVert_\infty
    \leq\sqrt{\frac{1-F_{\hat\phi}}{d_A}}
\end{equation}
Putting together,
\begin{equation}
    \left\lVert
        \Tr_{(BE)^k}
        \left[
            \Pi_{\hat\phi,j}^{(k)}
            \hat\psi^{\otimes k}
        \right]
    \right\rVert_\infty
    \leq
    \binom{k}{j}
    d_A^{-k+j/2}
    F_{\hat\phi}^{(k-j)/2}
    (1-F_{\hat\phi})^{j/2}.
\end{equation}
Substitution into the definition of \(L_j\) yields
\begin{equation}
\begin{aligned}
    \lVert L_j\rVert_\infty
    &\leq
    \frac{\binom{k}{j}d_A^{j/2}}{a_{nk}}
    \int_{\mathcal P_{ABE}(\tau_A)}
        F_{\hat\phi}^{\,n-(k+j)/2}
        (1-F_{\hat\phi})^{j/2}
        \,\mathrm d\hat\phi
    \\
    &\leq
    \frac{\binom{k}{j}d_A^{j/2}}{a_{nk}}
    \left(
        \int_{\mathcal P_{ABE}(\tau_A)}
            F_{\hat\phi}^{\,n-k}
            \,\mathrm d\hat\phi
    \right)^{1/2}
    \left(
        \int_{\mathcal P_{ABE}(\tau_A)}
            F_{\hat\phi}^{\,n-j}
            (1-F_{\hat\phi})^j
            \,\mathrm d\hat\phi
    \right)^{1/2},
\end{aligned}
\end{equation}
where the second inequality is Cauchy--Schwarz. The first integral
is by definition \(a_{nk}\), proving claimed bound.
\end{proof}

\subsection{Proof of \zcref{eq:cross-term-domination}}
\label{App:Proof:cross-term-domination}
\zcref{lem:constrained-haar-identity} states that
\begin{equation}
    \int_{\mathcal P_{ABE}(\tau_A)}
        \hat\phi^{\otimes k}\,\mathrm d\hat\phi
    =
    d_A^{-k}
    \bigl(K_{A^k}\otimes\id_{(BE)^k}\bigr)
    P_{\mathrm{sym}}^{(k)}.
\end{equation}
It follows that every positive operator \(R\) supported on
\(\operatorname{Sym}^k(ABE)\) satisfies
\begin{equation}
\label{eq:R-domination}
\begin{aligned}
R
&=
\bigl(K_{A^k}^{1/2}\otimes\id_{(BE)^k}\bigr)
\bigl(K_{A^k}^{-1/2}\otimes\id_{(BE)^k}\bigr)
R
\bigl(K_{A^k}^{-1/2}\otimes\id_{(BE)^k}\bigr)
\bigl(K_{A^k}^{1/2}\otimes\id_{(BE)^k}\bigr) \\
&\leq
\bigl(K_{A^k}^{1/2}\otimes\id_{(BE)^k}\bigr)
\left\lVert \bigl(K_{A^k}^{-1/2}\otimes\id_{(BE)^k}\bigr)
R
\bigl(K_{A^k}^{-1/2}\otimes\id_{(BE)^k}\bigr) \right\rVert_{\infty}P_{\mathrm{sym}}^{(k)}
\bigl(K_{A^k}^{1/2}\otimes\id_{(BE)^k}\bigr) \\
&\leq
\Tr\!\left[
    \bigl(K_{A^k}^{-1}\otimes\id_{(BE)^k}\bigr)R
\right]
\bigl(K_{A^k}\otimes\id_{(BE)^k}\bigr)
    P_{\mathrm{sym}}^{(k)}\\
&=d_A^k \Tr\!\left[
    \bigl(K_{A^k}^{-1}\otimes\id_{(BE)^k}\bigr)R
\right] \int_{\mathcal P_{ABE}(\tau_A)}
        \hat\phi^{\otimes k}\,\mathrm d\hat\phi,
\end{aligned}
\end{equation}
where we have repeatedly used the fact that \(K_{A^k}\otimes\id_{(BE)^k}\) commutes with \(P_{\mathrm{sym}}^{(k)}\), and that any Hermitian operator supported on \(\operatorname{Sym}^k(ABE)\) can be upper bounded by its operator norm times the symmetric subspace projector, and that the operator norm of a positive operator is upper bounded by its trace.

The operators \(L_j\), and hence \(L_{>r}\), commute with permutations
on \(A^k\). To see this, recall that in the definition \zcref{eq:Lj-definition}, both \(\Pi_{\hat\phi,j}^{(k)}\) and \(\hat\psi^{\otimes k}\) are invariant under permutations of \((ABE)^k\). Hence, for every
\(\pi\in S_k\),
\begin{equation}
\begin{aligned}
U_{A^k}^{\pi}L_j\bigl(U_{A^k}^{\pi}\bigr)^\dagger
&=
\frac{d_A^k}{a_{nk}}
\int_{\mathcal P_{ABE}(\tau_A)}
F_{\hat\phi}^{\,n-k}
\Tr_{(BE)^k}\!\left[
    U_{(ABE)^k}^{\pi}
    \Pi_{\hat\phi,j}^{(k)}
    \hat\psi^{\otimes k}
    \bigl(U_{(ABE)^k}^{\pi}\bigr)^\dagger
\right]
\mathrm d\hat\phi
\\
&=
\frac{d_A^k}{a_{nk}}
\int_{\mathcal P_{ABE}(\tau_A)}
F_{\hat\phi}^{\,n-k}
\Tr_{(BE)^k}\!\left[
    \Pi_{\hat\phi,j}^{(k)}
    \hat\psi^{\otimes k}
\right]
\mathrm d\hat\phi
=L_j.
\end{aligned}
\end{equation}

Notice that positivity of
\begin{equation}
    \left(
        \sqrt{\|L_{>r}\|_\infty}\,\id_{(ABE)^k}
        -
        \frac{L_{>r}\otimes\id_{(BE)^k}}
             {\sqrt{\|L_{>r}\|_\infty}}
    \right)
    \Psi^{(k)}
    \left(
        \sqrt{\|L_{>r}\|_\infty}\,\id_{(ABE)^k}
        -
        \frac{L_{>r}\otimes\id_{(BE)^k}}
             {\sqrt{\|L_{>r}\|_\infty}}
    \right)
\end{equation}
implies
\begin{equation}
\begin{aligned}
    &(L_{>r}\otimes\id_{(BE)^k})\Psi^{(k)}
    +
    \Psi^{(k)}(L_{>r}\otimes\id_{(BE)^k})
    \\
    &\quad\leq
    \|L_{>r}\|_\infty\Psi^{(k)}
    +
    \frac{1}{\|L_{>r}\|_\infty}
    (L_{>r}\otimes\id_{(BE)^k})
    \Psi^{(k)}
    (L_{>r}\otimes\id_{(BE)^k}).
\end{aligned}
\end{equation}
The operator on the right is positive and supported on
\(\operatorname{Sym}^k(ABE)\) due to the above commutation. Take \(R\) to be the right-hand side operator and apply \zcref{eq:R-domination},
\begin{equation}
\begin{aligned}
&\quad d_A^k \Tr\!\left[
    \bigl(K_{A^k}^{-1}\otimes\id_{(BE)^k}\bigr)R
\right]\\
    &=d_A^k
    \Tr\!\left[
        \bigl(K_{A^k}^{-1}\otimes\id_{(BE)^k}\bigr)
        \left(
            \|L_{>r}\|_\infty\Psi^{(k)}
            +
            \frac{
                (L_{>r}\otimes\id_{(BE)^k})
                \Psi^{(k)}
                (L_{>r}\otimes\id_{(BE)^k})
            }{\|L_{>r}\|_\infty}
        \right)
    \right]
    \\
    &=
    \|L_{>r}\|_\infty\Tr K_{A^k}^{-1}
    +
    \frac{
        \Tr\!\left[K_{A^k}^{-1}L_{>r}^2\right]
    }{\|L_{>r}\|_\infty}
    \leq
    2d_k\|L_{>r}\|_\infty,
\end{aligned}
\end{equation}
where the second equality follows from \(\Psi_{A^k}^{(k)}=\tau_A^{\otimes k}\), and the inequality follows from \(L_{>r}\leq\|L_{>r}\|_\infty\id_{A^k}\) and \(\Tr K_{A^k}^{-1}=\Tr P_{\mathrm{sym}}^{(k)}=d_k\). Hence,
\begin{equation}
    (L_{>r}\otimes\id_{(BE)^k})\Psi^{(k)}
    +
    \Psi^{(k)}(L_{>r}\otimes\id_{(BE)^k})
    \leq
    R
    \leq
    2d_k\|L_{>r}\|_\infty \int_{\mathcal P_{ABE}(\tau_A)}
        \hat\phi^{\otimes k}\,\mathrm d\hat\phi.
\end{equation}
Taking the partial trace over \(E^k\) therefore gives \zcref{eq:cross-term-domination}.

\subsection{Proof of \zcref{eq:lambda-bound}}
\label{App:Proof:lambda-bound}

All \(F_{\boldsymbol\alpha}\) below are understood to be tensored with
\(\id_{(BE)^k}\). The operator space \(\mathcal L((ABE)^k)\) is equipped with the usual Hilbert--Schmidt inner product:
\begin{equation}
    \langle X,Y\rangle_{\mathrm{HS}}
    \coloneq
    \Tr(X^\dagger Y),
    \qquad
    \|X\|_2^2
    \coloneq
    \langle X,X\rangle_{\mathrm{HS}}.
\end{equation}
Since \(\omega_{A^k}^{(k)}=\tau_A^{\otimes k}\), the definition of
\(h_{\boldsymbol\alpha}\) in \zcref{eq:def-halpha-lambda}, together with the operator identity we have used several times (e.g., \zcref{eq:op-identity}), gives
\begin{equation}
\begin{aligned}
d_A^k h_{\boldsymbol\alpha}
={}&
\Tr\!\left[
    F_{\boldsymbol\alpha}^{\dagger}
    \left(
        P_{\hat\phi}^{(k,r)}
        \omega^{(k)}
        P_{\hat\phi}^{(k,r)}
        -
        \omega^{(k)}
    \right)
\right]                                                               \\
={}&
-\Tr\!\left[
    F_{\boldsymbol\alpha}^{\dagger}
    \bigl(\id-P_{\hat\phi}^{(k,r)}\bigr)
    \omega^{(k)}
\right]                                                              
-
\Tr\!\left[
    F_{\boldsymbol\alpha}^{\dagger}
    \omega^{(k)}
    \bigl(\id-P_{\hat\phi}^{(k,r)}\bigr)
\right]                                                               \\
&+
\Tr\!\left[
    F_{\boldsymbol\alpha}^{\dagger}
    \bigl(\id-P_{\hat\phi}^{(k,r)}\bigr)
    \omega^{(k)}
    \bigl(\id-P_{\hat\phi}^{(k,r)}\bigr)
\right].
\end{aligned}
\end{equation}
The discarded mass can be written as
\begin{equation}
\delta_{kr}=\Tr\!\left[
    \bigl(\id-P_{\hat\phi}^{(k,r)}\bigr)
    \omega^{(k)}
\right]=\Tr\!\left[
    \sqrt{\omega^{(k)}}
    \bigl(\id-P_{\hat\phi}^{(k,r)}\bigr)
    \sqrt{\omega^{(k)}}
\right]=\left\|
    \bigl(\id-P_{\hat\phi}^{(k,r)}\bigr)
    \sqrt{\omega^{(k)}}
\right\|_2^2.
\end{equation}
The first two terms in the expansion of
\(d_A^kh_{\boldsymbol\alpha}\) are respectively
\begin{equation}
\begin{aligned}
&\Tr\!\left[
    F_{\boldsymbol\alpha}^{\dagger}
    \bigl(\id-P_{\hat\phi}^{(k,r)}\bigr)
    \omega^{(k)}
\right]                                                               
=
\left\langle
    F_{\boldsymbol\alpha}\sqrt{\omega^{(k)}},
    \bigl(\id-P_{\hat\phi}^{(k,r)}\bigr)
    \sqrt{\omega^{(k)}}
\right\rangle_{\mathrm{HS}},
                                                                        \\
&\Tr\!\left[
    F_{\boldsymbol\alpha}^{\dagger}
    \omega^{(k)}
    \bigl(\id-P_{\hat\phi}^{(k,r)}\bigr)
\right]                                                              
 =
\left\langle
    \bigl(\id-P_{\hat\phi}^{(k,r)}\bigr)
    \sqrt{\omega^{(k)}},
    F_{\boldsymbol\alpha}^{\dagger}\sqrt{\omega^{(k)}}
\right\rangle_{\mathrm{HS}}.
\end{aligned}
\end{equation}
Using again the Choi constraint \(\omega_{A^k}^{(k)}=\tau_A^{\otimes k}\),
\begin{equation}
\begin{aligned}
\left\langle
    F_{\boldsymbol\alpha}\sqrt{\omega^{(k)}},
    F_{\boldsymbol\beta}\sqrt{\omega^{(k)}}
\right\rangle_{\mathrm{HS}}
&=
\Tr\!\left[
    \omega^{(k)}
    F_{\boldsymbol\alpha}^{\dagger}F_{\boldsymbol\beta}
\right]                                                               
=
\Tr\!\left[
    \tau_A^{\otimes k}
    F_{\boldsymbol\alpha}^{\dagger}F_{\boldsymbol\beta}
\right]
=
\delta_{\boldsymbol\alpha,\boldsymbol\beta},
                                                                        \\
\left\langle
    F_{\boldsymbol\alpha}^{\dagger}\sqrt{\omega^{(k)}},
    F_{\boldsymbol\beta}^{\dagger}\sqrt{\omega^{(k)}}
\right\rangle_{\mathrm{HS}}
&=
\Tr\!\left[
    \omega^{(k)}
    F_{\boldsymbol\alpha}F_{\boldsymbol\beta}^{\dagger}
\right]                                                               
=
\Tr\!\left[
    \tau_A^{\otimes k}
    F_{\boldsymbol\alpha}F_{\boldsymbol\beta}^{\dagger}
\right]
=
\delta_{\boldsymbol\alpha,\boldsymbol\beta}.
\end{aligned}
\end{equation}
Thus both families \(\{F_{\boldsymbol\alpha}\sqrt{\omega^{(k)}}\}_{\boldsymbol\alpha},\{F_{\boldsymbol\alpha}^{\dagger}\sqrt{\omega^{(k)}}\}_{\boldsymbol\alpha}\) are orthonormal sets of \(\mathcal L((ABE)^k)\) in the Hilbert--Schmidt sense, and the first two terms in the expansion of
\(d_A^kh_{\boldsymbol\alpha}\) can be seen as coefficients of the operator \(\bigl(\id-P_{\hat\phi}^{(k,r)}\bigr)\sqrt{\omega^{(k)}}\) in these two orthonormal families.

For any orthonormal family \(\{|e_i\rangle\}_i\) and any vector
\(|v\rangle\), the sum of the squared overlaps
\(\sum_i|\langle e_i|v\rangle|^2\) is at most \(\|v\|^2\).
Applying this to
\(\bigl(\id-P_{\hat\phi}^{(k,r)}\bigr)\sqrt{\omega^{(k)}}\)
therefore gives
\begin{equation}
\label{eq:first-two-terms}
\begin{aligned}
&\sum_{\substack{\boldsymbol\alpha:
    \operatorname{wt}(\boldsymbol\alpha)\leq\min\{2r,k\}}}
\left|
    \Tr\!\left[
        F_{\boldsymbol\alpha}^{\dagger}
        \bigl(\id-P_{\hat\phi}^{(k,r)}\bigr)
        \omega^{(k)}
    \right]
\right|^2
\leq \left\|
    \bigl(\id-P_{\hat\phi}^{(k,r)}\bigr)
    \sqrt{\omega^{(k)}}
\right\|_2^2=
\delta_{kr},
                                                                        \\
&\sum_{\substack{\boldsymbol\alpha:
    \operatorname{wt}(\boldsymbol\alpha)\leq\min\{2r,k\}}}
\left|
    \Tr\!\left[
        F_{\boldsymbol\alpha}^{\dagger}
        \omega^{(k)}
        \bigl(\id-P_{\hat\phi}^{(k,r)}\bigr)
    \right]
\right|^2
\leq \left\|
    \bigl(\id-P_{\hat\phi}^{(k,r)}\bigr)
    \sqrt{\omega^{(k)}}
\right\|_2^2=
\delta_{kr}.
\end{aligned}
\end{equation}
For the third term in the expansion, Hölder’s inequality implies
\begin{equation}
\label{eq:third-term}
\left|
\Tr\!\left[
    F_{\boldsymbol\alpha}^{\dagger}
    \bigl(\id-P_{\hat\phi}^{(k,r)}\bigr)
    \omega^{(k)}
    \bigl(\id-P_{\hat\phi}^{(k,r)}\bigr)
\right]
\right|                                                             
\leq
\Tr\!\left[
    \bigl(\id-P_{\hat\phi}^{(k,r)}\bigr)
    \omega^{(k)}
    \bigl(\id-P_{\hat\phi}^{(k,r)}\bigr)
\right]
=
\delta_{kr}.
\end{equation}

Finally, there are
\begin{equation}
    N_{kr}\coloneq\sum_{j=0}^{\min\{2r,k\}}
        \binom{k}{j}(d_A^2-1)^j
\end{equation}
terms in the sum with weight at most \(\min\{2r,k\}\).
Consequently, Cauchy--Schwarz inequality, together with \zcref{eq:first-two-terms} and \zcref{eq:third-term} implies
\begin{equation}
\lambda
=
\sum_{\substack{\boldsymbol\alpha:
    \operatorname{wt}(\boldsymbol\alpha)\leq\min\{2r,k\}}}
d_A^k|h_{\boldsymbol\alpha}|                                             
\leq
2\sqrt{
    \delta_{kr}
    N_{kr}
}
+
\delta_{kr}
N_{kr},
\end{equation}
proving \zcref{eq:lambda-bound}.

\subsection{Proof of \zcref{eq:eta-exponential-bound}}
\label{App:Proof:eta-exponential-bound}
By the triangle inequality and the bound on \(L_j\) from \zcref{lem:ank-Lj-bounds},
\begin{equation}
    \lVert L_{>r}\rVert_\infty
    \leq
    \sum_{j=r+1}^{k}\lVert L_j\rVert_\infty                                     
\leq
    \frac{1}{\sqrt{a_{nk}}}
    \sum_{j=r+1}^{k}
        \frac{\binom{k}{j}d_A^{j/2}}{\sqrt{\binom{n}{j}}}
        \left(
            \binom{n}{j}
            \int_{\mathcal P_{ABE}(\tau_A)}
                F_{\hat\phi}^{\,n-j}
                (1-F_{\hat\phi})^j
                \,\mathrm d\hat\phi
        \right)^{1/2}.
\end{equation}
Applying the Cauchy--Schwarz inequality to the sum over \(j\) gives
\begin{equation}
    \lVert L_{>r}\rVert_\infty
    \leq
    \frac{1}{\sqrt{a_{nk}}}
    \left(
        \sum_{j=r+1}^{k}
        \frac{\binom{k}{j}^2d_A^j}{\binom{n}{j}}
    \right)^{1/2}                                                        
    \left(
        \sum_{j=r+1}^{k}
        \binom{n}{j}
        \int_{\mathcal P_{ABE}(\tau_A)}
            F_{\hat\phi}^{\,n-j}
            (1-F_{\hat\phi})^j
            \,\mathrm d\hat\phi
    \right)^{1/2}.
\label{eq:L-tail-cauchy-schwarz}
\end{equation}
The second factor is at most one by relaxing the sum from \(\sum_{j=r+1}^{k}\) to \(\sum_{j=0}^{n}\) and by using the binomial theorem. Using also \(a_{nk}\geq d_{n-k}^{-1}\) from \zcref{lem:ank-Lj-bounds}, we therefore obtain
\begin{equation}
    \lVert L_{>r}\rVert_\infty
    \leq
    \left(
        d_{n-k}
        \sum_{j=r+1}^{k}
        \frac{\binom{k}{j}^2d_A^j}{\binom{n}{j}}
    \right)^{1/2}.
    \label{eq:L-tail-explicit-bound}
\end{equation}
Consequently,
\begin{equation}
    1-\eta_{nkr}
    =
    \frac{
        2d_k\lVert L_{>r}\rVert_\infty
    }{
        1+2d_k\lVert L_{>r}\rVert_\infty
    }                                                                  
    \leq
    2d_k
    \left(
        d_{n-k}
        \sum_{j=r+1}^{k}
        \frac{\binom{k}{j}^2d_A^j}{\binom{n}{j}}
    \right)^{1/2}.
\label{eq:eta-explicit-bound}
\end{equation}
For \(j\geq r+1\),
\begin{equation}
    \frac{\binom{k}{j}^2d_A^j}{\binom{n}{j}}
    =
    \binom{k}{j}d_A^j
    \frac{\binom{k}{j}}{\binom{n}{j}}                                      
    \leq
    \left(\frac{e k}{j}\right)^j
    d_A^j
    \left(\frac{k}{n}\right)^j                                             
    \leq
    \left(
        \frac{e d_Ak^2}{n(r+1)}
    \right)^j.
\end{equation}
Suppose first that \(e d_Ak^2\leq n(r+1)\). Then
\begin{equation}
    \sum_{j=r+1}^{k}
        \frac{\binom{k}{j}^2d_A^j}{\binom{n}{j}}
    \leq
    (k-r)
    \left(
        \frac{e d_Ak^2}{n(r+1)}
    \right)^{r+1}.
\end{equation}
Substituting this into the preceding estimate gives
\begin{equation}
    1-\eta_{nkr}
    \leq
    2d_k\sqrt{d_{n-k}(k-r)}
    \exp\!\left[
        -\frac{r+1}{2}
        \ln\!\left(
            \frac{n(r+1)}{e d_Ak^2}
        \right)
    \right].
    \label{eq:eta-exponential-bound-app}
\end{equation}
If \(e d_Ak^2>n(r+1)\) and \(r<k\), the right-hand side of
\zcref{eq:eta-exponential-bound-app} is at least \(1\), whereas
\(1-\eta_{nkr}\leq1\). If \(r=k\), then \(L_{>k}=0\), so that
\(\eta_{n,k,k}=1\), and both sides of
\zcref{eq:eta-exponential-bound-app} vanish. Therefore,
\zcref{eq:eta-exponential-bound-app} holds for every
\(0\leq r\leq k\leq n\), without any additional assumption, proving \zcref{eq:eta-exponential-bound}.

\subsection{Proof of \zcref{eq:Nkr-exponential-bound}}
\label{App:Proof:Nkr-exponential-bound}
Recall that
\begin{equation}
    N_{kr}
    =
    \sum_{j=0}^{\min\{2r,k\}}
        \binom{k}{j}(d_A^2-1)^j.
\end{equation}
If \(0<2r\leq k\), then
\begin{equation}
\label{eq:Nkr-binomial-bound}
    N_{kr}
    \leq
    d_A^{4r}\sum_{j=0}^{2r}\binom{k}{j}                                  
    \leq
    d_A^{4r}
    \left(\frac{e k}{2r}\right)^{2r}
    =
    \left(
        \frac{e k d_A^2}{2r}
    \right)^{2r}.
\end{equation}
If \(2r>k\), then
\begin{equation}
    N_{kr}
    =
    \sum_{j=0}^{k}\binom{k}{j}(d_A^2-1)^j
    =
    d_A^{2k}.
\end{equation}
Since \(r\leq k\), we have \(e k/(2r)\geq e/2>1\) and \(4r>2k\), this again implies \zcref{eq:Nkr-binomial-bound}.

Furthermore, since \(e k d_A^2/(2r)>1\) always holds, we have
\begin{equation}
\begin{aligned}
    2\sqrt{d_{n-k}N_{kr}}
    \left(\frac{k}{n}\right)^{(r+1)/2}
    &\leq
    2\sqrt{d_{n-k}}
    \left(
        \frac{e k d_A^2}{2r}
    \right)^{r+1}
    \left(\frac{k}{n}\right)^{(r+1)/2}                                  \\
    &=
    2\sqrt{d_{n-k}}
    \exp\!\left[
        -\frac{r+1}{2}
        \ln\!\left(
            \frac{4nr^2}{e^2d_A^4k^3}
        \right)
    \right]                                                             \\
    &\leq
    2\sqrt{d_{n-k}}
    \exp\!\left[
        -\frac{r+1}{2}
        \ln\!\left(
            \frac{nr^2}{2d_A^4k^3}
        \right)
    \right],
\end{aligned}
\label{eq:completion-term-exponential-bound}
\end{equation}
where the last inequality follows from \(4/e^2>1/2\), proving \zcref{eq:Nkr-exponential-bound}.
\end{document}